\documentclass[12pt,journal,draftcls,onecolumn]{IEEEtran}

\pdfoutput=1

\usepackage{amsmath,amssymb,mathrsfs}
\usepackage{epsfig,epsf,subfigure,graphicx,graphics}
\usepackage{url}

\newtheorem{lemma}{Lemma}[section]
\newtheorem{theorem}{Theorem}[section]

\newtheorem{corollary}{Corollary}[section]
\newtheorem{definition}{Definition}[section]

\newtheorem{claim}{Claim}[section]

\newcommand{\pmc}{\mathsf{Pmc}}

\newcommand{\cloud}{\mcal{C}}
\newcommand{\mincut}{\mathrm{C}}

\newcommand{\aaaa}{\mathrm{(a)}}
\newcommand{\bbbb}{\mathrm{(b)}}
\newcommand{\cccc}{\mathrm{(c)}}
\newcommand{\dddd}{\mathrm{(d)}}
\newcommand{\eeee}{\mathrm{(e)}}

\newcommand{\lp}{\left(}
\newcommand{\rp}{\right)}
\newcommand{\lb}{\left[}
\newcommand{\rb}{\right]}
\newcommand{\lbp}{\left\{}
\newcommand{\rbp}{\right\}}

\newcommand{\ol}{\overline}
\newcommand{\mcal}{\mathcal}

\newcommand{\mbb}{\mathbb}
\newcommand{\msf}{\mathsf}
\newcommand{\mfrak}{\mathfrak}
\newcommand{\ra}{\rightarrow}

\newcommand{\lra}{\leftrightarrow}

\IEEEoverridecommandlockouts

\allowdisplaybreaks

\title{Two Unicast Information Flows over Linear Deterministic Networks}
\author{
\authorblockN{I-Hsiang Wang, Sudeep U. Kamath, and David N. C. Tse}\\
\authorblockA{Wireless Foundations\\
University of California at Berkeley,\\
Berkeley, California 94720, USA\\
\textsf{\{ihsiang, sudeep, dtse\}@eecs.berkeley.edu}}
}

\begin{document}
\maketitle

\begin{abstract}
We investigate the two unicast flow problem over layered linear deterministic networks with arbitrary number of nodes. When the minimum cut value between each source-destination pair is constrained to be $1$, 
it is obvious that the triangular rate region $\{(R_1,R_2):R_1,R_2\ge 0, R_1+R_2\le 1\}$ can be achieved, and that one cannot achieve beyond the square rate region $\{(R_1,R_2):R_1,R_2\ge 0, R_1\le1,R_2\le1\}$. Analogous to the work by Wang and Shroff for \emph{wired} networks \cite{WangShroff_10}, we provide the necessary and sufficient conditions for the capacity region to be the triangular region and the necessary and sufficient conditions for it to be the square region. Moreover, we completely characterize the capacity region and conclude that there are exactly three more possible capacity regions of this class of networks, in contrast to the result in wired networks where only the triangular and square rate regions are possible. Our achievability scheme is based on linear coding over an extension field with at most four nodes performing special linear coding operations, namely interference neutralization and zero forcing, while all other nodes perform random linear coding.

\end{abstract}

\section{Introduction}
Characterizing the fundamental limit of delivering information from multiple sources to multiple destinations over networks is the holy grail in network information theory. The ultimate goal is to characterize the capacity region of multi-source-multi-destination information flows over arbitrary networks.
Exploring \emph{wired} network models yields fruitful understanding in this problem, and the capacity of single unicast \cite{FordFulkerson_56} and multicast \cite{AhlswedeCai_00}
are fully characterized. 
In wired networks, however, all links are orthogonal to one another, and such a model cannot fully capture the \emph{broadcast} and \emph{superposition} nature of wireless networks.
In \cite{AvestimehrDiggavi_09}, a deterministic approach is proposed as a bridge for using results in wired networks to help understand wireless network information flow. The proposed \emph{linear deterministic network} model turns out to be very useful for studying wireless networks as it preserves the broadcast and superposition aspects. Capacity of several traffic patterns are characterized completely in linear deterministic networks and approximately in Gaussian networks, including single unicast and multicast \cite{AvestimehrDiggavi_09}.

In the above mentioned problems where good understanding has been established, there is only one user's information flow in the network and no \emph{interference} from other users. However, as for how multiple information flows interact as they interfere with one another, very little is known.
To the best of our knowledge, even for the two unicast problem, there is no capacity results for general wired networks, let alone the general multi-source-multi-destination information flow problem.
Instead of attempting directly to characterize the capacity region for the general two unicast problem, in \cite{WangShroff_10} Wang and Shroff study the \emph{solvability} of two-unicast wired networks, or equivalently, the achievability of the $(1,1)$ rate pair, for two unicast flows over arbitrary wired networks with integer link capacities, to make progress in this problem. They provide the necessary and sufficient condition for achieving the $(1,1)$ rate pair. They show that a simple sum rate outer bound called the Network Sharing Bound \cite{YanYang_06} turns out to be tight for the $(1,1)$ point, i.e. if the integer-valued bound is strictly greater than $1$, then $(1,1)$ can be achieved. The result in \cite{WangShroff_10} can also be understood as characterizing the capacity region for the class of wired networks where the minimum cut value between each source-destination pair is constrained to be $1$. This is because that rate pairs outside the square rate region $\mfrak{S}:=\{ (R_1,R_2): R_1,R_2\ge 0, R_1\le1,R_2\le1\}$ cannot be achieved, while those in the triangular rate region $\mfrak{T}:=\{ (R_1,R_2): R_1,R_2\ge 0,R_1+R_2\le1\}$ can always be achieved by time-sharing and routing. The result in \cite{WangShroff_10} implies that once one can achieve beyond the triangular region $\mfrak{T}$, one can achieve the square region $\mfrak{S}$. Hence, there are only two possible capacity regions for this class of networks, the triangular region $\mfrak{T}$ and the square region $\mfrak{S}$. See Fig.~\ref{fig_CapRegions} for an illustration of these rate regions.

\begin{figure}[htbp]
{\center
\subfigure[Network Sharing Bound $=1$]{\includegraphics[width=2in]{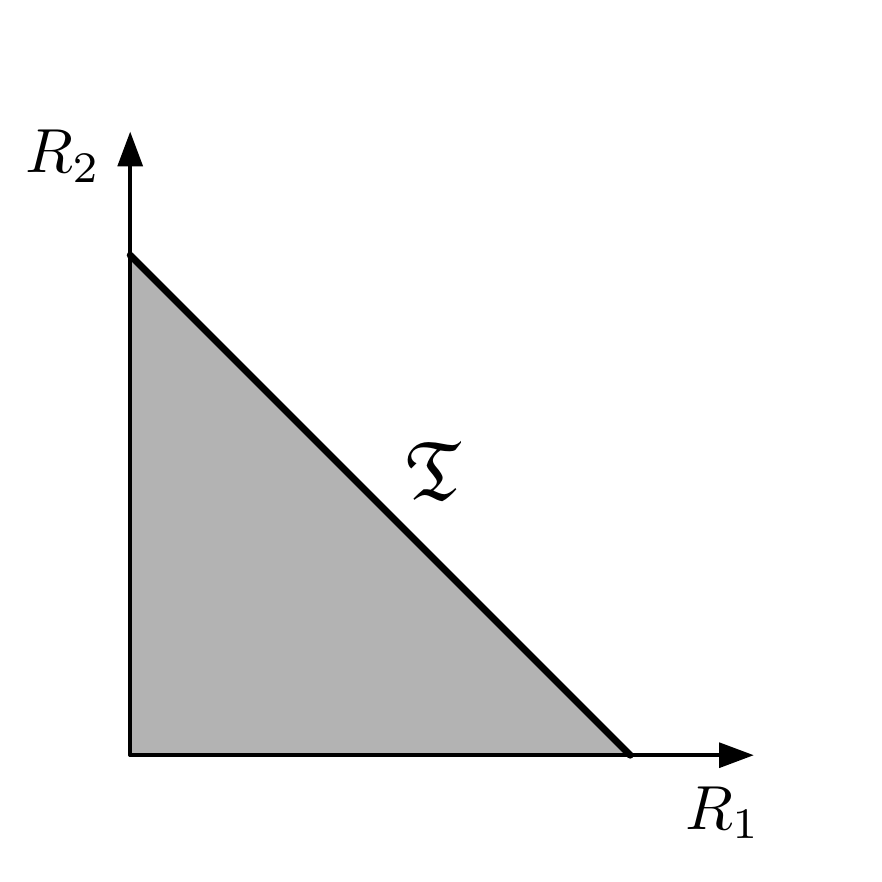}}
\subfigure[Network Sharing Bound $\ge2$]{\includegraphics[width=2in]{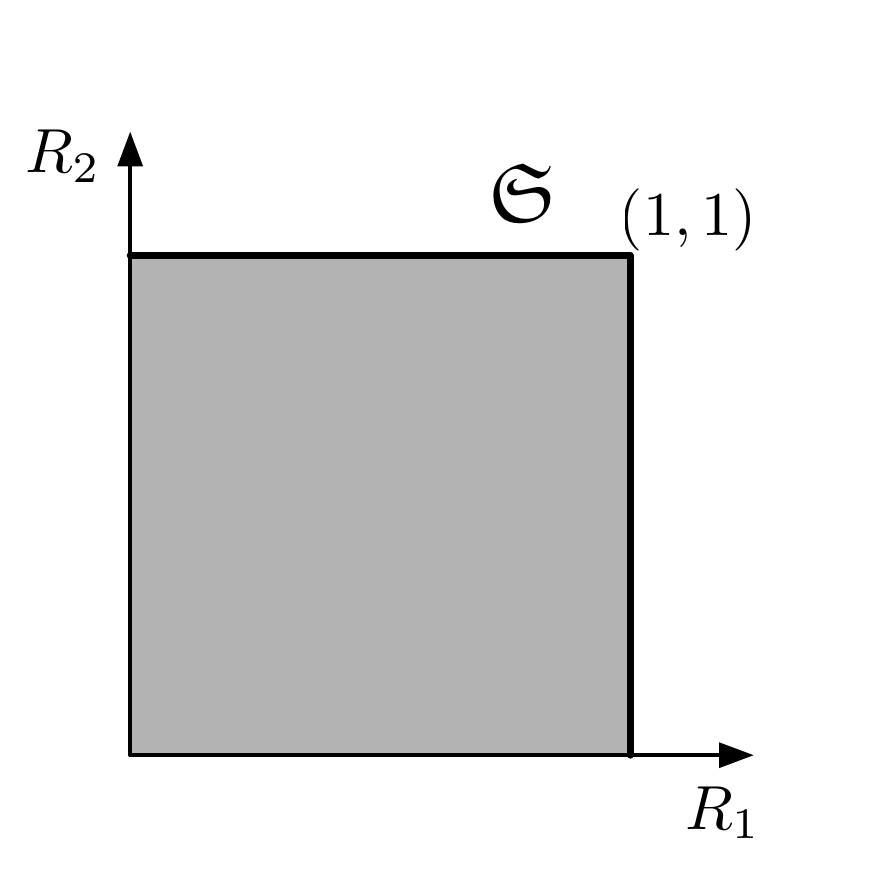}}
\caption{Capacity Regions for Wired Networks \cite{WangShroff_10}}
\label{fig_CapRegions}
}
\end{figure}

In this paper, we take an initial step towards understanding the two unicast flow problem over linear deterministic networks \cite{AvestimehrDiggavi_09} with arbitrary number of nodes. Our main result is an analog of \cite{WangShroff_10} over linear deterministic networks. We assume that all channel strengths are zero or unity, that the network is layered and that each source can reach its own destination, and hence the minimum cut value between each source-destination pair is constrained to be $1$. Similar to wired networks, rate pairs outside the square rate region $\mfrak{S}$ cannot be achieved, and rate pairs inside the triangular rate region $\mfrak{T}$ can be achieved by time-sharing between two users' single unicast flows.
For this class of networks, we completely characterize the capacity region. We show that the capacity region of such a network must be one of the five regions depicted in Fig.~\ref{fig_CapRegions_2}, and provide the necessary and sufficient conditions for the capacity region to be each of them. 

\begin{figure}[htbp]
{\center
\subfigure[$\textrm{O}$]{\includegraphics[width=2in]{CapRegions_T.pdf}}
\subfigure[$\textrm{T}^{(12)}\setminus \textrm{O}$]{\includegraphics[width=2in]{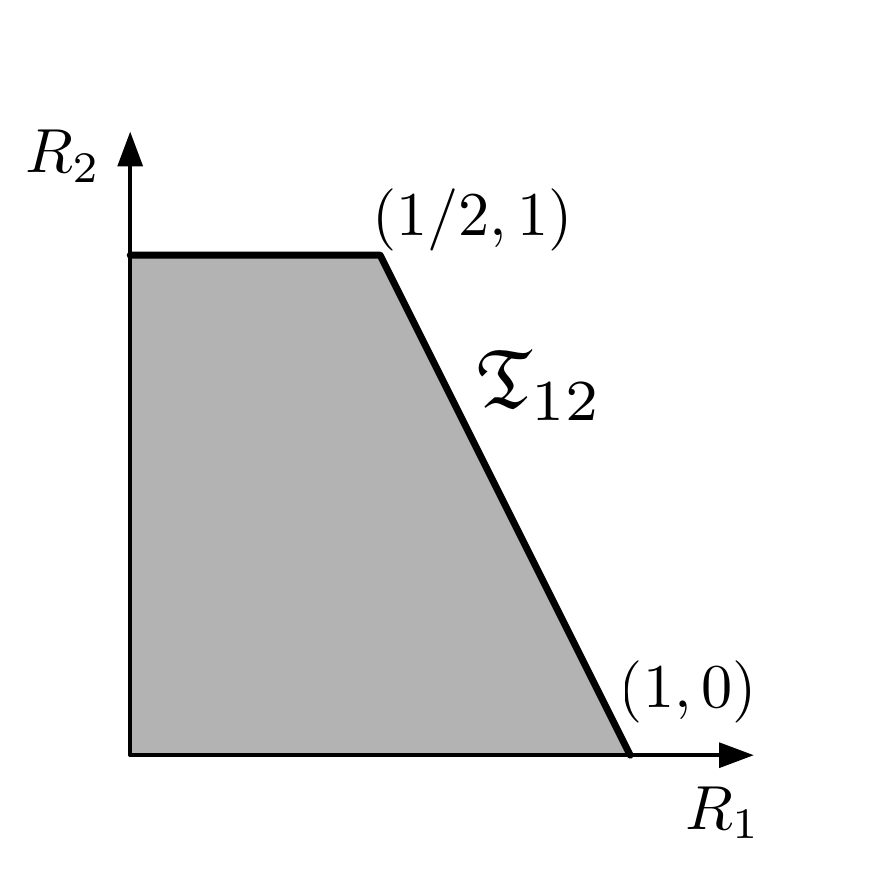}}
\subfigure[$\textrm{T}^{(21)}\setminus \textrm{O}$]{\includegraphics[width=2in]{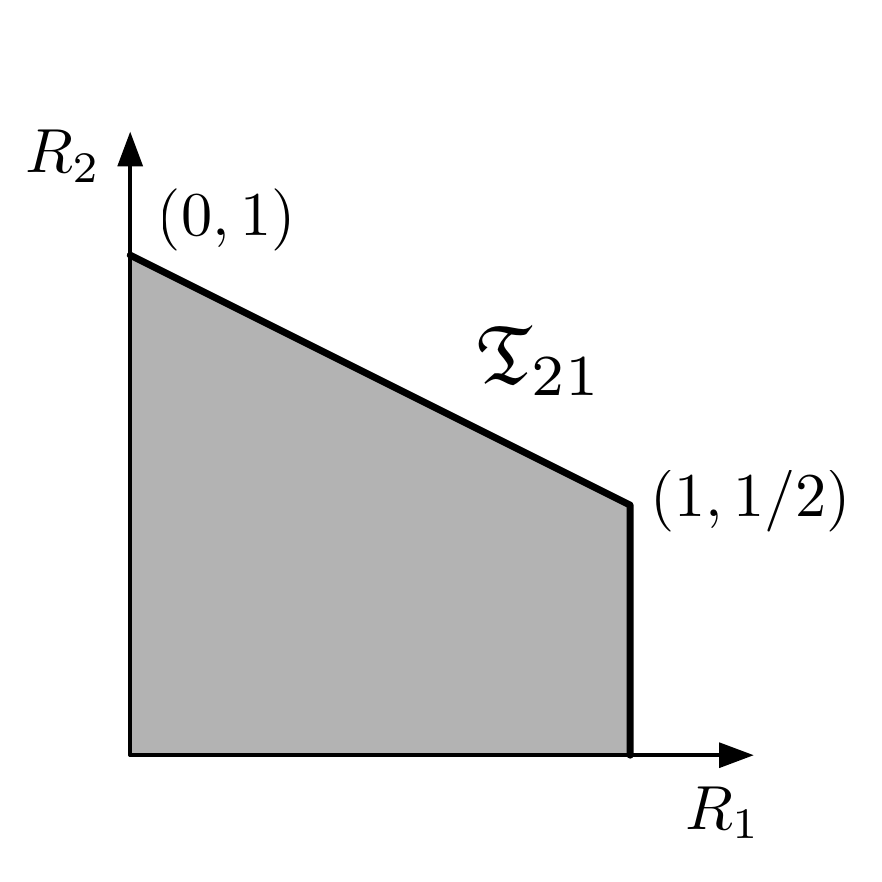}}
\subfigure[$\textrm{P}\setminus \lp \textrm{T}\cup \textrm{O}\rp$]{\includegraphics[width=2in]{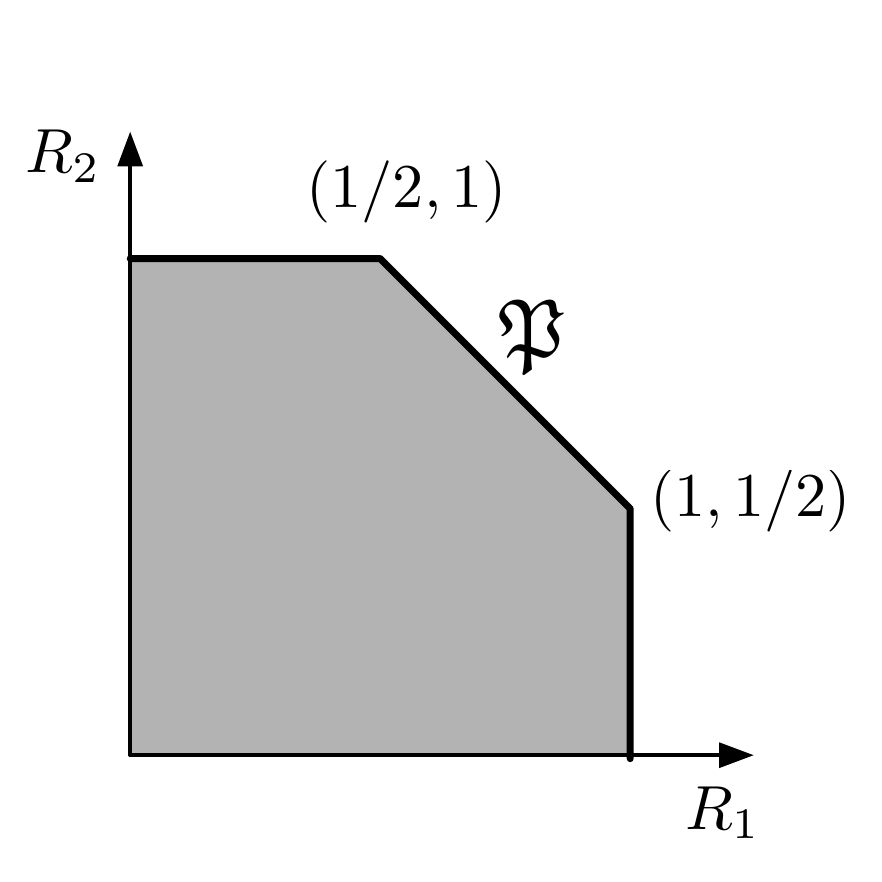}}
\subfigure[$\ol{\textrm{Q}}\setminus \textrm{O}$]{\includegraphics[width=2in]{CapRegions_S.pdf}}
\caption{Capacity Regions for Linear Deterministic Network}
\label{fig_CapRegions_2}
}
\end{figure}


Regarding when one can achieve beyond the trivially achievable $\mfrak{T}$, we provide a novel sum rate outer bound on two unicast flows over linear deterministic networks, analogous to the Network Sharing Bound. This outer bound is intimately related to the Generalized Network Sharing outer bound \cite{KamathTseAnantharam_11} for wired networks. We show that if this bound does not constrain the sum rate to be upper bounded by $1$, then indeed one can achieve \emph{beyond} the triangular rate region $\mfrak{T}$, and hence establish the necessary and sufficient condition for the capacity region being $\mfrak{T}$. In contrast however, to achievability of the $(1,1)$ point in \cite{WangShroff_10}, we find that we cannot always achieve $(1,1)$. Instead, we show that once one can achieve beyond $\mfrak{T}$, one can achieve either one of the two trapezoid rate regions: $\mfrak{T}_{12} := \{ (R_1,R_2): R_1,R_2\ge 0, R_2\le 1,2R_1+R_2\le 2\}$ and $\mfrak{T}_{21} := \{ (R_1,R_2): R_1,R_2\ge 0, R_1\le 1,R_1+2R_2\le 2\}$, and there are networks whose capacity regions are $\mfrak{T}_{12}$ or $\mfrak{T}_{21}$.

Regarding when one can achieve the full square $\mfrak{S}$, we investigate the achievability of the $(1,1)$ point, and find the necessary and sufficient conditions for it. For single source unicast and multicast problems, random linear coding over a large finite field at all nodes suffices to achieve capacity in wired as well as linear deterministic networks \cite{AhlswedeCai_00}, \cite{AvestimehrDiggavi_09}. This is no longer the case for the two-unicast problem since each destination is interested only in the message of its own source. Indeed, we can identify two nodes, one for each destination, that must be able to decode the messages of their respective destinations. We call these two nodes \emph{critical nodes} and their receptions are required to be completely free of \emph{interference} from the other user. 
For this purpose, at certain nodes interference from the other user has to be cancelled ``over-the-air", which is called \emph{interference neutralization} in the literature \cite{MohajerDiggavi_11} \cite{GouJafar_10}. Other than the nodes performing interference neutralization, all other nodes may perform random linear coding. The parents of each critical node are the natural candidates to perform interference neutralization, although they are not the only ones. 
We introduce a systematic approach to capture the effect on the rest of the network caused by interference neutralization, and provide the graph-theoretic necessary and sufficient conditions for $(1,1)$-achievability. 
Moreover, we show that if $(1,1)$ cannot be achieved, then the capacity region is contained in the pentagon region $\mfrak{P}:= \{ (R_1,R_2): R_1,R_2\ge 0, R_1,R_2\le 1, R_1+R_2\le 3/2\}$. Moreover, there are networks whose capacity regions are $\mfrak{P}$.


Continuing further, we characterize the necessary and sufficient conditions for the capacity region to be $\mfrak{T}_{12}$, $\mfrak{T}_{21}$, and $\mfrak{P}$ respectively. The outer bounds on $2R_1+R_2, R_1+2R_2$ for the trapezoids $\mfrak{T}_{12},\mfrak{T}_{21}$ respectively and that on $R_1+R_2$ for the pentagon $\mfrak{P}$ are inspired from the interference channel outer bounds \cite{El-GamalCosta_82}. 
The scheme we propose is linear over the \emph{extension field} $\mbb{F}_{2^r}$ for $r$ sufficiently large. Note that unlike single multicast where a \emph{random} (vector) linear scheme over the \emph{base field} $\mbb{F}_2$ suffices to achieve the capacity \cite{AvestimehrDiggavi_09}, in the two-unicast problem not only does the linear scheme operate on a larger field but also some nodes need to perform special linear coding (in contrast to random linear coding), including \emph{interference neutralization} (canceling interference over the air) and \emph{zero forcing} (canceling interference within a node). Later we will show by an example that both operating on a larger field and special coding at certain nodes are necessary for achieving capacity. It turns out that, fortunately, the number of nodes which are required to take special coding operation is bounded above by $4$ and can be found explicitly. More specifically, they are usually parents of the two critical nodes and hence lie in two layers at most. Other than these special nodes, others can perform \emph{random linear coding} (RLC) over the extension field.

\subsection*{Related Works}
In the literature, the study of two unicast information flows over wireless networks using the deterministic approach begins with the investigation of the two-user interference channel \cite{HanKobayashi_81} \cite{El-GamalCosta_82} \cite{EtkinTse_07} and its variants, including interference channels with cooperation \cite{PrabhakaranViswanathSRC_09} \cite{PrabhakaranViswanath_09} \cite{WangTse_09} \cite{WangTse_11} and two-hop interference networks \cite{MohajerDiggavi_11} \cite{GouJafar_10}. Focusing on small networks (four nodes in total), researchers are able to characterize the capacity region exactly in the linear deterministic case \cite{El-GamalCosta_82} \cite{WangTse_09} \cite{WangTse_11} and to within a bounded gap in the Gaussian case \cite{EtkinTse_07} \cite{WangTse_09} \cite{WangTse_11}, but the extension to larger networks seems non-trivial \cite{MohajerDiggavi_11}. The present work takes a step in that direction.

Another approach is directly looking at the Gaussian model but focusing on a cruder metric, degrees of freedom, instead of bounded gap to capacity. In \cite{GouJafar_10}, a systematic approach for interference neutralization called ``aligned interference neutralization" is proposed for the 2x2x2 interference network, and it is shown that full degrees of freedom (one for each user) can be achieved \emph{almost surely}. Later, in a recent independent work \cite{ShomoronyAvestimehr_11} such a scheme is employed and the authors characterize the degrees-of-freedom region of two unicast Gaussian networks almost surely. Interestingly, it is shown that \cite{ShomoronyAvestimehr_11} there are five possible degrees-of-freedom regions \emph{almost surely} and they are the same as the five regions reported in this paper. The connection between the two results is yet to be understood and explored. These degrees-of-freedom results, however, rely heavily either on the assumption that there is infinite channel diversity, or on the rationality/irrationality of the channel gains for the scheme to work.

The rest of the paper is organized as follows. In Section~\ref{sec_Def}, we formulate the problem and give several useful definitions. In Section~\ref{sec_Results}, we state our main results, and in Section~\ref{sec_Examples} we furnish examples that motivate linear scheme based on field extension and illustrate several important elements in achievability and outer bounds. Then we devote to details of achievability proof as well as outer bounds in Section~\ref{sec_Achieve} and \ref{sec_OutBd}, respectively. Finally, we conclude the paper by discussing possible extensions to more general linear deterministic networks in Section~\ref{sec_Conclude}.

\section{Problem Formulation}\label{sec_Def}
A two-source-two-destination layered network is 
a \emph{directed, acyclic, layered} graph $\mcal{G} = (\mcal{V}, \mcal{E})$, i.e. where the collection of nodes $\mcal{V}$ can be partitioned into $L+2$ \emph{layers} ($L\geq 0$): 
\begin{align*}
&\mcal{V} = \bigcup_{k=0}^{L+1} \mcal{L}_k,\ \mcal{L}_k \cap \mcal{L}_j \neq \emptyset,\ \forall k\ne j,
\end{align*}
such that for any edge $(\msf{u},\msf{v}) \in\mcal{E}$, $\exists\ k,\ 0\leq k\leq L\ \mathrm{s.t.}\ \msf{u}\in\mcal{L}_k, \msf{v}\in\mcal{L}_{k+1}.$ The first layer $\mcal{L}_0 = \{ \msf{s}_1, \msf{s}_2\}$ consists of the two source nodes, and the last layer $\mcal{L}_{L+1} = \{\msf{d}_1, \msf{d}_2\}$ consists of the two destination nodes. Without loss of generality we assume each node in the network can be reached by at least one of the source nodes and can reach at least one of the destination nodes.

For each node $\msf{v}\in\mcal{V}\setminus\{\msf{s}_1,\msf{s}_2\}$, we define nodes that can reach $\msf{v}$ as its \emph{predecessors}. Let $\mcal{P}(\msf{v})$ denote the set of predecessors that can reach $\msf{v}$ in one step. We will call the nodes in $\mcal{P}(\msf{v})$ as the \emph{parents} of $\msf{v}.$ Let $X_{\msf{u}}, Y_{\msf{u}}\in\mbb{F}_2$ denote the transmission and reception of node $\msf{u}$ respectively. The reception of a node is the binary XOR of the transmission of its parents:
$Y_{\msf{v}} = \bigoplus_{\msf{u}\in \mcal{P}(\msf{v})} X_{\msf{u}}$. For example, in Fig.~\ref{fig_examples}(a), the reception at node $\msf{u}_4$ will be given by $Y_{\msf{u}_4} = X_{\msf{u}_1}\oplus X_{\msf{u}_2}.$

\begin{figure}[htbp]
{\center
\subfigure[Zigzag Network]{\includegraphics[width=2.5in]{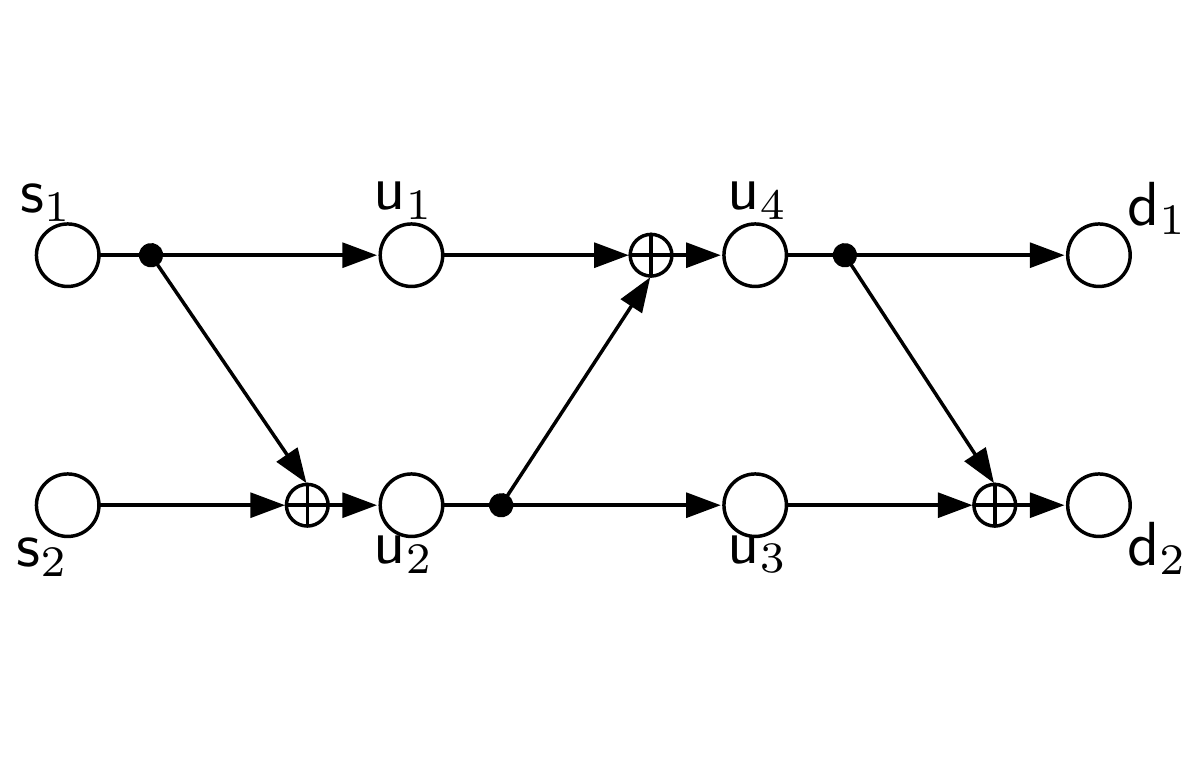}}
\subfigure[An Asymmetric Network]{\includegraphics[width=2.5in]{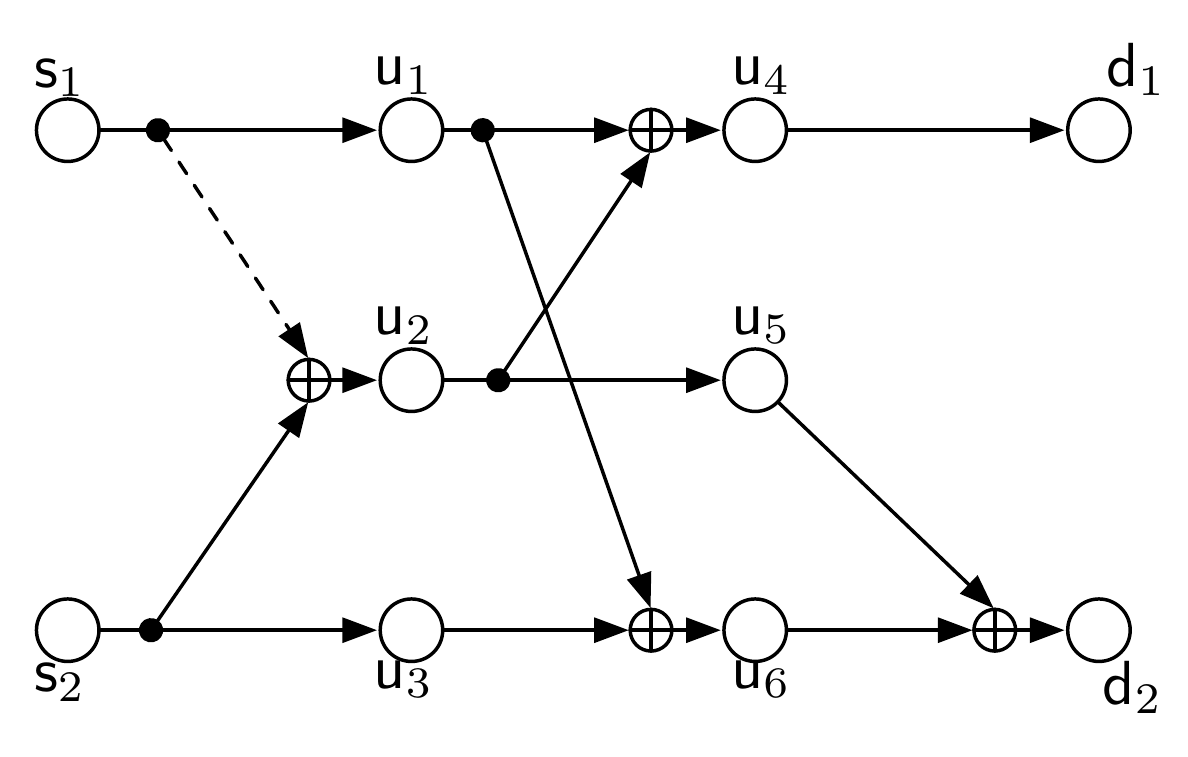}}
\caption{Examples}
\label{fig_examples}
}
\end{figure}

The channel model we have used is a special case of the linear deterministic network from \cite{AvestimehrDiggavi_09}. The simplification is that if there is a link from one node to another, then the channel strength is unity. We note that the essential nature of the linear deterministic network, namely broadcast and superposition, is preserved. As an example, in the network in Fig.~\ref{fig_examples}(a), the transmission of $\msf{u}_2$ is broadcasted to $\msf{u}_3$ and $\msf{u}_4$, and hence the two edges $(\msf{u}_2,\msf{u}_3)$ and $(\msf{u}_2,\msf{u}_4)$ carry the same signal. The reception of $\msf{u}_4$, as mentioned above, is the binary XOR of the transmission of $\msf{u}_1$ and $\msf{u}_2$.


\section{Main Result}\label{sec_Results}
If, for each $i=1,2,$ $\msf{s}_i$ can reach $\msf{d}_i,$ then it is trivial to see that the \emph{triangular} rate region $\mfrak{T}$
can be achieved, and that one cannot achieve beyond the square rate region $\mfrak{S}$. However, it is not clear under what conditions the triangular region or the square region is the capacity region. Our main result gives a complete answer to this question (and beyond). To state the result, we will need a few definitions.

A node is \emph{$\msf{s}_i$-reachable} if it can be reached by $\msf{s}_i$. It is \emph{$\msf{s}_i$-only-reachable} if it can be reached by $\msf{s}_i$ but not $\msf{s}_j$, $j\ne i$. It is \emph{$\msf{s}_1\msf{s}_2$-reachable} if it can be reached by both $\msf{s}_1$ and $\msf{s}_2$. 

For each node $\msf{v}\in\mcal{V}\setminus\{\msf{s}_1,\msf{s}_2\},$
\begin{itemize}
\item let $\mcal{P}(\msf{v})$ denote the set of parents of $\msf{v}$ that are reachable from at least one of $\msf{s}_1,\msf{s}_2,$
\item let $\mcal{P}^{\msf{s}_i}(\msf{v})\subseteq \mcal{P}(\msf{v})$ denote the set of parents of $\msf{v}$ \emph{reachable} by source $\msf{s}_i$, $i=1,2,$
\item let $\mcal{K}(\msf{v}):=\lbp \msf{u}: \mcal{P}(\msf{u}) = \mcal{P}(\msf{v})\rbp$ denote the \emph{clones} of $\msf{v},$ the set of nodes that receive the same signal as $\msf{v},$ 
\item let $\mcal{K}^{\msf{s}_i}(\msf{v}):=\lbp \msf{u}: \mcal{P}^{\msf{s}_i}(\msf{u}) = \mcal{P}^{\msf{s}_i}(\msf{v})\rbp, i=1,2,$ the set of nodes that have the same $\msf{s}_i$-reachable parents as $\msf{v}.$ We called these nodes the $\msf{s}_i$-clones of $\msf{v}$.
\end{itemize}

The following table illustrates these sets of nodes for the node $\msf{u}_4$ in the two example networks in Fig.~\ref{fig_examples}. For the network in (b), we assume for now that there is no edge from $\msf{s}_1$ to $\msf{u}_2$.
\begin{center}
{\small
\begin{tabular}{|c|c|c|c|c|c|c|c| }
  \hline
  & Fig.~\ref{fig_examples}(a) & Fig.~\ref{fig_examples}(b) \\ \hline
  $\mcal{P}(\msf{u}_4)$ & $\{\msf{u}_1,\msf{u}_2\}$ & $\{\msf{u}_1,\msf{u}_2\}$ \\ \hline
  $\mcal{P}^{\msf{s}_1}(\msf{u}_4)$ & $\{\msf{u}_1,\msf{u}_2\}$ & $\{\msf{u}_1\}$ \\ \hline
  $\mcal{P}^{\msf{s}_2}(\msf{u}_4)$ & $\{\msf{u}_2\}$ & $\{\msf{u}_2\}$ \\ \hline
  $\mcal{K}(\msf{u}_4)$ & $\{\msf{u}_4\}$ & $\{\msf{u}_4\}$ \\ \hline
  $\mcal{K}^{\msf{s}_1}(\msf{u}_4)$ & $\{\msf{u}_4\}$ & $\{\msf{u}_4,\msf{u}_6\}$ \\ \hline
  $\mcal{K}^{\msf{s}_2}(\msf{u}_4)$ & $\{\msf{u}_3,\msf{u}_4\}$ & $\{\msf{u}_4,\msf{u}_5\}$ \\ \hline
\end{tabular}
}
\end{center}

For two sets of nodes $\mcal{U}_1$ and $\mcal{U}_2$, we say a collection of nodes $\mcal{T}$ is a $(\mcal{U}_1;\mcal{U}_2)$-vertex-cut if in the graph obtained from the deletion of $\mcal{T},$ there are no paths from any node in $\mcal{U}_1\setminus \mcal{T}$ to any node in $\mcal{U}_2\setminus\mcal{T}.$ Note that this definition allows $\mcal{T}$ to have nodes from $\mcal{U}_1$ or $\mcal{U}_2.$

We say a node $\msf{v}\in\mcal{V}$ is \emph{omniscient} if it satisfies either of (A) or (B) below:
\begin{itemize}
{\rm
\item[(A)] $\mcal{K}(\msf{v})$ is a $\lp\msf{s}_1,\msf{s}_2; \msf{d}_1\rp$-vertex-cut and $\mcal{K}^{\msf{s}_2}(\msf{v})$ is a $\lp \msf{s}_2; \msf{d}_2\rp$-vertex-cut.
\item[(B)] $\mcal{K}(\msf{v})$ is a $\lp\msf{s}_1,\msf{s}_2; \msf{d}_2\rp$-vertex-cut and $\mcal{K}^{\msf{s}_1}(\msf{v})$ is a $\lp \msf{s}_1; \msf{d}_1\rp$-vertex-cut.
}
\end{itemize}


\begin{theorem}[Characterization of $\mfrak{T}$]\label{thm_1}
Assume that $\msf{s}_i$ can reach $\msf{d}_i$ for $i=1,2.$
{\flushleft \rm (a)}
If there exists an omniscient node in the network, then the capacity region is the triangular region $\mfrak{T}$.
{\flushleft \rm (b)}
Conversely, if no node in the network is omniscient, then the capacity region is strictly larger than $\mfrak{T}.$ Further, the capacity region contains at least one of the trapezoid regions $\mfrak{T}_{12}$ and $\mfrak{T}_{21}.$ In particular, $(2/3,2/3)$ is achievable and at least one of $(1/2,1)$ and $(1,1/2)$ is achievable.
\end{theorem}

It turns out that we are able to give the necessary and sufficient condition for the capacity region to be either $\mfrak{T}_{12}$ or $\mfrak{T}_{21}$. Before describing the theorem, we need some extra definitions. 

\begin{definition}[Critical Nodes]
For each $i=1,2$, we define the \emph{critical node} $\msf{v}_i^*$ as any node with the smallest possible layer index such that $\mcal{K}(\msf{v}_i^*)$ is a $\lp\msf{s}_1,\msf{s}_2; \msf{d}_i\rp$-vertex-cut.
\begin{itemize}
\item \emph{Existence}: $\{\msf{d}_i\}$ is a $\lp\msf{s}_1,\msf{s}_2; \msf{d}_i\rp$-vertex-cut.
\item \emph{Uniqueness up to clones}: if $\msf{u},\msf{w}$ are nodes in the same layer with $\mcal{K}(\msf{u})$ and $\mcal{K}(\msf{w})$ both being $\lp\msf{s}_1,\msf{s}_2; \msf{d}_i\rp$-vertex-cuts, then $\mcal{K}(\msf{u})=\mcal{K}(\msf{w}),$ i.e. $\msf{u}$ and $\msf{w}$ are clones.
\end{itemize}
We use $\mcal{L}_{k^*_i}$ to denote the layer where critical nodes $\msf{v}^*_i$ lies, for $i=1,2$. 
\end{definition}
For example in Fig.~\ref{fig_examples}, $\msf{v}_1^* = \msf{u}_4, k_1^*=2$ and $\msf{v}_2^* = \msf{d}_2, k_2^*=3$ for both networks.

The critical nodes defined here are directly analogous to the edges performing the ``reset'' operation in the add-up-and-reset construction of Wang and Shroff \cite{WangShroff_10}.

Below we describe one scenario where we get a result similar to the one in \cite{WangShroff_10}. This lemma strengthens part (b) of Theorem~\ref{thm_1} in this special scenario.
\begin{lemma}\label{k_1^*=0}
Suppose in a network $\msf{s}_2$ cannot reach $\msf{d}_1,$ i.e. $k_1^*=0.$ Then, the capacity region of this network is the triangle $\mfrak{T}$ or the square $\mfrak{S}$ depending on whether there is an omniscient node in the network or not, i.e. depending on whether $\msf{v}_2^*$ is omniscient or not (using Lemma~\ref{lem_critical_omniscient}). If $k_2^*=0,$ then $(1,1)$ can be achieved by all nodes performing random linear coding. If $k_2^*>0$ and there is no omniscient node, then $(1,1)$ is achieved with high probability when all nodes except nodes in $\mcal{P}(\msf{v}_2^*)$ performing random linear coding over a sufficiently large field.
\end{lemma}

Next we define cut values and min-cut on the network.
\begin{definition}[Cut Value and Min-Cut]
Fix a set of nodes in layer $k,$ $\mcal{U}\subseteq \mcal{L}_k.$ Consider a partition of $\mcal{V}$ into $(\mcal{T},\mcal{T}^c)$ with $\msf{s}_1,\msf{s}_2\in\mcal{T}$ and $\mcal{U}\subseteq \mcal{T}^c.$ Construct the transfer matrix $G$ with rows indexed by elements of $\mcal{T}$ and columns indexed by elements of $\mcal{T}^c$ where the $(\msf{u},\msf{w})$ entry of $G$ is 1 if there is a directed edge from $\msf{u}$ to $\msf{w}$ and 0 otherwise. The \emph{rank-mincut} \cite{AvestimehrDiggavi_09} from $\{\msf{s}_1,\msf{s}_2\}$ to $\mcal{U}$ is defined as the minimum rank of the transfer matrix $G$ over all such partitions $(\mcal{T},\mcal{T}^c)$, and is denoted by $\mincut\lp\msf{s}_1.\msf{s}_2;\mcal{U}\rp$.
\end{definition}

The following two lemmas provide some important properties of critical nodes. Their proofs are left in the appendix.
\begin{lemma}\label{lem_critical} 
For $i=1,2$, $\mincut\lp\msf{s}_1,\msf{s}_2;\mcal{P}(\msf{v}_i^*)\rp=2$ if $k_i^*\geq 2$.
\end{lemma}
\begin{lemma}\label{lem_critical_omniscient}
A network has an omniscient node if and only if one of the critical nodes $\msf{v}_1^*$ or $\msf{v}_2^*$ is omniscient.
\end{lemma}

Once we define the cut value, we can define \emph{primary min-cut nodes} for any set of nodes $\mcal{U}$ with $\mincut\lp \msf{s}_1,\msf{s}_2; \mcal{U}\rp=1$, due to the following lemma. What these primary min-cut nodes receive determines what $\mcal{U}$ receive.
\begin{lemma}[Primary Min-Cut]\label{lem_PrimaryMinCut}
By $\mcal{U}_l, 0\leq l < k,$ denote the set of nodes in layer $\mcal{L}_l$ that can reach some node in $\mcal{U}.$ Let $l^*$ be the minimum index such that $\mincut(\msf{s}_1,\msf{s}_2;\mcal{U}_{l^*}) =1$.  Then, $\mcal{U}_{l^*}\subseteq \mcal{K}(\msf{u})$ for any $\msf{u}\in\mcal{U}_{l^*},$ i.e. nodes in $\mcal{U}_{l^*}$ are all clones of each other.

We then define any of the nodes in $\mcal{K}(\msf{u})$ as the \emph{primary min-cut node} of $\mcal{U}$, denoted by $\pmc\lp \mcal{U}\rp$. It is unique up to clones.
\end{lemma}
\underline{Comment}:
Note that the reception of any node in $\mcal{U}$ is a function of the reception of $\pmc\lp \mcal{U}\rp$. 

For example, in Fig.~\ref{fig_examples}(b) when there is an edge from $\msf{s}_1$ to $\msf{u}_2$, $\pmc(\msf{u}_5) = \msf{u}_2$; when there is no edge from $\msf{s}_1$ to $\msf{u}_2$, $\pmc(\msf{u}_5) = \msf{s}_2$. We also see that the critical node $\msf{v}_i^* = \pmc(\msf{d}_i), i=1,2$.

Next, we define induced graph $\mcal{G}_{12}(\msf{w})$ for a node $\msf{w}\in\mcal{P}^{\msf{s}_2}(\msf{v}_1^*)$ as follows. The purpose of these induced graph is two-fold: 1) to capture the effect on the rest of the network caused by interference neutralization for $(1,1)$-achievability, and 2) to capture the Markov relations that are useful in the derivation of outer bounds.
\begin{definition}[Induced Graph $\mcal{G}_{12}$]
If $\mincut\lp \msf{s}_1,\msf{s}_2; \mcal{P}^{\msf{s}_2}(\msf{v}_1^*)\rp = 2$ then $\mcal{G}_{12}(\msf{w}):= \mcal{G}$. If $\mincut\lp \msf{s}_1,\msf{s}_2; \mcal{P}^{\msf{s}_2}(\msf{v}_1^*)\rp = 1$, then we define $\mcal{G}_{12}(\msf{w})$ as the graph obtained by modifying only the parents of nodes in $\mcal{L}_{k_1^*}$ as follows. For $\msf{u}\in\mcal{L}_{k_1^*},$
\begin{align*}
\mcal{P}_{\mcal{G}_{12}(\msf{w})}(\msf{u}) = \begin{cases}
  \mcal{P}(\msf{u}) & \text{if } \msf{w}\notin \mcal{P}(\msf{u}) \\
  \mcal{P}(\msf{u}) \Delta \mcal{P}^{\msf{s}_2}(\msf{v}_1^*) & \text{if } \msf{w}\in \mcal{P}(\msf{u}),
\end{cases}
\end{align*}
where $\Delta$ denotes symmetric set difference: $A\Delta B:=(A\setminus B)\cup(B\setminus A)$. We then drop nodes in $\mcal{G}_{12}(\msf{w})$ that cannot be reached by either of the two sources. In the rest of this paper, a graph theoretic object with a graph (say, $\mcal{G}_{12}$) in its subscript, like $\mcal{P}_{\mcal{G}_{12}(\msf{w})}(\msf{u})$ above, denote the graph theoretic object in the induced graph $\mcal{G}_{12}$. Define $\mcal{R}(\msf{w})$ as the set of nodes in $\mcal{P}^{\msf{s}_2}(\msf{v}_1^*)$ that can reach one of the two destinations in $\mcal{G}_{12}(\msf{w})$.
\end{definition}
Similarly we can define $\mcal{G}_{21}(\msf{w})$ with indices $1$ and $2$ swapped in the above definition.

\begin{figure}[htbp]
{\center
\subfigure[Zigzag Network]{\includegraphics[width=2.5in]{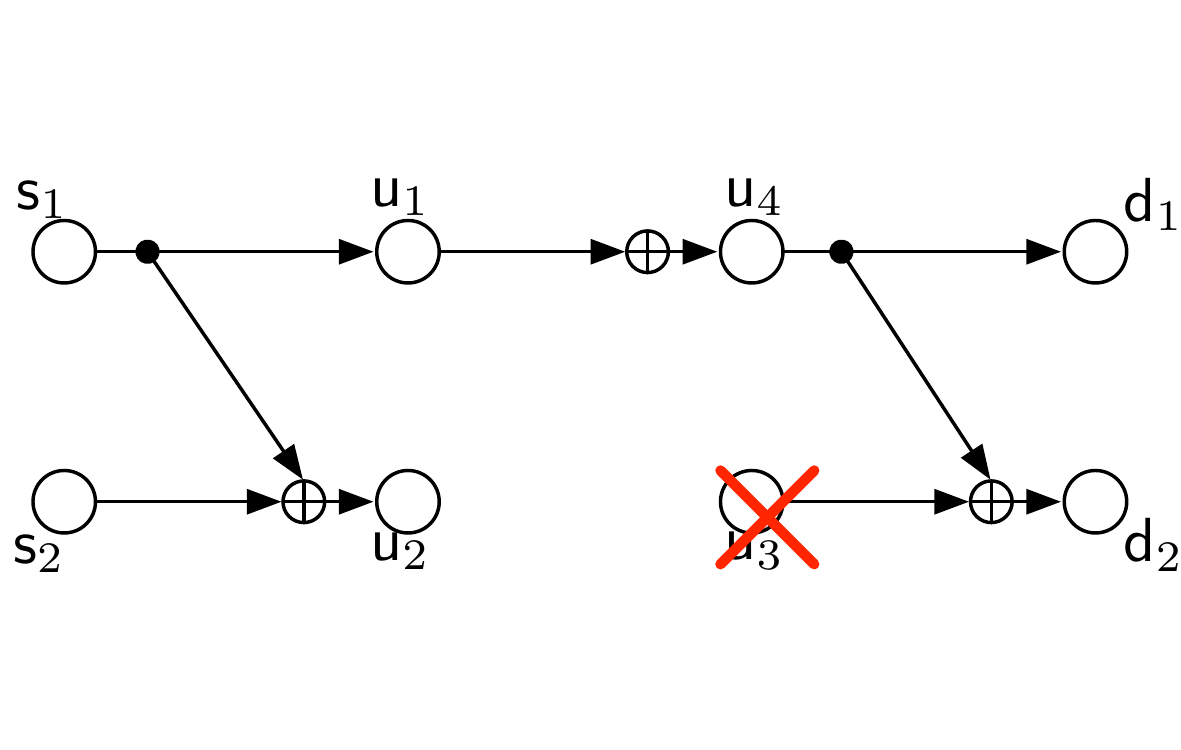}}
\subfigure[An Asymmetric Network]{\includegraphics[width=2.5in]{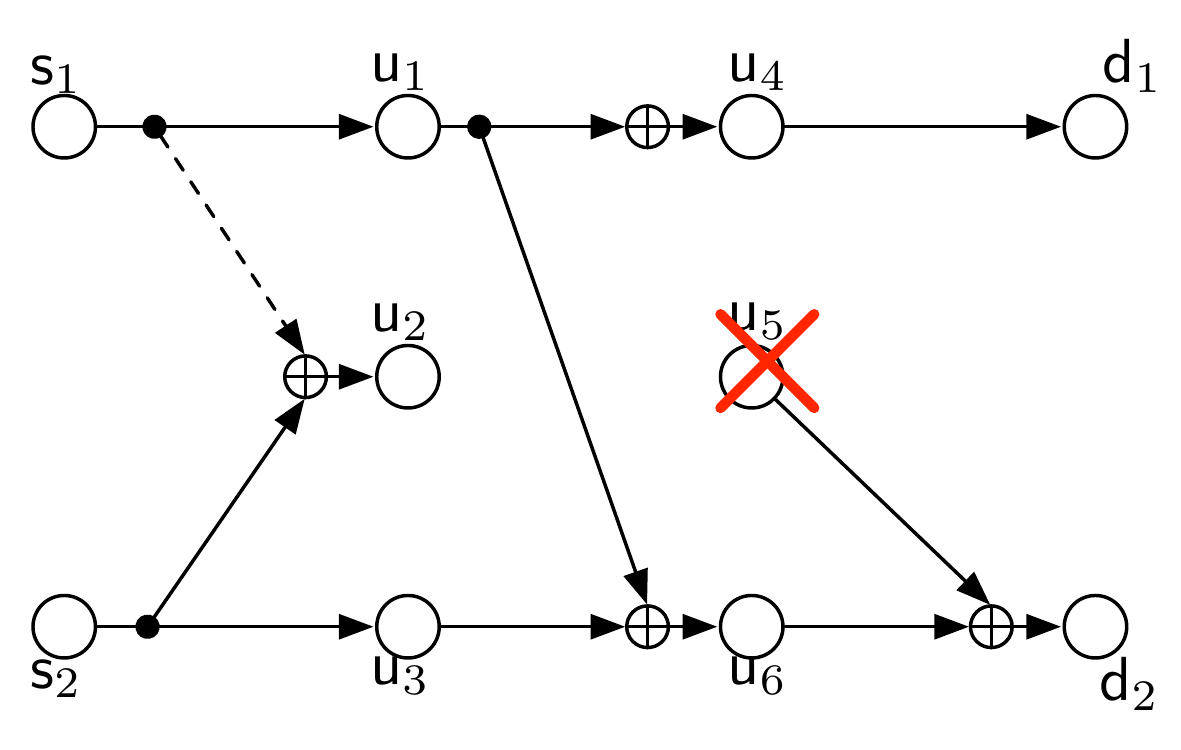}}
\caption{Induced Graph $\mcal{G}_{12}$ for Example Networks in Fig.~\ref{fig_examples}.}
\label{fig_examples_induced}
}
\end{figure}

For example, induced graphs for the networks in Fig.~\ref{fig_examples} are depicted in Fig.~\ref{fig_examples_induced}. For $\mcal{G}_{12}$ in (a), $\msf{s}_2$ can no longer reach $\msf{d}_2$, as $\msf{u}_4$ is omniscient in the original network $\mcal{G}$. In (b), node $\msf{u}_6$ becomes omniscient in $\mcal{G}_{12}$ while there is no omniscient node in the original network $\mcal{G}$.

We will use $\mcal{G}_{12}(\msf{w})$ when $k_1^*\le k_2^*$ and $\mcal{G}_{21}(\msf{w})$ when $k_2^* \le k_1^*$. We will only use these graphs in relation to whether or not there is an omniscient node in $\mcal{G}_{12}(\msf{w})$. Lemma~\ref{all_G12s_are_essentially_the_same} below allows us to drop the $\msf{w}$ and refer to any of the $\mcal{G}_{12}(\msf{w})$ as $\mcal{G}_{12}$ and talk about whether or not there is an omniscient node in $\mcal{G}_{12}$.



\begin{lemma}\label{all_G12s_are_essentially_the_same}
Suppose, in a network with no omniscient node, and with $k_1^*\le k_2^*,$ there exists a node $\msf{w}_0\in\mcal{P}^{\msf{s}_2}(\msf{v}_1^*)$ such that there is an omniscient node in $\mcal{G}_{12}(\msf{w}_0)$. Then for any node $\msf{w}\in\mcal{P}^{\msf{s}_2}(\msf{v}_1^*)$, there is an omniscient node in $\mcal{G}_{12}(\msf{w}).$
\end{lemma}


\begin{theorem}[Characterization of $\mfrak{T}_{12}$ and $\mfrak{T}_{21}$]\label{thm_2}
Consider a network $\mcal{G}$ in which no node is omniscient.
{\flushleft \rm (a)} 
If the network $\mcal{G}$ satisfies the following conditions, then the capacity region is the trapezoid region $\mfrak{T}_{12}$:
{\flushleft \rm
\begin{itemize}
\item $\textrm{T}^{(12)}_1$:
$0<k^*_1 \le k^*_2$.

\item $\textrm{T}^{(12)}_2$:
$\mincut\lp\msf{s}_1,\msf{s}_2;\mcal{P}^{\msf{s}_2}\lp \msf{v}^*_1\rp\rp = 1$. Let $\msf{w}_{12}$ denote $\pmc\lp \mcal{P}^{\msf{s}_2}\lp \msf{v}^*_1\rp\rp$. 

\item $\textrm{T}^{(12)}_3$:
Let $\msf{u}_{21}:=\pmc_{\mcal{G}_{12}}\lp \msf{v}^*_2 \rp$. $\msf{u}_{21}$ is omniscient in $\mcal{G}_{12}$.

\item $\textrm{T}^{(12)}_4$:
$\msf{w}_{12} = \msf{s}_2$, i.e., $\mcal{P}^{\msf{s}_2}\lp \msf{v}^*_1\rp$ cannot be reached by $\msf{s}_1$.

\end{itemize}
}
We call the conjunction of the above conditions $\textrm{T}^{(12)}$. Symmetrically, if $\mcal{G}$ satisfies the above condition with indices $1$ and $2$ (in the superscript) exchanged, then the capacity region is the trapezoid region $\mfrak{T}_{21}$.
{\flushleft \rm (b)}
Conversely, if neither condition $\textrm{T}^{(12)}$ nor $\textrm{T}^{(21)}$ is satisfied, then the two trapezoid regions are strictly contained in the capacity region. Moreover, both $(1/2,1)$ and $(1,1/2)$ are achievable and hence the pentagon $\mfrak{P}$.
\end{theorem}
\underline{Remark}: 
Based on Lemma~\ref{k_1^*=0}, if $k_1^*=0$, then the capacity region of this network is the triangle $\mfrak{T}$ or the square $\mfrak{S}$ depending on whether there is an omniscient node in the network. This is why in $\textrm{T}^{(12)}_1$ we need to constrain $k_1^* > 0$.



Next we give the necessary and sufficient condition for the capacity region being the pentagon region $\mfrak{P}:=\{ (R_1,R_2): R_1,R_2\ge 0, R_1\le1,R_2\le1, R_1+R_2\le 3/2\}$.
\begin{theorem}[Characterization of $\mfrak{P}$ and $\mfrak{S}$]\label{thm_3}
Consider a network $\mcal{G}$ in which no node is omniscient and neither $\textrm{T}^{(12)}$ nor $\textrm{T}^{(21)}$ is satisfied.
{\flushleft \rm (a)} 
Denote the conjunction of the below conditions by $\textrm{P}^{(12)}$:
{\flushleft \rm
\begin{itemize}
\item $\textrm{P}^{(12)}_1 \equiv \textrm{T}^{(12)}_1$, $\textrm{P}^{(12)}_2 \equiv \textrm{T}^{(12)}_2$, $\textrm{P}^{(12)}_3 \equiv \textrm{T}^{(12)}_3$



\item $\textrm{P}^{(12)}_4$:
$\msf{w}_{12} \ne \msf{s}_2$ and $\mcal{K}^{\msf{s}_2}\lp\msf{w}_{12}\rp$ forms an $\lp \msf{s}_2; \msf{d}_2\rp$-vertex-cut in $\mcal{G}$.


\end{itemize}
}
Similarly we define condition $\textrm{P}^{(21)}$ with indices $1$ and $2$ (in the superscript) exchanged. If the network $\mcal{G}$ satisfies condition $\textrm{P}^{(12)}$ or $\textrm{P}^{(21)}$, then the capacity region is $\mfrak{P}$. 
{\flushleft \rm (b)}
Conversely, if neither condition $\textrm{P}^{(12)}$ nor $\textrm{P}^{(21)}$ is satisfied, then the pentagon region is strictly contained in the capacity region. Moreover, $(1,1)$ is achievable and hence the square $\mfrak{S}$.
\end{theorem}

We can easily see that $\textrm{T}^{(12)}_4 \vee \textrm{P}^{(12)}_4 = \{ \text{$\mcal{K}^{\msf{s}_2}\lp\msf{w}_{12}\rp$ forms an $\lp \msf{s}_2; \msf{d}_2\rp$-vertex-cut in $\mcal{G}$.} \}$ and hence $\textrm{Q}^{(12)}:=\textrm{T}^{(12)}\vee \textrm{P}^{(12)}$ is the conjunction of the following:
\begin{itemize}
\item $\textrm{Q}^{(12)}_1 \equiv \textrm{T}^{(12)}_1$, $\textrm{Q}^{(12)}_2 \equiv \textrm{T}^{(12)}_2$, $\textrm{Q}^{(12)}_3 \equiv \textrm{T}^{(12)}_3$
\item $\textrm{Q}^{(12)}_4$:
$\mcal{K}^{\msf{s}_2}\lp\msf{w}_{12}\rp$ forms an $\lp \msf{s}_2; \msf{d}_2\rp$-vertex-cut in $\mcal{G}$.
\end{itemize}
\begin{corollary}[Complete Characterization of Capacity]
As a corollary of Theorem~\ref{thm_1}, \ref{thm_2} and \ref{thm_3}, we completely characterize all possible capacity regions of two unicast flows over the linear deterministic networks as formulated in Section~\ref{sec_Def}, as follows: (Short-hand notations: $\textrm{O} := \{\exists \text{ an omniscient node}\}$, $\textrm{T} := \textrm{T}^{(12)}\vee\textrm{T}^{(21)}$, $\textrm{P} := \textrm{P}^{(12)}\vee\textrm{P}^{(21)}$, and $\textrm{Q} := \textrm{Q}^{(12)}\vee\textrm{Q}^{(21)}=\textrm{T}\vee\textrm{P}$. Also, in the context that no confusion will be caused, we use the same notation to denote the set of networks that satisfy the condition.)
\begin{align*}
&\textrm{O} &\iff& & &\mfrak{T}\\
&\textrm{T}^{(12)}\setminus \textrm{O} &\iff& & &\mfrak{T}_{12}\\
&\textrm{T}^{(21)}\setminus \textrm{O} &\iff& & &\mfrak{T}_{21}\\
&\textrm{P}\setminus \lp \textrm{T}\cup \textrm{O}\rp &\iff& & &\mfrak{P}\\
&\ol{\textrm{Q}}\setminus \textrm{O} &\iff& & &\mfrak{S}
\end{align*}
Fig.~\ref{fig_CapRegions_2} give an illustration of all these regions.
\end{corollary}











\section{Motivating Examples}\label{sec_Examples}
Before going into proofs of our main result, let us visit some examples to illustrate several important elements in our scheme.

\subsection{Why Random Linear Coding Fails}
We first demonstrate, through a simple example, why random linear coding, while successful in achieving the capacity of single multicast over wired and linear deterministic networks \cite{AhlswedeCai_00} \cite{AvestimehrDiggavi_09}, cannot achieve capacity for multiple unicast. Also, by the example we will show that most of the nodes in the network can perform random linear coding and only up to four nodes are needed to do special linear coding. 

The example is depicted in Fig.~\ref{fig_Examples_1}(a). Random linear coding for achieving the $(1,1)$ point, in the context of this example, means that each node sends out a symbol in a large field of characteristic $2$ and each intermediate node scales its reception by a randomly uniformly chosen coefficient from the field, independent of others, and transmits it. How and why we lift the symbols from the base field $\mbb{F}_2$ to a larger field will be explained later. Random linear coding achieves the capacity of single multicast with high probability.

However, for two unicast if we perform random linear coding, in the network in Fig.~\ref{fig_Examples_1}(a), destinations $\msf{d}_1$ and $\msf{d}_2$ will receive linear combinations of the two symbols from sources, say $a$ from source 1 and $b$ from source 2, and their coefficients are non-zero with high probability. This is because both $\msf{d}_1$ and $\msf{d}_2$ can be reached by $\msf{s}_1$ and $\msf{s}_2$.

On the other hand, if nodes $\msf{u}_4,\msf{u}_5,\msf{u}_6$ choose their scaling coefficients more carefully, both $\msf{d}_1$ and $\msf{d}_2$ are able to receive a clean copy of their desired symbols. This is due to the fact that the reception of $\msf{u}_4$ (which is the same as that of $\msf{u}_6$) and the reception of $\msf{u}_5$ are linearly independent with high probability under random linear coding at all other nodes in previous layers, since $\mincut\lp\msf{s}_1,\msf{s}_2;\msf{u}_4,\msf{u}_5\rp = 2$. The scaling coefficients chosen at $\msf{u}_4$ is such that the $b$-component in the transmission is cancelled over-the-air. Because the reception of $\msf{u}_4$ and $\msf{u}_5$ are linearly independent, the $a$-coefficient remains non-zero. Similarly, $\msf{u}_6$ can choose its scaling coefficient so that $\msf{d}_2$ receives a non-zero scaled-copy of symbol $b$.

We observe that in this example only nodes $\msf{u}_4$ and $\msf{u}_6$ need to perform linear coding carefully. It turns out that for arbitrary layered networks, at most $4$ nodes need to perform special linear coding.

\begin{figure}[htbp]
{\center
\subfigure[Network]{\includegraphics[width=3in]{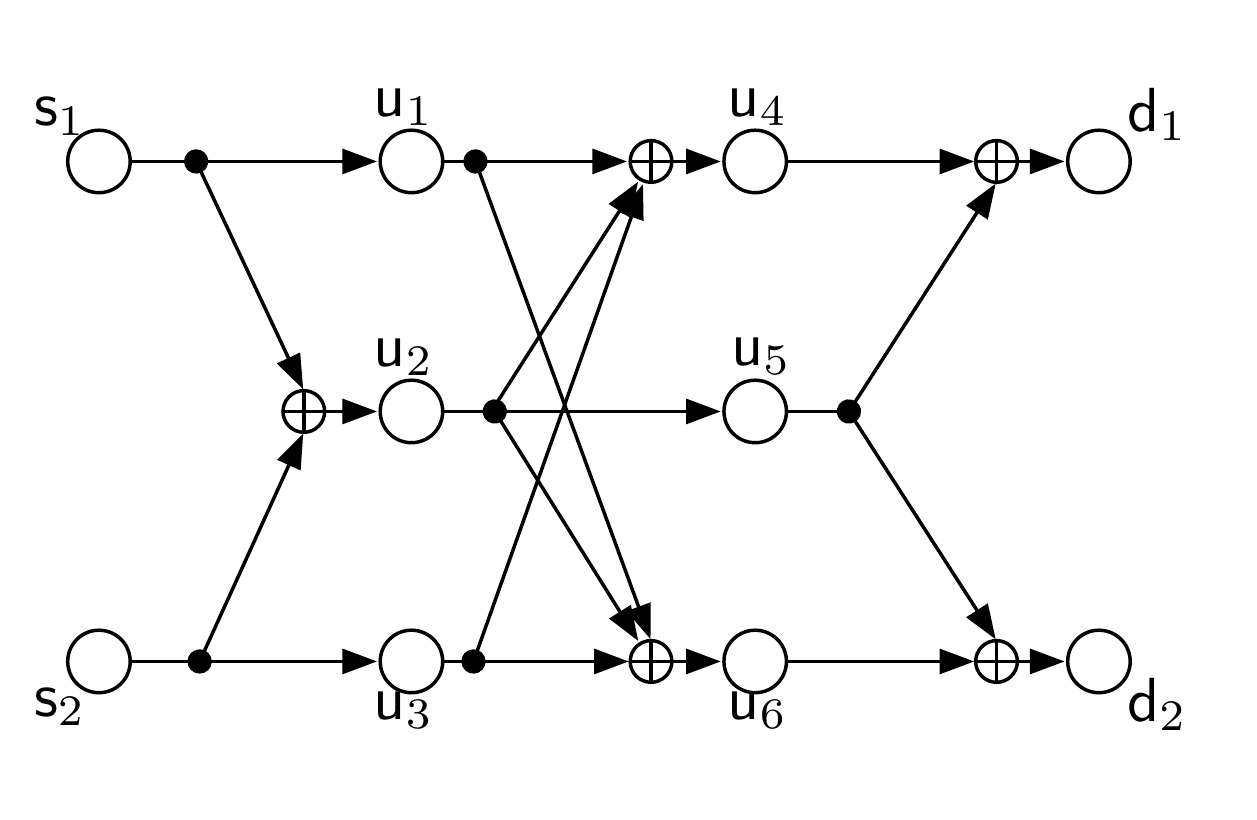}}
\subfigure[Linear Scheme over $\mbb{F}_4$ achieving $(1,1)$]{\includegraphics[width=3in]{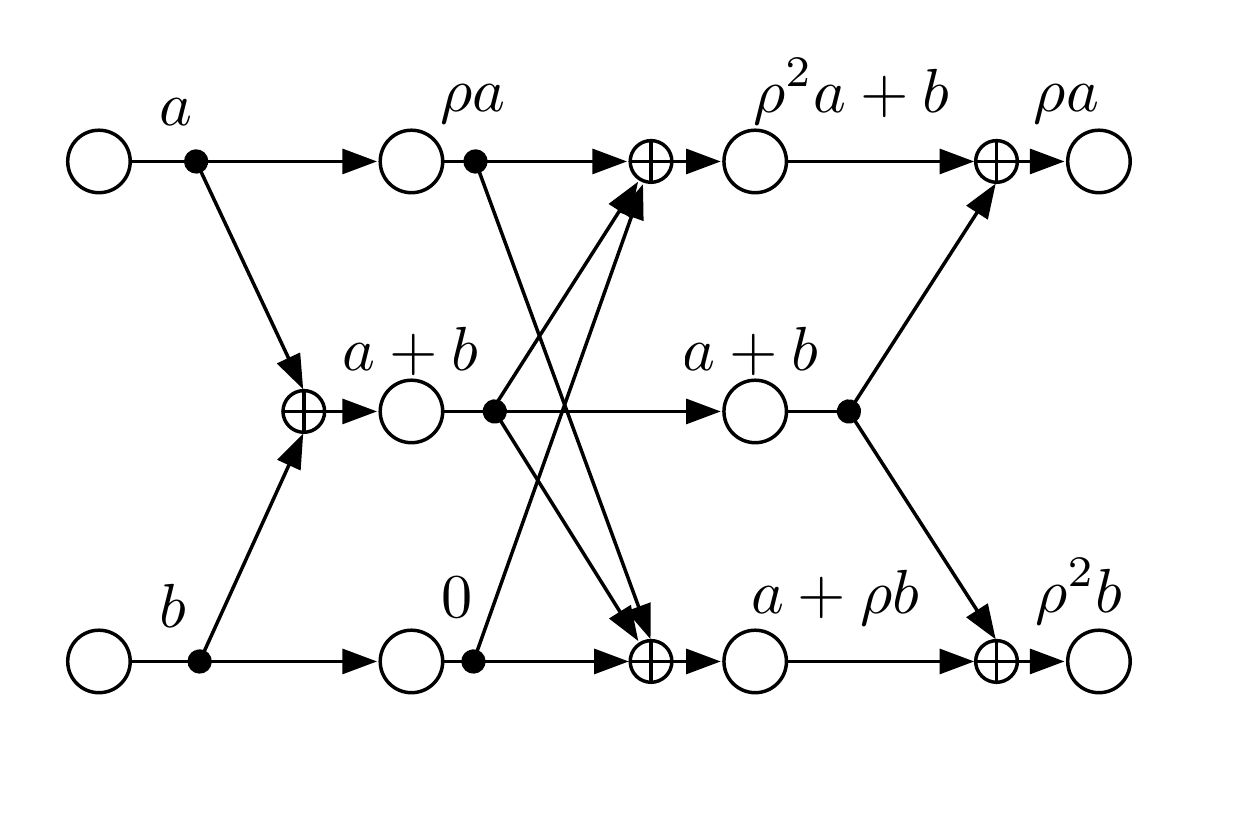}}
\caption{Examples}
\label{fig_Examples_1}
}
\end{figure}

\subsection{Why Field Extension is Necessary}
We give an example to illustrate the limitation if we do not use field extension and stick to vector linear scheme in $\mbb{F}_2$. The network is depicted in Fig.~\ref{fig_Examples_1}(a). Let the total number of channel uses be $T$, and source $\msf{s}_i$ would like to deliver $B_i$ bits to its own destination $\msf{d}_i$, $i=1,2$. We consider achieving beyond the triangular region $\mfrak{T}$, and hence assume $B_1+B_2 > T$. Therefore, at $\msf{u}_2$, at least $B_1+B_2-T$ bits from each source get corrupted, while $B_1 - (B_1+B_2-T) = T-B_2$ bits from $\msf{s}_1$ and $B_2 - (B_1+B_2-T) = T-B_1$ are clean. $\msf{u}_5$'s reception is just a function of what $\msf{u}_2$ receives, and hence it cannot obtain more information than what $\msf{u}_2$ possesses. In particular, $\msf{u}_5$ cannot obtain the two length-$(B_1+B_2-T)$ chunks of bits of user 1 and user 2 that get corrupted at $\msf{u}_2$. If $\msf{u}_5$ does not transmit this corrupted chunk, $\msf{u}_4$ needs to supply the clean chunk for user 1 to $\msf{d}_1$ and $\msf{u}_6$ needs to supply the clean chunk for user 2 to $\msf{d}_2$, respectively. But the reception of $\msf{u}_4$ and $\msf{u}_6$ is identical, and therefore both of them should be able to decode these two chunks. As their reception has at most $T$ bits, we have $2(B_1+B_2-T) \le T \implies 2B_1+2B_2 \le 3T$. If $\msf{u}_5$ transmits this corrupted chunk, still $\msf{u}_4$ needs to have the clean chunk for user 2 and $\msf{u}_6$ needs to have the clean chunk for user 1. This is due to the property of $\mbb{F}_2$. Hence we can again conclude $2B_1+2B_2 \le 3T$. Therefore, we see that this linear scheme over vector space of $\mbb{F}_2$ cannot achieve beyond the pentagon $\mfrak{P}$. 

On the other hand, if instead we use a linear scheme over finite field $\mbb{F}_4$, we are able to achieve $(1,1)$. Recall that from the standard construction of extension field, the field $\mbb{F}_4$ comprises $\{0, 1, \rho, \rho^2\}$ with the following addition and multiplication and one-to-one correspondence with $\lp\mbb{F}_2\rp^2$:
\begin{align*}
&\begin{array}{c|cccc}
+ & 0 & 1 & \rho & \rho^2\\
\hline
0 & 0 & 1 & \rho & \rho^2\\
1 & 1 & 0 & \rho^2 & \rho\\
\rho & \rho & \rho^2 & 0 & 1\\
\rho^2 & \rho^2 & \rho & 1 & 0
\end{array}&
&\begin{array}{c|cccc}
\times & 0 & 1 & \rho & \rho^2\\
\hline
0 & 0 & 0 & 0 & 0\\
1 & 0 & 1 & \rho & \rho^2\\
\rho & 0 & \rho & \rho^2 & 1\\
\rho^2 & 0 & \rho^2 & 1 & \rho
\end{array}&
&\begin{array}{c|c}
\mbb{F}_2\times\mbb{F}_2 & \mbb{F}_4\\
\hline
(0,0) & 0\\
(0,1) & 1\\
(1,0) & \rho\\
(1,1) & \rho^2
\end{array}
\end{align*}
Therefore, we can use two time slots to translate the following scalar coding scheme over $\mbb{F}_4$ (depicted in Fig.~\ref{fig_Examples_1}(b)) back to a \emph{nonlinear} coding scheme over $\lp\mbb{F}_2\rp^2$: $a,b\in\mbb{F}_4$,
\begin{center}
{\small
\begin{tabular}{c|c|c|c|c|c|c|c|c|c|c }
& $\msf{s}_1$ & $\msf{s}_2$ & $\msf{u}_1$ & $\msf{u}_2$ & $\msf{u}_3$ & $\msf{u}_4$ & $\msf{u}_5$ & $\msf{u}_6$ & $\msf{d}_1$ & $\msf{d}_2$ \\ \hline
Transmits & $a$ & $b$ & $\rho a$ & $a+b$ & $0$ & $\rho^2a+b$ & $a+b$ & $a+\rho b$ & & \\ \hline
Receives & & & $a$ & $a+b$ & $b$ & $\rho^2a+b$ & $a+b$ & $\rho^2a+b$ & $\rho a$ & $\rho^2b$ 
\end{tabular}
}
\end{center}
Note that since the network is layered, one can without loss of generality assume that there is no processing delay within a node.

From the above example we see the benefit of working in the extension field is that, at each node there are more choices of scaling coefficients. In vector space $\lp\mbb{F}_2\rp^r$, $r\ge 2$, the encoding matrix at each node has entries that are either $0$ or $1$, which limits the achievable rates of such a scheme.

\subsection{Example of Networks with Different Capacity Regions}
We provide examples of networks for each of the five possible capacity regions and use them to illustrate the important elements in our proposed scheme (including interference neutralization and zero forcing).

{\flushleft 1) Network with capacity region $\mfrak{T}$}\par
The example is depicted in Fig.~\ref{fig_examples}(a). For achievability we know that $\mfrak{T}$ can be achieved via time-sharing between rate pairs $(1,0)$ and $(0,1)$. For the outer bound, we notice that $\msf{u}_4$ is omniscient, and the reception of the destination $\msf{d}_1$ is a function of the reception of $\msf{u}_4$. This means $\msf{u}_4$ can decode the message of $\msf{s}_1$. The reception of each node in $\mcal{K}^{\msf{s_2}}(\msf{u}_4) = \{\msf{u}_4.\msf{u}_3\}$ is some function of the reception of node $\msf{u}_4$ and the transmission of $\msf{s}_1$. Since $\msf{u}_4$ can now recover the transmission of $\msf{s}_1$, and since $\mcal{K}^{\msf{s_2}}(\msf{u}_4)$ forms a $\lp \msf{s}_2; \msf{d}_2\rp$-vertex-cut, $\msf{u}_4$ can recover the reception of $\msf{d}_2,$ and thus, also the message of $\msf{s}_2$. Therefore, the sum rate cannot be greater than the maximum entropy of the reception of $\msf{u}_4$, which is $1$. 

{\flushleft 2) Network with capacity region $\mfrak{T}_{12}$}\par
The example is depicted in Fig.~\ref{fig_examples}(b) (without the dashed edge). We shall use this example to illustrate $(1/2,1)$-achievability. For achieving the rate pair $(1/2,1)$, we use the following scheme:
\begin{center}
{\small
\begin{tabular}{c|c|c|c|c|c|c|c|c|c|c }
& $\msf{s}_1$ & $\msf{s}_2$ & $\msf{u}_1$ & $\msf{u}_2$ & $\msf{u}_3$ & $\msf{u}_4$ & $\msf{u}_5$ & $\msf{u}_6$ & $\msf{d}_1$ & $\msf{d}_2$ \\ \hline
Time 1 Transmits & $a$ & $b_1$ & $a$ & $b_1$ & $b_1$ & $a+b_1$ & $b_1$ & $0$ & & \\ \hline
Time 1 Receives & & & $a$ & $b_1$ & $b_1$ & $a+b_1$ & $b_1$ & $a+b_1$ & $a+b_1$ & $b_1$ \\ \hline\hline
Time 2 Transmits & $a$ & $b_2$ & $a$ & $0$ & $b_2$ & $a$ & $0$ & $\boxed{b_2-b_1}$ & & \\ \hline
Time 2 Receives & & & $a$ & $b_2$ & $b_2$ & $a$ & $0$ & $a+b_2$ & $a$ & $\boxed{b_2-b_1}$ 
\end{tabular}
}
\end{center}
Note that in the first time slot, all nodes transmit what they receive except for $\msf{u}_6$. This is because the reception of $\msf{u}_6$ contains $a$ and hence it transmits $0$ instead so that $\msf{d}_2$ receives $b_1$. In the second time slot, $\msf{u}_2$ has to keep silent so that $\msf{u}_4$, the critical node for $\msf{d}_1$, is able to decode $a$. $\msf{u}_5$ hence receives $0$, and $b_2$ needs to be provided by $\msf{u}_6$. Still, it is necessary for $\msf{u}_6$ to transmit a linear combination that does not contain $a$. Therefore, it makes use of the two linear combinations it receives over the two time slots, $a+b_1$ and $a+b_2$, to \emph{zero-force} interference $a$ and sends out $b_2-b_1$.

To see that the capacity region is $\mfrak{T}_{12}$, we shall verify that the network satisfies $\textrm{T}^{(12)}$. Obviously $\textrm{T}^{(12)}_1$, $\textrm{T}^{(12)}_2$, and $\textrm{T}^{(12)}_4$ hold, as $\mcal{P}^{\msf{s}_2}(\msf{v}_1^*)=\{\msf{u}_2\}$. Induced graph $\mcal{G}_{12}$ is $\mcal{G}$ with edges $(\msf{u}_2,\msf{u}_4)$ and $(\msf{u}_2,\msf{u}_5)$ deleted. It can be seen that $\msf{u}_6$ becomes omniscient in $\mcal{G}_{12}$. Therefore $\textrm{T}^{(12)}_3$ also holds. 

{\flushleft 3) Network with capacity region $\mfrak{P}$}\par
The example is the one depicted in Fig.~\ref{fig_examples}(b) with an additional (dashed) edge $(\msf{s}_1,\msf{u}_2)$. To see that the capacity region is $\mfrak{P}$, we shall verify that both $(1/2,1)$ and $(1,1/2)$ are achievable and the network satisfies $\textrm{P}^{(12)}$. To achieve $(1/2,1)$, we use a similar scheme as above except that in the first time slot, $\msf{u}_5$ and $\msf{u}_6$ have to carry out interference neutralization to cancel $a$ over the air. To achieve $(1,1/2)$, we use the following scheme:
\begin{center}
{\small
\begin{tabular}{c|c|c|c|c|c|c|c|c|c|c }
& $\msf{s}_1$ & $\msf{s}_2$ & $\msf{u}_1$ & $\msf{u}_2$ & $\msf{u}_3$ & $\msf{u}_4$ & $\msf{u}_5$ & $\msf{u}_6$ & $\msf{d}_1$ & $\msf{d}_2$ \\ \hline
Time 1 Transmits & $a_1$ & $b$ & $a_1$ & $0$ & $b$ & $a_1$ & $0$ & $a_1+b$ & & \\ \hline
Time 1 Receives & & & $a_1$ & $a_1+b$ & $b$ & $a_1$ & $0$ & $a_1+b$ & $a_1$ & $a_1+b$ \\ \hline\hline
Time 2 Transmits & $a_2$ & $b$ & $0$ & $\boxed{a_2-a_1}$ & $b$ & $a_2-a_1$ & $0$ & $b$ & & \\ \hline
Time 2 Receives & & & $a_2$ & $a_2+b$ & $b$ & $\boxed{a_2-a_1}$ & $\boxed{a_2-a_1}$ & $b$ & $a_2-a_1$ & $b$ 
\end{tabular}
}
\end{center}
Note that in the first time slot, all nodes transmit what they receive except for $\msf{u}_2$. This is because the reception of $\msf{u}_6$ contains $b$ and hence it transmits $0$ instead so that $\msf{u}_4$, the critical node for $\msf{d}_1$, receives $a_1$. In the second time slot, $\msf{u}_2$ makes use of the two linear combinations it receives over the two time slots, $a_1+b$ and $a_2+b$, to zero-force interference $b$ and sends out $a_2-a_1$. Meanwhile, $\msf{u}_1$ keeps silent so that $\msf{u}_6$ is able to decode $b$.

For the outer bound, obviously $\textrm{P}^{(12)}_1$, $\textrm{P}^{(12)}_2$, and $\textrm{P}^{(12)}_4$ hold, as $\mcal{P}^{\msf{s}_2}(\msf{v}_1^*)=\{\msf{u}_2\}$ and $\msf{w}_{12}=\msf{u}_2$. Induced graph $\mcal{G}_{12}$ is $\mcal{G}$ with edges $(\msf{u}_2,\msf{u}_4)$ and $(\msf{u}_2,\msf{u}_5)$ deleted. It can be seen that $\msf{u}_6$ becomes omniscient in $\mcal{G}_{12}$. Therefore $\textrm{P}^{(12)}_3$ also holds.  

{\flushleft 4) Network with capacity region $\mfrak{S}$}\par
The example is depicted in Fig.~\ref{fig_Examples_1}(a). Interference neutralization happens right at $\msf{d}_1$ and $\msf{d}_2$, which is carried out by $\msf{u}_4,\msf{u}_5,\msf{u}_6$. As explained in the previous subsection, such interference neutralization is not possible without coding in the extension field.

\section{Proof of Achievability}\label{sec_Achieve}
In this section, we shall establish various achievability results \emph{beyond} the trivially-achievable triangular rate region $\mfrak{T}$. We assume that in the network no nodes are omniscient and describe a coding scheme that achieves $(1,1)$, $(1,1/2)$, or $(1/2,1).$
We will use a linear scheme over the finite field $\mbb{F}_{2^r}$, for some $r>0.$ We map the $r$-length binary sequences in $(\mbb{F}_2)^r$ to symbols in $\mbb{F}_{2^r}$ such that the bitwise modulo-two addition in $(\mbb{F}_2)^r$ translates to the addition operation in $\mbb{F}_{2^r}.$ Such a mapping is always possible by the standard construction of the extension field $\mbb{F}_{2^r}$.
Under such a mapping, we are able to abstract $r$ usages in the original network to a single channel use in a network with the same topology, but with inputs and outputs in the extension field $\mbb{F}_{2^r}$. A node is said to perform \emph{Random Linear Coding} (RLC) over $\mbb{F}_{2^r}$ if the coefficient(s) chosen by the node in the linear transformation is chosen uniformly at random in $\mbb{F}_{2^r}$ and independently of the coefficients chosen by all its predecessors.

We will focus on schemes achieving rate pairs $(1,1)$ and $(1/2,1)$ respectively. To achieve $(1,1),$ each source aims to convey a symbol in $\mbb{F}_{2^r}$ to its own destination over one symbol-time slot. The block length used by each node would be $r.$ To achieve $(1/2,1)$, $\msf{s}_1$ aims to deliver one symbol while $\msf{s}_2$ aims to deliver two symbols to their respective destinations. The block length here would be $2r.$ 
Note that the functions transforming an incoming $T$-block of bits ($T=r$ for the $(1,1)$-scheme and $T=2r$ for the $(1/2,1)$-scheme and) to an outgoing $T$-block of bits is not a linear transformation over the vector space $(\mbb{F}_2)^T$ and must necessarily be understood as operations over the extension field $\mbb{F}_{2^r}$ for our proofs to work. 


A scalar linear coding scheme over $\mbb{F}_{2^r}$ is specified by the following collection of \emph{linear coding coefficients}:
$\{\alpha_{\msf{v}}\in\mbb{F}_{2^r}:\msf{v}\in\mcal{V}\setminus\{\msf{d}_1,\msf{d}_2\}\}$.
Define for each $\msf{v}\in\mcal{V}$ the \emph{global coefficients} $\beta_{\msf{v},\msf{s}_1}, \beta_{\msf{v},\msf{s}_2}\in\mbb{F}_{2^r}$ as follows.
\begin{itemize}
\item Initialize: $\beta_{\msf{s}_1,\msf{s}_1}:=1,\ \beta_{\msf{s}_1,\msf{s}_2}:=0,\ \beta_{\msf{s}_2,\msf{s}_1}:=0,\ \beta_{\msf{s}_2,\msf{s}_2}:=1.$
\item For $\msf{v}\in\mcal{V}\setminus\{\msf{s}_1,\msf{s}_2\},$ we define 
  \begin{align*}
  \beta_{\msf{v},\msf{s}_1}:=\sum_{\msf{u}\in\mcal{P}(\msf{v})}\alpha_{\msf{u}}\beta_{\msf{u},\msf{s}_1}, \quad
  \beta_{\msf{v},\msf{s}_2}:=\sum_{\msf{u}\in\mcal{P}(\msf{v})}\alpha_{\msf{u}}\beta_{\msf{u},\msf{s}_2}.
  \end{align*}
\end{itemize}

If the messages of source $\msf{s}_1$ and $\msf{s}_2$ are $a$ and $b$ respectively, then the reception of node $\msf{v}\in\mcal{V}\setminus\{\msf{s}_1,\msf{s}_2\}$ is given by $\beta_{\msf{v},\msf{s}_1}\cdot a+\beta_{\msf{v},\msf{s}_2}\cdot b.$

Recall Lemma~\ref{lem_critical} and \ref{lem_critical_omniscient}. These two lemmas explain why we define critical nodes. Lemma~\ref{lem_critical} shows that the rank-influence from the sources to destination $\msf{d}_i$ drops precisely at the critical node $\msf{v}_i^*$ and hence, the nodes in $\mcal{P}(\msf{v}_i^*)$ are natural candidates for special coding so as to cancel interference and arrange user $i$'s symbol(s) to be received at $\msf{v}_i^*$ even while other nodes may perform random linear coding. Note that this kind of special coding is a linear operation over the finite field $\mbb{F}_{2^r}$ making use of the superposition feature of the channel. Lemma~\ref{lem_critical_omniscient} shows that the critical nodes suffice to capture the property of existence of an omniscient node in the network.



The reception of destination $\msf{d}_i$ is just a function of that of the critical node $\msf{v}_i^*.$ Hence we define the cloud $\cloud_i,$ for $i=1,2,$ to be the set of nodes that can be reached by some node in $\mcal{K}(\msf{v}_i^*)$ and that can reach $\msf{d}_i.$ All nodes in the cloud receive functions of the reception of the critical node. Our scheme will ensure that $\msf{v}_i^*$ can decode what $\msf{d}_i$ aims to decode, $i=1,2$.

Below we provide several useful lemmas. Proofs of these lemmas can be found in Appendix~\ref{app_PfLemmas}.


\begin{lemma}\label{lem_reachability}
If $\msf{u}$ is $\msf{s}_i$-reachable, and all its predecessors do RLC with one symbol from each source, 
then $\msf{s}_i$'s symbol has a non-zero coefficient in the reception of $\msf{u}$ with high probability.
\end{lemma}

\begin{lemma}\label{lem_two_source}
Consider $\mcal{U}\subseteq\mcal{L}_k$ with $\mincut(\msf{s}_1,\msf{s}_2;\mcal{U})=2$. Suppose each source transmits one symbol and all nodes in the network up to and including layer $\mcal{L}_{k-1}$ perform RLC.
{\flushleft (a)}
Then the nodes in $\mcal{U}$ can collectively decode both of the transmitted symbols with high probability.
{\flushleft (b)}
If a node $\msf{v}$ has $\mcal{U} \subseteq \mcal{P}(\msf{v})$, then with all nodes except nodes in $\mcal{U}$ performing arbitrary linear coding, nodes in $\mcal{U}$ can arrange their transmission so that $\msf{v}$ receives any desired linear combination of the source symbols with high probability.
{\flushleft (c)}
Let $\msf{u}\in\mcal{U}\subseteq\mcal{P}(\msf{v}).$ If nodes in $\mcal{P}(\msf{v})\setminus \mcal{U}$ stay silent and nodes in $\mcal{U}\setminus\{\msf{u}\}$ do RLC, then $\msf{u}$ is able to arrange its transmission so that $\msf{v}$ receives any linear combination linearly independent of the reception of $\msf{u}$ with high probability.
{\flushleft (d)}
As a corollary of (c), if the node $\msf{u}$ is $\msf{s}_1\msf{s}_2$-reachable, then $\msf{u}$ can adjust its transmission so that $\msf{v}$ can decode either $\msf{s}_1$ or $\msf{s}_2$'s symbol with high probability.
\end{lemma}

\begin{lemma}\label{lem_mincut_pair}
If $\mcal{U}\subseteq \mcal{L}_k$ satisfies $\mincut(\msf{s}_1.\msf{s}_2;\mcal{U})=2$, then for any $\msf{u}\in\mcal{U}$, we can find some $\msf{w}\in\mcal{U}$ such that $\mincut(\msf{s}_1.\msf{s}_2;\msf{u},\msf{w})=2$.
\end{lemma}

Next, we shall prove achievability in different cases. Formal proofs of lemmas and claims are left in Appendix~\ref{app_PfLemmas}.
Without loss of generality, we assume that $k_1^* \le k_2^*$.  If $k_1^*=0$, based on Lemma~\ref{k_1^*=0} we know that if there is no omniscient node, then $(1,1)$ is achievable. If $k_1^*=1$, then by the definition of critical node $\msf{v}_1^*$, both $\msf{s}_1$ and $\msf{s}_2$ are $\msf{v}_1^*$'s parents and hence it is omniscient.
Therefore we focus on $2\le k_1^*\le k_2^*$ below. We shall distinguish into two cases: $k_1^*=k_2^*$ and $k_1^* < k_2^*$. 

\subsection{$k_1^* = k_2^* = k^*$}
\subsubsection{Special Patterns Implied by the Conditions}
When the critical nodes are in the same layer, it turns out that if the network $\mcal{G}$ satisfies the conditions given in Theorem~\ref{thm_2} or Theorem~\ref{thm_3}, it has a special pattern. The fact is summarized in the following lemma. Let
$\mcal{P}_1 := \mcal{P}\lp \msf{v}_1^*\rp \setminus \mcal{P}\lp \msf{v}_2^*\rp$, $\mcal{P}_2 := \mcal{P}\lp \msf{v}_2^*\rp \setminus \mcal{P}\lp \msf{v}_1^*\rp$,
$\mcal{P}_{12} := \mcal{P}\lp \msf{v}_1^*\rp \cap \mcal{P}\lp \msf{v}_2^*\rp$.

\begin{lemma}\label{lem_Pattern}
When $k_1^* = k_2^*=k^*$ and there is no omniscient node, we have the following equivalence relations.
\begin{align*}
&\textrm{T}^{(12)} \iff \\
&\quad\lbp\begin{array}{l}
\mcal{P}_1 \text{ is } \msf{s}_1\text{-only-reachable}\\ 
\mincut\lp\msf{s}_1,\msf{s}_2;\mcal{P}_2\rp = 1, \msf{u}_{21}:=\pmc\lp\mcal{P}_2\rp\ne\msf{s}_i, i=1,2\\
\mcal{P}_{12} \text{ is } \msf{s}_2\text{-only-reachable}\\
\mcal{K}^{\msf{s}_1}\lp \msf{u}_{21}\rp \text{ forms } \lp\msf{s}_1;\msf{d}_1\rp\text{-vertex-cut}.
\end{array}\right.\\
&\textrm{P}^{(12)}\setminus\textrm{T}^{(21)} \iff \\
&\quad\lbp\begin{array}{l}
\mcal{P}_1 \text{ is } \msf{s}_1\text{-only-reachable}\\ 
\mincut\lp\msf{s}_1,\msf{s}_2;\mcal{P}_2\rp = 1, \msf{u}_{21}:=\pmc\lp\mcal{P}_2\rp\ne\msf{s}_i, i=1,2\\
\mincut\lp\msf{s}_1,\msf{s}_2;\mcal{P}_{12}\rp = 1, \msf{w}_{12}:=\pmc\lp\mcal{P}_{12}\rp\ne\msf{s}_i, i=1,2\\
\mcal{K}^{\msf{s}_2}\lp \msf{w}_{12}\rp \text{ forms } \lp\msf{s}_2;\msf{d}_2\rp\text{-vertex-cut}.\\
\mcal{K}^{\msf{s}_1}_{\mcal{G}_{12}}\lp \msf{u}_{21}\rp \text{ forms } \lp\msf{s}_1;\msf{d}_1\rp\text{-vertex-cut in } \mcal{G}_{12}.
\end{array}\right.
\end{align*}
and the equivalence relation for $\textrm{T}^{(21)}$ ($\textrm{P}^{(21)}\setminus\textrm{T}^{(12)}$) is the one for $\textrm{T}^{(12)}$ ($\textrm{P}^{(12)}\setminus\textrm{T}^{(21)}$) with indices ``1" and ``2" exchanged.
\end{lemma}
\begin{proof}Proof is detailed in the appendix.\end{proof}

One direct consequence of the above lemma is that, $\textrm{T}^{(12)}\cap \textrm{T}^{(21)} = \textrm{P}^{(12)}\cap \textrm{P}^{(21)} = \emptyset$.

\subsubsection{Proof of Achievability}
In this case, it is sufficient to show that $\msf{v}_1^*, \msf{v}_2^*$ can decode the symbols desired by destinations $\msf{d}_1,\msf{d}_2$ respectively. This is because the network past layer $\mcal{L}_{k^*}$ has no interference to $\msf{d}_i$ from any node in $\mcal{K}(\msf{v}_j^*)$ for $(i,j)=(1,2)$ or $(2,1).$


By definition, we can see that $\mcal{K}(\msf{v}_1^*) \cup \mcal{K}(\msf{v}_2^*) = \mcal{L}_{k^*}$. For suppose, there exists node $\msf{u}\in \mcal{L}_{k^*}\setminus \lp \mcal{K}(\msf{v}_1^*) \cup \mcal{K}(\msf{v}_2^*) \rp.$ As each node can reach at least one of the destinations, $\msf{u}$ can reach either $\msf{d}_1$ or $\msf{d}_2$ thus violating the definition of a critical node. Now, suppose $\mcal{K}(\msf{v}_1^*) \cap \mcal{K}(\msf{v}_2^*) \ne \emptyset$, then $\mcal{K}(\msf{v}_1^*) = \mcal{K}(\msf{v}_2^*)$ and so, both $\msf{v}_1^*$ and $\msf{v}_2^*$ are omniscient, 
violating the assumption. Hence $\mcal{K}(\msf{v}_1^*)$ and $\mcal{K}(\msf{v}_2^*)$ form a partition of $\mcal{L}_{k^*}$.

Since neither $\msf{v}_1^*$ nor $\msf{v}_2^*$ is omniscient, $\mcal{P}^{\msf{s}_1}(\msf{v}_1^*) \ne \mcal{P}^{\msf{s}_1}(\msf{v}_2^*)$ and $\mcal{P}^{\msf{s}_2}(\msf{v}_2^*) \ne \mcal{P}^{\msf{s}_2}(\msf{v}_1^*)$.
It can be stated equivalently as
\begin{align*}
&\mcal{P}^{\msf{s}_1}(\msf{v}_1^*) \setminus \mcal{P}^{\msf{s}_1}(\msf{v}_2^*) \ne \emptyset \textrm{ or } \mcal{P}^{\msf{s}_1}(\msf{v}_1^*) \subsetneq \mcal{P}^{\msf{s}_1}(\msf{v}_2^*) \textrm{ and }\\
&\mcal{P}^{\msf{s}_2}(\msf{v}_2^*) \setminus \mcal{P}^{\msf{s}_2}(\msf{v}_1^*) \ne \emptyset \textrm{ or } \mcal{P}^{\msf{s}_2}(\msf{v}_2^*) \subsetneq \mcal{P}^{\msf{s}_2}(\msf{v}_1^*)
\end{align*}
For notational convenience, let us define $\mcal{P}_1^{\msf{s}_1} := \mcal{P}^{\msf{s}_1}\lp \msf{v}_1^*\rp \setminus \mcal{P}^{\msf{s}_1}\lp \msf{v}_2^*\rp$ and $\mcal{P}_2^{\msf{s}_2} := \mcal{P}^{\msf{s}_2}\lp \msf{v}_2^*\rp \setminus \mcal{P}^{\msf{s}_2}\lp \msf{v}_1^*\rp$.

Below we first show that $(1,1)$ is achievable when the network $\mcal{G}$ does not fall into any of the above four patterns described in Lemma~\ref{lem_Pattern}. Next we show that in the patterns corresponding to $\textrm{P}^{(12)}\setminus\textrm{T}^{(21)}$ and $\textrm{P}^{(21)}\setminus\textrm{T}^{(12)}$, both $(1,1/2)$ and $(1/2,1)$ can be achieved. Finally we show that in the pattern corresponding to $\textrm{T}^{(12)}$, $(1/2,1)$ can be achieved, and in the pattern corresponding to $\textrm{T}^{(21)}$, $(1,1/2)$ can be achieved.

As a first step, we show the following claim.
\begin{claim}\label{claim_SameLay_1}
$(1,1)$ is achievable if $\mcal{P}_1^{\msf{s}_1} = \emptyset$ or $\mcal{P}_2^{\msf{s}_2} = \emptyset$ or $\mcal{P}_{12}=\emptyset$, under the assumption that there is no omniscient node. 
\end{claim}
\begin{proof}
See appendix.
\end{proof}

In the following we focus on the case where $\mcal{P}_1^{\msf{s}_1}\ne\emptyset$, $\mcal{P}_2^{\msf{s}_2}\ne\emptyset$, and $\mcal{P}_{12}\ne\emptyset$.
We then show the following claim.
\begin{claim}\label{claim_SameLay_2}
Consider the conditions
\begin{description}
\item [$A1$] $\forall\ \msf{u}_1 \in \mcal{P}_1^{\msf{s}_1}$, $\msf{u}_1$ is $\msf{s}_1$-only-reachable.
\item [$A2$] $\forall\ \msf{u}_2 \in \mcal{P}_2^{\msf{s}_2}$, $\msf{u}_2$ is $\msf{s}_1\msf{s}_2$-reachable.
\item [$B1$] $\forall\ \msf{u}_1 \in \mcal{P}_1^{\msf{s}_1}$, $\msf{u}_1$ is $\msf{s}_1\msf{s}_2$-reachable.
\item [$B2$] $\forall\ \msf{u}_2 \in \mcal{P}_2^{\msf{s}_2}$, $\msf{u}_2$ is $\msf{s}_2$-only-reachable.
\end{description}
Let $A = A1\wedge A2$ and $B = B1\wedge B2$. Then the negation of $A\vee B$ implies that $(1,1)$ is achievable.
\end{claim}
\underline{Remark}: Note that $A\vee B$ is implied by the disjunction of $\textrm{T}^{(12)}$, $\textrm{T}^{(21)}$, $\textrm{P}^{(12)}$, and $\textrm{P}^{(21)}$. Therefore this claim proves $(1,1)$-achievability for some cases.
\begin{proof}
See appendix.
\end{proof}

So far we have demonstrated $(1,1)$-achievability when condition $A\vee B$ is not satisfied. Since $A$ and $B$ are disjoint, we can separate into two different cases. Besides, discussion on one case will lead to similar arguments for the other case by symmetry.

{\flushleft \bf Case $A$}: $\forall\ \msf{u}_1 \in \mcal{P}_1^{\msf{s}_1}$, $\msf{u}_1$ is $\msf{s}_1$-only-reachable, and $\forall\ \msf{u}_2 \in \mcal{P}_2^{\msf{s}_2}$, $\msf{u}_2$ is $\msf{s}_1\msf{s}_2$-reachable.\par
For this case, if $\mcal{P}_1\setminus\mcal{P}_1^{\msf{s}_1}\ne \emptyset$, that is, there exists a node in $\mcal{P}_1$ and it is $\msf{s}_2$-only-reachable, then $\mincut\lp \msf{s}_1,\msf{s}_2; \mcal{P}_1\rp = 2$. We can achieve $(1,1)$, by first arranging the transmission of $\mcal{P}(\msf{v}_2^*)$ so that $\msf{v}_2^*$ can decode $b$ and then arranging the transmission of $\mcal{P}_1$ to form any linear combination of $a$ and $b$; in particular, the one that combined with the transmission from $\mcal{P}_{12}$ forms $a$ at $\msf{v}_1^*$. If $\mcal{P}_2\setminus\mcal{P}_2^{\msf{s}_2}\ne \emptyset$, that is, there exists a node in $\mcal{P}_2$ and it is $\msf{s}_1$-only-reachable, then $\mincut\lp \msf{s}_1,\msf{s}_2; \mcal{P}_2\rp = 2$. $(1,1)$ is then achievable by a similar argument as above.

We now narrow down to the case $\forall\ \msf{u}_1 \in \mcal{P}_1$, $\msf{u}_1$ is $\msf{s}_1$-only-reachable, and $\forall\ \msf{u}_2 \in \mcal{P}_2$, $\msf{u}_2$ is $\msf{s}_1\msf{s}_2$-reachable. If $\mincut\lp \msf{s}_1,\msf{s}_2; \mcal{P}_2\rp = 2$, obviously $(1,1)$ is achievable, as $\msf{v}_1^*$ can always get $a$ from $\mcal{P}_1$ (whose transmission does not affect $\msf{v}_2^*$) and one can arrange $\mcal{P}_2$'s transmission (which does not affect $\msf{v}_1^*$) to ensure $\msf{v}_2^*$ decode $b$. If $\mincut\lp \msf{s}_1,\msf{s}_2; \mcal{P}_{12}\rp = 2$, we can achieve $(1,1)$ by arranging the transmission of $\mcal{P}_{12}$ so that their aggregate is $a$. Hence $\msf{v}_1^*$ can decode $a$. Then nodes in $\mcal{P}_2$ just scale their received linear combinations so that $a$ gets neutralized at $\msf{v}_2^*$ and only $b$ is left. If $\mincut\lp \msf{s}_1,\msf{s}_2; \mcal{P}_{12}\rp = 1$, we identify $\msf{w}_{12} = \pmc\lp \mcal{P}_{12}\rp$. If $\mcal{K}^{\msf{s}_2}(\msf{w}_{12})$ does not form a $(\msf{s}_2;\msf{d}_2)$-vertex-cut, we can arrange its parents' transmission so that $\msf{w}_{12}$ can decode $a$, and at the same time $\mcal{P}_2$ can receive a linear combination with a non-zero $b$-coefficient. Hence nodes in $\mcal{P}_{12}$ can send out a scaled version of $a$ to neutralize $a$ at $\msf{v}_2^*$ if necessary, and $\msf{v}_1^*$ can always obtain $a$ from $\mcal{P}_1$.

So far we have shown that in Case $A$, if one of the following is violated, then $(1,1)$ is achievable:
\begin{itemize}
\item $\mcal{P}_1$ is $\msf{s}_1$-only-reachable
\item $\mcal{P}_2$ is $\msf{s}_1\msf{s}_2$-reachable, and $\mincut\lp\msf{s}_1,\msf{s}_2;\mcal{P}_2\rp = 1$
\item $\mincut\lp\msf{s}_1,\msf{s}_2;\mcal{P}_{12}\rp = 1$, $\mcal{K}^{\msf{s}_2}(\msf{w}_{12})$ forms a $(\msf{s}_2;\msf{d}_2)$-vertex-cut
\end{itemize} 

To complete the proof of $(1,1)$-achievability, we need to show that if $\msf{u}_{21}:=\pmc\lp\mcal{P}_2\rp$ is not ommiscient in $\mcal{G}_{12}$, then $(1,1)$ can be achieved. We can simply arrange the transmission of $\mcal{P}_{12}$ so that their aggregate becomes $0$ at $\msf{v}_1^*$. Effectively we are in $\mcal{G}_{12}$ with this special linear coding operation. 
Since in $\mcal{G}_{12}$, $\msf{d}_1$ is $\msf{s}_1$-only-reachable and $\msf{u}_{21}$ is the new critical node of $\msf{d}_2$, by Lemma~\ref{k_1^*=0} we know that $(1,1)$ can be achieved in $\mcal{G}_{12}$. We then translate the linear coding scheme in $\mcal{G}_{12}$ back to a linear coding scheme in $\mcal{G}$.

The next thing to show for Case $A$: if a network is in $\textrm{P}^{(12)}\setminus\textrm{T}^{(21)}$, then $(1,1/2)$ can be achieved. To show it, we employ a two-time-slot coding scheme. We aim to deliver two symbols $a_1$ and $a_2$ for user 1 and one symbol $b$ for user 2 over two symbol time slots. Symbols are drawn from the extension field $\mbb{F}_{2^r}$.  In the first time slot, we do RLC with $\msf{s}_1$ transmitting $a_1$ and $\msf{s}_2$ transmitting $b$, up to layer $\mcal{L}_{k^*-1}$. Pick one node $\msf{w}\in\mcal{P}_{12}$ and one node $\msf{u}\in\mcal{P}_2$. Keep other nodes in $\mcal{P}_{12}$ and $\mcal{P}_2$ silent, while nodes in $\mcal{P}_1$ do RLC. We turn off the transmission of $\msf{w}$. $\msf{v}_1^*$ can then decode $a_1$. In the second time slot again we do RLC with $\msf{s}_1$ transmitting $a_2$ and $\msf{s}_2$ transmitting $b$, up to layer $\mcal{L}_{k^*-1}$. We use the two linear combinations $\msf{w}$ receives over the two time slots to zero-force (ZF) $b$ and produce a linear combination of $a_1$ and $a_2$: (the superscripts of $\beta$'s denote the time indices)
\begin{align}
&\beta_{\msf{w},\msf{s}_1}^{(1)}\cdot a_1 + \beta_{\msf{w},\msf{s}_2}^{(1)} \cdot b;\ \beta_{\msf{w},\msf{s}_1}^{(2)}\cdot a_2 + \beta_{\msf{w},\msf{s}_2}^{(2)} \cdot b \notag\\
&\overset{\rm{ZF}}{\Longrightarrow}\ \beta_{\msf{w},\msf{s}_1}^{(2)}\cdot a_2 + \frac{\beta_{\msf{w},\msf{s}_1}^{(1)}}{\beta_{\msf{w},\msf{s}_2}^{(1)}}\cdot \beta_{\msf{w},\msf{s}_2}^{(2)} \cdot a_1. \label{eq_LC1}
\end{align} 
$\msf{w}$ then scales this ZF output and sends it out. Hence $\msf{v}_1^*$ can use $a_1$ as side information to decode $a_2$.

As for user 2, in the first time slot $\msf{v}_2^*$ receives a linear equation from $\msf{u}$: $\beta_{\msf{u},\msf{s}_1}^{(1)}\cdot a_1 + \beta_{\msf{u},\msf{s}_2}^{(1)} \cdot b$. In the second time slot $\msf{u}$ receives $\beta_{\msf{u},\msf{s}_1}^{(2)}\cdot a_2 + \beta_{\msf{u},\msf{s}_2}^{(2)} \cdot b$. $\msf{u}$ makes use of of the two linear combinations to zero-force $b$ and generate a linear combination of $a_1$ and $a_2$:
\begin{align}
&\beta_{\msf{u},\msf{s}_1}^{(1)}\cdot a_1 + \beta_{\msf{u},\msf{s}_2}^{(1)} \cdot b;\ \beta_{\msf{u},\msf{s}_1}^{(2)}\cdot a_2 + \beta_{\msf{u},\msf{s}_2}^{(2)} \cdot b \notag\\
&\overset{\rm{ZF}}{\Longrightarrow}\ \beta_{\msf{u},\msf{s}_1}^{(2)}\cdot a_2 + \frac{\beta_{\msf{u},\msf{s}_1}^{(1)}}{\beta_{\msf{u},\msf{s}_2}^{(1)}}\cdot \beta_{\msf{u},\msf{s}_2}^{(2)} \cdot a_1. \label{eq_LC2}
\end{align} 
As long as the two linear combinations in \eqref{eq_LC1} and \eqref{eq_LC2} are not aligned, $\msf{u}$ can scale \eqref{eq_LC2} properly to form $a_1$ at $\msf{v}_2^*$ in the second time slot. Then with reception of the first time slot, $\msf{v}_2^*$ can decode $b$.

The two linear combinations in \eqref{eq_LC1} and \eqref{eq_LC2} are aligned if and only if the determinant
\begin{align*}
&\begin{vmatrix}
\beta_{\msf{w},\msf{s}_1}^{(2)}& \frac{\beta_{\msf{w},\msf{s}_1}^{(1)}}{\beta_{\msf{w},\msf{s}_2}^{(1)}}\cdot \beta_{\msf{w},\msf{s}_2}^{(2)} \\
\beta_{\msf{u},\msf{s}_1}^{(2)}& \frac{\beta_{\msf{u},\msf{s}_1}^{(1)}}{\beta_{\msf{u},\msf{s}_2}^{(1)}}\cdot \beta_{\msf{u},\msf{s}_2}^{(2)}
\end{vmatrix} = 0 \iff \\
&\beta_{\msf{u},\msf{s}_1}^{(1)}\beta_{\msf{w},\msf{s}_2}^{(1)}\beta_{\msf{w},\msf{s}_1}^{(2)}\beta_{\msf{u},\msf{s}_2}^{(2)} = \beta_{\msf{u},\msf{s}_2}^{(1)}\beta_{\msf{w},\msf{s}_1}^{(1)}\beta_{\msf{w},\msf{s}_2}^{(2)}\beta_{\msf{u},\msf{s}_1}^{(2)},
\end{align*}
which is of very low probability due to the same reason in Appendix~\ref{app_Pf_NonzeroDet}.

Therefore, we show that in Case $A$, if $\textrm{T}^{(12)}\vee \textrm{P}^{(12)}$ is violated, then $(1,1)$ can be achieved; if $\textrm{T}^{(12)}$ is violated, then $(1,1/2)$ can be achieved. It remains to show that if there is no omniscient node, then in Case $A$, $(1/2,1)$ is always achievable. 

We aim to deliver one symbol $a$ for user 1 and two symbols $b_1,b_2$ for user 2 over two time slots. 
Pick nodes $\msf{u}_1\in\mcal{P}_1,\msf{u}_2\in\mcal{P}_2,\msf{w}_2\in\mcal{P}_{12}$. Both $\msf{u}_2$ and $\msf{w}_2$ zero-force user 1's symbol $a$ and form a linear combination of $b_1,b_2$. These two linear equations are linearly independent with high probability, as shown in Appendix~\ref{app_Pf_NonzeroDet}. $\msf{w}_2$ transmits in the first time slot, while $\msf{u}_1$ and $\msf{u}_2$ transmit in the second time slot. Therefore $\msf{v}_1^*$ can obtain $a$, $\msf{v}_2^*$ can solve $b_1$ and $b_2$, and $(1/2,1)$ is achievable.


{\flushleft \bf Case $B$}: $\forall\ \msf{u}_2 \in \mcal{P}_2^{\msf{s}_2}$, $\msf{u}_2$ is $\msf{s}_2$-only-reachable, and $\forall\ \msf{u}_1 \in \mcal{P}_1^{\msf{s}_1}$, $\msf{u}_1$ is $\msf{s}_1\msf{s}_2$-reachable.\par
Similar to Case $A$, we show that in Case $B$, if $\textrm{T}^{(21)}\vee \textrm{P}^{(21)}$ is violated, then $(1,1)$ can be achieved; if $\textrm{T}^{(21)}$ is violated, then $(1/2,1)$ can be achieved. Besides, if there is no omniscient node, then in Case $B$, $(1,1/2)$ is always achievable.

\subsection{$k_1^* < k_2^*$}

Since $\msf{v}_1^*$ is not omniscient,
$\mcal{K}^{\msf{s}_2}(\msf{v}_1^*)$ does not form a $\lp\msf{s}_2;\msf{d}_2\rp$-vertex-cut, which is equivalent to 
\begin{align*}
\exists\ \msf{u}_2\in \mcal{L}_{k_1^*}\setminus \mcal{K}(\msf{v}_1^*) \textrm{ so that } \mcal{P}^{\msf{s}_2}(\msf{u}_2)\ne\emptyset,\mcal{P}^{\msf{s}_2}(\msf{u}_2)\ne \mcal{P}^{\msf{s}_2}(\msf{v}_1^*).
\end{align*}


The following lemma makes sure that $\msf{u}_2$ can still receive a linear combination where user 2's symbol has a non-zero coefficient with high probability even if $\mcal{P}^{\msf{s}_2}(\msf{v}_1^*)$ do some special linear coding.
\begin{lemma}\label{lem_special}
Consider all nodes doing RLC for each source sending one symbol up to $\mcal{L}_{k_1^*-1}$ including $\mcal{L}_{k_1^*-1}$, except $\mcal{P}^{\msf{s}_2}(\msf{v}_1^*)$. If $\mcal{P}^{\msf{s}_2}(\msf{u}_2)\ne \mcal{P}^{\msf{s}_2}(\msf{v}_1^*)$, then it is possible with high probability that $\mcal{P}^{\msf{s}_2}(\msf{v}_1^*)$ can arrange their transmission so that $\msf{v}_1^*$ receives a linear combination solely of user 1's symbol and $\msf{u}_2$ receives a linear combination of at least user 2's symbol, that is, the coefficient of user 2's symbol is non-zero.
\end{lemma}

\subsubsection{$(1,1)$-Achievability}
For the $(1,1)$-achievability we need to prove the following claim
\begin{claim}
$\neg \textrm{Q}^{(12)} \implies (1,1)$ is achievable.
\end{claim}
\begin{proof}
Since $\textrm{Q}^{(12)}_1$ is satisfied, in a network that does not satisfy $\textrm{Q}^{(12)}$, at least one of the following holds:
\begin{itemize}
\item 
$\mincut\lp \msf{s}_1,\msf{s}_2; \mcal{P}^{\msf{s}_2}(\msf{v}_1^*)\rp = 2$

\item 
$\mincut\lp \msf{s}_1,\msf{s}_2; \mcal{P}^{\msf{s}_2}(\msf{v}_1^*)\rp = 1$ and $\mcal{K}_{\mcal{G}_{12}}^{\msf{s}_1}\lp\msf{u}_{21}\rp$ does not form an $\lp \msf{s}_1; \msf{d}_1\rp$-vertex-cut in $\mcal{G}_{12}$

\item
$\mincut\lp \msf{s}_1,\msf{s}_2; \mcal{P}^{\msf{s}_2}(\msf{v}_1^*)\rp = 1$ and $\mcal{K}_{\mcal{G}_{12}}^{\msf{s}_1}\lp\msf{u}_{21}\rp$ forms an $\lp \msf{s}_1; \msf{d}_1\rp$-vertex-cut in $\mcal{G}_{12}$ and $\mcal{K}^{\msf{s}_2}(\msf{w}_{12})$ does not form a $(\msf{s}_2;\msf{d}_2)$-vertex-cut in $\mcal{G}$


\end{itemize}


{\flushleft \bf Case 1:} $\mincut\lp \msf{s}_1,\msf{s}_2; \mcal{P}^{\msf{s}_2}(\msf{v}_1^*)\rp = 2$\par 
In this case, the idea is to arrange the transmission of $\mcal{P}^{\msf{s}_2}(\msf{v}_1^*)$ so that their aggregate contains user 1's symbol $a$ only. Mathematically, we aim to have
\begin{align}
\label{eq_neutralize_1}
\sum_{\msf{w}\in\mcal{P}^{\msf{s}_2}(\msf{v}_1^*)} \alpha_{\msf{w}}\beta_{\msf{w},\msf{s}_1}
 \ne 0,\  
\sum_{\msf{w}\in\mcal{P}^{\msf{s}_2}(\msf{v}_1^*)} \alpha_{\msf{w}}\beta_{\msf{w},\msf{s}_2}
 = 0, 
\end{align}
This is doable since $\mincut\lp \msf{s}_1,\msf{s}_2; \mcal{P}^{\msf{s}_2}(\msf{v}_1^*)\rp = 2$. Other nodes in the same layer simply do RLC. Therefore $\msf{v}_1^*$ is able to decode $a$. Nodes in $\mcal{C}_1$ do RLC and $\msf{d}_1$ can decode $a$. 

As for user 2, we look at $\msf{v}_2^*$. If $\msf{v}_2^*$ has no parents in the cloud $\mcal{C}_1$, we only need to guarantee that the parents of $\msf{v}_2^*$ can collectively decode $b$.
If $\msf{v}_2^*$ has some parent(s) in the cloud $\mcal{C}_1$, the parent(s) will inject user 1's symbol $a$ to the reception of $\msf{v}_1^*$. As $\msf{v}_2^*$ is not omniscient, there must exist $\msf{u}_1\in\mcal{L}_{k_2^*}\cap\cloud_1$ such that $\mcal{P}(\msf{u}_1) \ne \mcal{P}^{\msf{s}_1}(\msf{v}_2^*)$. 
If $\mcal{P}(\msf{u}_1) \subsetneq \mcal{P}^{\msf{s}_1}(\msf{v}_2^*)$, we can arrange the nodes in $\mcal{P}(\msf{u}_1)$ to make sure that $\msf{u}_1$ can decode user 1's symbol, and then arrange the nodes in $\mcal{P}^{\msf{s}_1}(\msf{v}_2^*)\setminus\mcal{P}(\msf{u}_1)$ to neutralize user 1's symbol at $\msf{v}_2^*$, given that these $\msf{s}_1$-reachable nodes can still receive a linear combination with non-zero $a$-coefficient under the special coding carried out by $\mcal{P}^{\msf{s}_2}(\msf{v}_1^*)$. Then since some node in $\mcal{P}(\msf{v}_2^*)$ will receive a linear combination with non-zero $b$-coefficient (eg., a successor of $\msf{u}_2$ in Lemma~\ref{lem_special}), one can always ensure $\msf{v}_2^*$ to decode $b$. If $\mcal{P}(\msf{u}_1) \setminus \mcal{P}^{\msf{s}_1}(\msf{v}_2^*)\ne \emptyset$, we can arrange the nodes in $\mcal{P}(\msf{v}_2^*)$ to form a linear combination that only contains $b$ at $\msf{v}_2^*$ given that the reception of $\mcal{P}(\msf{v}_2^*)$ can collectively decode $b$. Then use nodes in $\mcal{P}(\msf{u}_1)\setminus \mcal{P}^{\msf{s}_1}(\msf{v}_2^*)$ to place user 1's symbol at $\msf{u}_1$ if necessary. 

In summary, we want to guarantee that under the special linear coding carried out by $\mcal{P}^{\msf{s}_2}(\msf{v}_1^*)$ so that the neutralization criterion \eqref{eq_neutralize_1} is met, $\mcal{P}(\msf{v}_2^*)$ can still collectively decode $b$ and every node in $\mcal{P}^{\msf{s}_1}(\msf{v}_2^*)$ receives a linear combination with non-zero $a$-coefficient. The latter is quite obvious, as nodes affected by the special linear coding still receive linear combinations with non-zero $a$-coefficients. As for the former, note that if every node up to layer $k_2^*-2$ does RLC, it holds with high probability since $\mincut\lp\msf{s}_1,\msf{s}_2;\mcal{P}(\msf{v}_2^*)\rp=2$. With the special linear coding carried out by $\mcal{P}^{\msf{s}_2}(\msf{v}_1^*)$ described above, however, we cannot claim it with the existing random linear network coding argument.

We shall use the following two lemmas to overcome the difficulty, by breaking the network into two stages: one from the source layer to the layer $\mcal{L}_{k_1^*}$, and the other from layer $\mcal{L}_{k_1^*}$ to layer $\mcal{L}_{k_2^*-1}$. The first lemma claims that, under the special operation at $\mcal{P}^{\msf{s}_2}(\msf{v}_1^*)$ so that neutralization criterion \eqref{eq_neutralize_1} is satisfied, with high probability all nodes in layer $\mcal{L}_{k_1^*}$ that can reach $\msf{d}_2$ (call this set $\mcal{U}$) receive a non-zero linear combination of $a$ and $b$, and the subspace spanned by their reception has dimension two when all other nodes perform RLC. The second lemma claims that, once $\mcal{U}$'s reception satisfies the above property, then $\mcal{P}(\msf{v}_2^*)$ can collectively decode both $a$ and $b$ with high probability, when all nodes between $\mcal{L}_{k_1^*}$ and $\mcal{L}_{k_2^*-1}$ perform RLC. The lemmas are made concrete below.

\begin{lemma}[Reception of $\mcal{U}$]\label{lem_Reception}
Let us recall that $\mcal{U} := \lbp\msf{u}\in\mcal{L}_{k_1^*}: \msf{u} \text{ can reach } \msf{d}_1\rbp$. Consider RLC with $\msf{s}_1$ transmitting $a$ and $\msf{s}_2$ transmitting $b$. All nodes perform RLC up to layer $\mcal{L}_{k_1^*-1}$. In $\mcal{L}_{k_1^*-1}$, nodes except $\mcal{P}^{\msf{s}_2}(\msf{v}_1^*)$ also perform RLC. Then under special coding operation of $\mcal{P}^{\msf{s}_2}(\msf{v}_1^*)$ such that neutralization criterion \eqref{eq_neutralize_1} is satisfied, with high probability all nodes in $\mcal{U}$ receive a non-zero linear combination of $a$ and $b$, and the subspace spanned by their reception has dimension two. Further, if node $\msf{u}\in\mcal{U}$ is $\msf{s}_1$-reachable, then its reception has a non-zero coefficient of $\msf{s}_1$'s symbol $a$ with high probability.
\end{lemma}

\begin{lemma}[Rank Conservation]\label{lem_Rank}
Suppose $\mcal{U}$ and $\mcal{V}$ are the first and the last layers of a
linear deterministic network and each node $\msf{u}\in\mcal{U}$
possesses a linear combination of the symbols $a,b$ given by $\lambda_{\msf{u}}\cdot
a+\mu_{\msf{u}}\cdot b.$ Suppose
\begin{itemize}
\item each node in $\mcal{U}$ can reach some node in $\mcal{V},$
\item $\mincut\lp \mcal{U};\mcal{V}\rp\geq 2,$
\item for each $\msf{u}\in\mcal{U},$ we have $\lambda_{\msf{u}},\mu_{\msf{u}}$ not
  both 0,
\item the $|\mcal{U}|\times 2$ matrix with rows given by
  $\begin{bmatrix}\lambda_{\msf{u}} & \mu_{\msf{u}}\end{bmatrix}$ for each $\msf{u}\in\mcal{U}$
  has full rank (i.e. rank 2).
\end{itemize}
If all nodes in the network perform RLC, then nodes in
$\mcal{V}$ can collectively decode both the symbols $a$ and $b$ with
high probability.
\end{lemma}
{\flushleft \bf Case 2:} $\mincut\lp \msf{s}_1,\msf{s}_2; \mcal{P}^{\msf{s}_2}(\msf{v}_1^*)\rp = 1$ and $\mcal{K}_{\mcal{G}_{12}}^{\msf{s}_1}\lp\msf{u}_{21}\rp$ does not form an $\lp \msf{s}_1; \msf{d}_1\rp$-vertex-cut in $\mcal{G}_{12}$\par
In this case, since $\mincut\lp \msf{s}_1,\msf{s}_2; \mcal{P}^{\msf{s}_2}(\msf{v}_1^*)\rp = 1$, effectively they receive only one linear equation of $a$ and $b$. 
Since $\mcal{P}^{\msf{s}_2}(\msf{v}_1^*)$ can be reached by $\msf{s}_2$, the coefficient of $b$ in this linear equation is non-zero in general.
Hence we need to arrange their transmission so that $\sum_{\msf{w}\in\mcal{P}^{\msf{s}_2}(\msf{v}_1^*)} X_{\msf{w}} = 0$, that is,
\begin{align}\label{eq_neutralize_2}
\sum_{\msf{w}\in\mcal{P}^{\msf{s}_2}(\msf{v}_1^*)} \alpha_{\msf{w}}\beta_{\msf{w},\msf{s}_1}
 =
\sum_{\msf{w}\in\mcal{P}^{\msf{s}_2}(\msf{v}_1^*)} \alpha_{\msf{w}}\beta_{\msf{w},\msf{s}_2}
 = 0, 
\end{align}
Since $\mincut\lp \msf{s}_1,\msf{s}_2; \mcal{P}(\msf{v}_1^*)\rp=2$, $\msf{v}_1^*$ must have some $\msf{s}_1$-only-reachable parents. Therefore $\msf{v}_1^*$ can decode $a$.

With such special operation in $\mcal{P}^{\msf{s}_2}(\msf{v}_1^*)$, effectively we are in the induced graph $\mcal{G}_{12}$. In other words, any linear coding scheme in the induced graph $\mcal{G}_{12}$ can be translated to a linear coding scheme in $\mcal{G}$ satisfying the neutralization criterion \eqref{eq_neutralize_2}, in the sense that the reception of $\msf{d}_i$ remains the same in both schemes, for $i=1,2$. In $\mcal{G}_{12}$, note that $\msf{d}_1$ can only be reached by $\msf{s}_1$ but not $\msf{s}_2$. Hence by Lemma~\ref{k_1^*=0}, as long as the critical node for destination $\msf{d}_2$ in $\mcal{G}_{12}$, $\msf{u}_{21}$, is not omniscient, $(1,1)$ is achievable. $\msf{u}_{21}$ is not omniscient in $\mcal{G}_{12}$ by the assumption of this case.




{\flushleft \bf Case 3:} 
$\mincut\lp \msf{s}_1,\msf{s}_2; \mcal{P}^{\msf{s}_2}(\msf{v}_1^*)\rp = 1$ and $\mcal{K}_{\mcal{G}_{12}}^{\msf{s}_1}\lp\msf{u}_{21}\rp$ forms an $\lp \msf{s}_1; \msf{d}_1\rp$-vertex-cut in $\mcal{G}_{12}$ and $\mcal{K}^{\msf{s}_2}(\msf{w}_{12})$ does not form an $(\msf{s}_2;\msf{d}_2)$-vertex-cut in $\mcal{G}$\par

In this case the idea is to enable $\msf{w}_{12}$ to decode user 1's symbol $a$ while keeping user 2's flow to $\msf{v}_2^*$, making use of the fact that $\mcal{K}^{\msf{s}_2}(\msf{w}_{12})$ does not form a $(\msf{s}_2;\msf{d}_2)$-vertex-cut in $\mcal{G}$. Effectively we impose the neutralization criterion on $\msf{w}_{12}$ instead of $\msf{v}_1^*$, and carry out the special coding operation at $\mcal{P}^{\msf{s}_2}(\msf{w}_{12})$ instead of $\mcal{P}^{\msf{s}_2}(\msf{v}_1^*)$.

As for user 1, obviously $\msf{w}_{12}\ne\msf{s}_2$, and hence it can be reached by $\msf{s}_1$ due to the definition of critical nodes. Since $\mincut\lp \msf{s}_1,\msf{s}_2; \mcal{P}(\msf{w}_{12})\rp = 2$, we can enable $\msf{w}_{12}$ to decode $a$ by satisfying the following neutralization condition:
\begin{align}\label{eq_neutralize_3}
\sum_{\msf{w}\in\mcal{P}(\msf{w}_{12})} \alpha_{\msf{w}}\beta_{\msf{w},\msf{s}_1}
 = \beta_{\mcal{P}(\msf{w}_{12}),\msf{s}_1},\ 
\sum_{\msf{w}\in\mcal{P}(\msf{w}_{12})} \alpha_{\msf{w}}\beta_{\msf{w},\msf{s}_2}
 = 0.
\end{align}
Once $\msf{w}_{12}$ decodes $a$, it simply sends out a scaled copy of $a$. With all other nodes performing RLC up to layer $\mcal{L}_{k_1^*-1}$ (including $\mcal{P}^{\msf{s}_2}(\msf{v}_1^*)$), critical node $\msf{v}_1^*$ can decode $a$.
 
As for user 2, note that depending on the value of $\mincut\lp\msf{s}_1,\msf{s}_2;\mcal{P}^{\msf{s}_2}(\msf{w}_{12})\rp$ being $2$ or $1$, $\beta_{\mcal{P}(\msf{w}_{12}),\msf{s}_1}$ is either non-zero or zero. If $\mincut\lp\msf{s}_1,\msf{s}_2;\mcal{P}^{\msf{s}_2}(\msf{w}_{12})\rp=2$, we use the two lemmas, Lemma~\ref{lem_Reception} and \ref{lem_Rank}, in the first case to show that the parents of $\msf{v}_2^*$ can recover both user's symbols with high probability under the special operation at $\mcal{P}^{\msf{s}_2}(\msf{w}_{12})$. 
If $\mincut\lp\msf{s}_1,\msf{s}_2;\mcal{P}^{\msf{s}_2}(\msf{w}_{12})\rp=1$, we construct \emph{another} induced graph $\mcal{G}'_{12}$ to capture the constraints that such special coding lays on the reception of other nodes in the same layer as $\msf{w}_{12}$, which is similar to $\mcal{G}_{12}$ in the second case. In $\mcal{G}'_{12}$, the critical node for user 2 may no longer be $\msf{v}_2^*$, as the min-cut value from the sources $\{\msf{s}_1,\msf{s}_2\}$ to the parents of $\msf{v}_1^*$ may drop to $1$. Note that as in $\mcal{G}_{12}$, destination $\msf{d}_1$ is now $\msf{s}_1$-only-reachable in $\mcal{G}'_{12}$. Hence we only need the new critical node for destination $\msf{d}_2$ is not omniscient in $\mcal{G}'_{12}$. 

The following lemma guarantees it in this case.
\begin{lemma}\label{claim_Special}
In this case (Case 3) when $\mincut\lp\msf{s}_1,\msf{s}_2;\mcal{P}^{\msf{s}_2}(\msf{w}_{12})\rp=1$, for all possible $\mcal{G}'_{12}$, the $\msf{s}_1$-clones of $\pmc_{\mcal{G}'_{12}}\lp \msf{d}_2\rp$ do not form a $(\msf{s}_1;\msf{d}_1)$-vertex-cut.
\end{lemma}

Combining the above three cases, we complete the proof for the claim and the $(1,1)$-achievability.
\end{proof}


\subsubsection{$(1,1/2)$-Achievability}
For the $(1,1/2)$-achievability we need to prove the following claim
\begin{claim}
$\textrm{P}^{(12)} \implies (1,1/2)$ is achievable.
\end{claim}


\begin{proof}
Consider two cases, distinguishing whether $\msf{v}_2^*$ has parents from the cloud or not.
{\flushleft 1)} $\mcal{P}(\msf{v}_2^*)\cap \cloud_1 \ne \emptyset$:
Under the condition that $\mcal{P}(\msf{v}_2^*)\cap \cloud_1 \ne \emptyset$, we know that in $\mcal{G}_{12}$ the critical node for $\msf{d}_2$ is still $\msf{v}_2^*$, as $\mincut_{\mcal{G}_{12}}\lp\msf{s}_1,\msf{s}_2; \mcal{P}_{\mcal{G}_{12}}(\msf{v}_2^*)\rp=2$. This is because nodes in the cloud $\mcal{C}_1$ become $\msf{s}_1$-only-reachable in $\mcal{G}_{12}$ while some nodes in $\mcal{P}_{\mcal{G}_{12}}(\msf{v}_2^*)$ are $\msf{s}_2$-reachable in $\mcal{G}_{12}$. $\textrm{P}^{(12)}$ implies that the $\msf{s}_1$-clones of $\msf{v}_2^*$ in $\mcal{G}_{12}$ becomes a $(\msf{s}_1;\msf{d}_1)$-vertex-cut. Therefore, some nodes in $\mcal{P}(\msf{v}_2^*)$ must be dropped in generating $\mcal{G}_{12}$ (as they cannot be reached by either one of the sources), and $\mcal{P}_{\mcal{G}_{12}}(\msf{v}_2^*)\ne\mcal{P}(\msf{v}_2^*)$. 

We aim to deliver two symbols $a_1,a_2$ for user 1 and one symbol $b$ for user 2 over two symbol time slots. Symbols are drawn from the extension field $\mbb{F}_{2^r}$. In the first time slot we do RLC with $\msf{s}_1$ transmitting $a_1$ and $\msf{s}_2$ transmitting $b$, up to layer $\mcal{L}_{k_1^*-1}$. Then we arrange the transmission of $\mcal{P}^{\msf{s}_2}(\msf{v}_1^*)$ so that their aggregate becomes zero, as in Case 2. $\msf{v}_1^*$ can hence decode $a_1$, and transmit a scaled version of it. The rest of the nodes keep performing RLC. It is as if the communication is over the induced graph $\mcal{G}_{12}$, and effectively nodes in $\mcal{P}(\msf{v}_2^*)\setminus\mcal{P}_{\mcal{G}_{12}}(\msf{v}_2^*)$ will receive nothing. As the $\msf{s}_1$-clones of $\msf{v}_2^*$ form a $(\msf{s}_1;\msf{d}_1)$-vertex-cut in $\mcal{G}_{12}$, $\msf{v}_2^*$ will receive a linear equation of $a_1$ and $b$, and both symbols have non-zero coefficients. Therefore, $\msf{d}_1$ can decode $a_1$ in the first time slot.

In the second time slot, we do RLC with $\msf{s}_1$ transmitting $a_2$ and $\msf{s}_2$ transmitting $b$, up to the layer right before $\msf{w}_{12}$. For those nodes in $\mcal{K}(\msf{w}_{12})$ that can reach $\mcal{P}^{\msf{s}_2}(\msf{v}_1^*)$, instead of scaling their reception and transmitting it, they replace their reception by a linear combination of $a_1$ and $a_2$. This linear combination is obtained by zero-forcing $b$ using the reception of the first and the second time slot:
\begin{align*}
&\beta_{\msf{w}_{12},\msf{s}_1}^{(1)}\cdot a_1 + \beta_{\msf{w}_{12},\msf{s}_2}^{(1)} \cdot b;\ \beta_{\msf{w}_{12},\msf{s}_1}^{(2)}\cdot a_2 + \beta_{\msf{w}_{12},\msf{s}_2}^{(2)} \cdot b\\
&\overset{\rm{ZF}}{\Longrightarrow}\ \beta_{\msf{w}_{12},\msf{s}_1}^{(2)}\cdot a_2 + \frac{\beta_{\msf{w}_{12},\msf{s}_1}^{(1)}}{\beta_{\msf{w}_{12},\msf{s}_2}^{(1)}}\cdot \beta_{\msf{w}_{12},\msf{s}_2}^{(2)} \cdot a_1.
\end{align*}
The rest of the nodes perform RLC up to layer $\mcal{L}_{k_1^*}$. Since $\msf{v}_1^*$ already obtains $a_1$ in the first time slot and it receives a linear combination of $a_1,a_2$ with non-zero $a_2$-coefficient in the second time slot, it can decode $a_2$. Onwards it transmits a scaled copy of $a_2$, while other nodes perform RLC. The nodes in $\mcal{P}(\msf{v}_2^*)\setminus\mcal{P}_{\mcal{G}_{12}}(\msf{v}_2^*)$, unlike in the first time slot, will receive a linear combination of $a_1,a_2$, which is a scaled version of that transmitted by $\msf{w}_{12}$. Hence, we can arrnage the transmission of $\mcal{P}(\msf{v}_2^*)\setminus\mcal{P}_{\mcal{G}_{12}}(\msf{v}_2^*)$ and $\mcal{P}(\msf{v}_2^*)\cap \cloud_1$ so that $\msf{v}_2^*$ can decode $a_1$. Therefore, using the reception from the first time slot, $\msf{v}_2^*$ can decode $b$.

{\flushleft 2)} $\mcal{P}(\msf{v}_2^*)\cap \cloud_1 = \emptyset$:
We aim to deliver two symbols $a_1$ and $a_2$ for user 1 and one symbol $b$ for user 2 over two symbol time slots. Again the symbols are drawn from the extension field $\mbb{F}_{2^r}$. In the first time slot, we do RLC with $\msf{s}_1$ transmitting $a_1$ and $\msf{s}_2$ transmitting $b$, up to layer $\mcal{L}_{k_1^*-1}$. Then we arrange the transmission of $\mcal{P}^{\msf{s}_2}(\msf{v}_1^*)$ so that their aggregate becomes zero, as in Case 2. It is as if the communication is over the induced graph $\mcal{G}_{12}$. Since $\msf{u}_{21}$ is the critical node for the parents of $\msf{v}_2^*$ in $\mcal{G}_{12}$, in the first time slot the they effectively receive only one equation
\begin{align*}
\beta_{\msf{u}_{21},\msf{s}_1}^{(1)}\cdot a_1 + \beta_{\msf{u}_{21},\msf{s}_2}^{(1)} \cdot b,
\end{align*}
where $\beta_{\msf{u}_{21},\msf{s}_1}^{(1)}$ is non-zero with high probability since $\mcal{K}^{\msf{s}_1}_{\mcal{G}_{12}}\lp\msf{u}_{21}\rp$ forms a $\lp \msf{s}_1; \msf{d}_1\rp$-vertex-cut in $\mcal{G}_{12}$ and hence $\msf{u}_{21}$ must be reachable by $\msf{s}_1$. $\beta_{\msf{u}_{21},\msf{s}_2}^{(1)}$ is non-zero with high probability since $\msf{v}_1^*$ is not omniscient and hence $\msf{s}_2$ must be able to reach $\msf{v}_2^*$ in $\mcal{G}_{12}$.

In the second time slot, we do RLC with $\msf{s}_1$ transmitting $a_2$ and $\msf{s}_2$ transmitting $b$, up to the layer right before $\msf{w}_{12}$. For those nodes in $\mcal{K}(\msf{w}_{12})$ that can reach $\mcal{P}^{\msf{s}_2}(\msf{v}_1^*)$, instead of scaling their reception and transmitting it, they replace their reception by a linear combination of $a_1$ and $a_2$. This linear combination is obtained by zero-forcing $b$ using the reception of the first and the second time slot:
\begin{align*}
&\beta_{\msf{w}_{12},\msf{s}_1}^{(1)}\cdot a_1 + \beta_{\msf{w}_{12},\msf{s}_2}^{(1)} \cdot b;\ \beta_{\msf{w}_{12},\msf{s}_1}^{(2)}\cdot a_2 + \beta_{\msf{w}_{12},\msf{s}_2}^{(2)} \cdot b\\
&\overset{\rm{ZF}}{\Longrightarrow}\ \beta_{\msf{w}_{12},\msf{s}_1}^{(2)}\cdot a_2 + \frac{\beta_{\msf{w}_{12},\msf{s}_1}^{(1)}}{\beta_{\msf{w}_{12},\msf{s}_2}^{(1)}}\cdot \beta_{\msf{w}_{12},\msf{s}_2}^{(2)} \cdot a_1.
\end{align*}
The rest of the nodes remain doing RLC, up to the layer right before $\msf{v}_2^*$. In the second time slot, $\msf{v}_2^*$'s parents receive at least two linear equations in $\{a_1, a_2, b\}$. Pick two nodes $\msf{u},\msf{w} \in \mcal{P}(\msf{v}_2^*)$ such that $\mincut\lp\msf{s}_1,\msf{s}_2;\msf{u},\msf{w}\rp=2$. Let their reception be
\begin{align*}
&\beta_{\msf{u},a_1}\cdot a_1 + \beta_{\msf{u},a_2}\cdot a_2 + \beta_{\msf{u},b}\cdot b,\\
&\beta_{\msf{w},a_1}\cdot a_1 + \beta_{\msf{w},a_2}\cdot a_2 + \beta_{\msf{w},b}\cdot b,
\end{align*}
respectively.
We shall show that the following determinant
\begin{align}\label{eq_det}
&\left | \begin{array}{ccc}
\beta_{\msf{u}_{21},\msf{s}_1}^{(1)} & 0 & \beta_{\msf{u}_{21},\msf{s}_2}^{(1)}\\
\beta_{\msf{u},a_1} & \beta_{\msf{u},a_2} & \beta_{\msf{u},b}\\
\beta_{\msf{w},a_1} & \beta_{\msf{w},a_2} & \beta_{\msf{w},b}
\end{array}\right | \\
&= 
\beta_{\msf{u}_{21},\msf{s}_2}^{(1)}\begin{vmatrix}
\beta_{\msf{u},a_1} & \beta_{\msf{u},a_2}\\
\beta_{\msf{w},a_1} & \beta_{\msf{w},a_2}
\end{vmatrix} + 
\beta_{\msf{u}_{21},\msf{s}_1}^{(1)}\begin{vmatrix}
\beta_{\msf{u},a_2} & \beta_{\msf{u},b}\\
\beta_{\msf{w},a_2} & \beta_{\msf{w},b}
\end{vmatrix} \notag
\end{align}
is non-zero with high probability.

Note that in the second time slot, we choose the scaling coefficients $\alpha$'s for all nodes up to the layer right before $\msf{v}_2^*$ in the same way as RLC. The only difference from RLC is that at the nodes in $\mcal{K}(\msf{w}_{12})$ that can reach $\mcal{P}^{\msf{s}_2}(\msf{v}_1^*)$, the term scaled and transmitted is replaced by the zero-forced output $\beta_{\msf{w}_{12},\msf{s}_1}^{(2)}\cdot a_2 + \frac{\beta_{\msf{w}_{12},\msf{s}_1}^{(1)}}{\beta_{\msf{w}_{12},\msf{s}_2}^{(1)}}\cdot \beta_{\msf{w}_{12},\msf{s}_2}^{(2)} \cdot a_1$. Suppose we do RLC, then $\msf{u}$ and $\msf{w}$ will receive
\begin{align*}
&\beta^{(2)}_{\msf{u},\msf{s}_1}\cdot a_2 + \beta^{(2)}_{\msf{u},\msf{s}_2}\cdot b,\text{ and }
\beta^{(2)}_{\msf{w},\msf{s}_1}\cdot a_2 + \beta^{(2)}_{\msf{w},\msf{s}_2}\cdot b
\end{align*}
respectively, where 
$D^{(2)} := \begin{vmatrix}
\beta^{(2)}_{\msf{u},\msf{s}_1} & \beta^{(2)}_{\msf{u},\msf{s}_2}\\
\beta^{(2)}_{\msf{w},\msf{s}_1} & \beta^{(2)}_{\msf{w},\msf{s}_2}
\end{vmatrix} \ne 0
$
with high probability since $\mincut\lp\msf{s}_1,\msf{s}_2;\msf{u},\msf{w}\rp=2$. As pointed out above, from the connection of the scheme to RLC, we see
\begin{align*}
\begin{vmatrix}
\beta_{\msf{u},a_1} & \beta_{\msf{u},a_2}\\
\beta_{\msf{w},a_1} & \beta_{\msf{w},a_2}
\end{vmatrix}
= \frac{\beta_{\msf{w}_{12},\msf{s}_1}^{(1)}}{\beta_{\msf{w}_{12},\msf{s}_2}^{(1)}}\cdot D^{(2)}_Z, \quad
\begin{vmatrix}
\beta_{\msf{u},a_2} & \beta_{\msf{u},b}\\
\beta_{\msf{w},a_2} & \beta_{\msf{w},b}
\end{vmatrix}
= D^{(2)}_R,
\end{align*}
where 
$$
D^{(2)}_Z := \begin{vmatrix}
\beta^{(2)}_{\msf{u},\msf{s}_1} & \lb\beta^{(2)}_{\msf{u},\msf{s}_2}\rb_Z\\
\beta^{(2)}_{\msf{w},\msf{s}_1} & \lb\beta^{(2)}_{\msf{w},\msf{s}_2}\rb_Z
\end{vmatrix},
$$ and $$
D^{(2)}_R := \begin{vmatrix}
\beta^{(2)}_{\msf{u},\msf{s}_1} & \lb\beta^{(2)}_{\msf{u},\msf{s}_2}\rb_R\\
\beta^{(2)}_{\msf{w},\msf{s}_1} & \lb\beta^{(2)}_{\msf{w},\msf{s}_2}\rb_R
\end{vmatrix}.
$$
Here $\lb\beta^{(2)}_{\msf{u},\msf{s}_2}\rb_R$ denotes the coefficient of $b$ that $\msf{u}$ receives under a \emph{virtual} RLC with the same coding operation as regular RLC except that nodes in $\mcal{K}(\msf{w}_{12})$ that can reach $\mcal{P}^{\msf{s}_2}(\msf{v}_1^*)$ (call this set $\mcal{Z}$) are transmitting scaled copies of $a$ (with the same scaling coefficient as in regular RLC) instead of a linear combination of $a,b$. $\lb\beta^{(2)}_{\msf{u},\msf{s}_2}\rb_Z$ denotes the coefficient of $b$ that $\msf{u}$ receives under a \emph{virtual} RLC with the same coding operation as regular RLC except that the $\msf{s}_2$-reachable predecessors of $\msf{u}$ in the same layer as $\msf{w}_{12}$ other than $\mcal{Z}$ are transmitting scaled copies of the $a$-components in their reception (with the same scaling coefficient). Note that if there is no $\msf{s}_2$-reachable predecessor of $\msf{u}$ in $\mcal{Z}$, then $\lb\beta^{(2)}_{\msf{u},\msf{s}_2}\rb_Z = 0$.
We have $D^{(2)}_Z+D^{(2)}_R = D^{(2)}$.
The determinant in \eqref{eq_det} equals to zero if and only if
\begin{align*}
\beta_{\msf{u}_{21},\msf{s}_2}^{(1)}\beta_{\msf{w}_{12},\msf{s}_1}^{(1)}D^{(2)}_Z + \beta_{\msf{u}_{21},\msf{s}_1}^{(1)}\beta_{\msf{w}_{12},\msf{s}_2}^{(1)}D^{(2)}_R = 0.
\end{align*}

Suppose $\begin{vmatrix} \beta_{\msf{w}_{12},\msf{s}_1}^{(1)} & \beta_{\msf{w}_{12},\msf{s}_2}^{(1)} \\ \beta_{\msf{u}_{21},\msf{s}_1}^{(1)} & \beta_{\msf{u}_{21},\msf{s}_2}^{(1)}\end{vmatrix}$ is a zero polynomial, then we are done since $D^{(2)}\ne 0$ with high probability.

Note that $D^{(2)}_Z$ and $D^{(2)}_R$ cannot simultaneously be zero with high probability, as their sum is non-zero with high probability. First assume that $D^{(2)}_Z\ne 0$. The determinant in \eqref{eq_det} equals to zero if and only if
\begin{align*}
\frac{D^{(2)}_R}{D^{(2)}_Z} = \frac{\beta_{\msf{u}_{21},\msf{s}_2}^{(1)}\beta_{\msf{w}_{12},\msf{s}_1}^{(1)}}{\beta_{\msf{u}_{21},\msf{s}_1}^{(1)}\beta_{\msf{w}_{12},\msf{s}_2}^{(1)}}.
\end{align*}
RHS and LHS are independent. We only need to consider the case $\begin{vmatrix} \beta_{\msf{w}_{12},\msf{s}_1}^{(1)} & \beta_{\msf{w}_{12},\msf{s}_2}^{(1)} \\ \beta_{\msf{u}_{21},\msf{s}_1}^{(1)} & \beta_{\msf{u}_{21},\msf{s}_2}^{(1)}\end{vmatrix}$ is not a zero polynomial. 
The probability distribution of RHS is ``almost" uniform and there is \emph{no} particular value at which it has a non-vanishing probability (see Lemma~\ref{spreads_out} in Appendix~\ref{app_Pf_NonzeroDet}). Hence the above equality happens with vanishing probability.

Similar conclusion can be drawn in the case $D^{(2)}_R\ne 0$.
\end{proof}

\subsubsection{$(1/2,1)$-Achievability}
For the $(1/2,1)$-achievability, we argue that if there is no omniscient node, then $(1/2,1)$ is achievable.
We shall use nodes reachable from $\msf{u}_2$ in $\mcal{P}(\msf{v}_2^*)$ to provide user 2's symbols. Define the collection of these nodes by $\mcal{S}_{k_2^*-1}(\msf{u}_2)$. Consider the following two cases.

{\flushleft 1)} $\mcal{P}(\msf{v}_2^*)\cap \cloud_1 \ne \emptyset$:
If $\textrm{Q}^{(12)}$ is violated, $(1,1)$ can be achieved and so can $(1/2,1)$. Hence we focus on the case in which $\textrm{Q}^{(12)}$ is satisfied. Under the condition that $\mcal{P}(\msf{v}_2^*)\cap \cloud_1 \ne \emptyset$, from the analysis of the previous case we know that some nodes in $\mcal{P}(\msf{v}_2^*)$ must be dropped in generating $\mcal{G}_{12}$, and $\mcal{P}_{\mcal{G}_{12}}(\msf{v}_2^*)\ne\mcal{P}(\msf{v}_2^*)$.


We aim to deliver one symbol $a$ for user 1 and two symbols $b_1, b_2$ for user 2 over two time slots. In the first time slot, all nodes up to layer $k_2^*-1$ perform RLC with $\msf{s}_1$ transmitting a scaled copy of $a,$ and $\msf{s}_2$ transmitting a scaled copy of $b_1$, except that nodes in $\mcal{P}^{\msf{s}_2}(\msf{v}_1^*)$ perform special linear coding to make sure their aggregate transmission is zero. Hence effectively we are in $\mcal{G}_{12}$, and the nodes in $\mcal{P}_{\mcal{G}_{12}}^{\msf{s}_1}(\msf{v}_2^*)$, which form a $(\msf{s}_1;\msf{d}_1)$-vertex-cut in $\mcal{G}_{12}$ and therefore lie in the cloud $\mcal{C}_1$, can decode $a$. Since $\mincut_{\mcal{G}_{12}}\lp\msf{s}_1,\msf{s}_2; \mcal{P}_{\mcal{G}_{12}}(\msf{v}_2^*)\rp=2$, we can arrange the transmission of $\mcal{P}_{\mcal{G}_{12}}(\msf{v}_2^*)$ so that $\msf{v}_2^*$ can decode $b_1$ and so can $\msf{d}_2$. But $\msf{d}_1$ will receive nothing, as $\mcal{K}_{\mcal{G}_{12}}^{\msf{s}_1}(\msf{v}_2^*)$ is a $(\msf{s}_1;\msf{d}_1)$-vertex-cut in $\mcal{G}_{12}$. In the second time slot, all nodes up to layer $k_2^*-1$ perform RLC with $\msf{s}_1$ transmitting a scaled copy of $a,$ and $\msf{s}_2$ transmitting a scaled copy of $b_2$. This time the nodes in $\mcal{P}(\msf{v}_2^*)\setminus\mcal{P}_{\mcal{G}_{12}}(\msf{v}_2^*)$ will receive a non-triavial linear combination of $a$ and $b_2$ with a non-zero $a$-coefficient. Then we let nodes in $\mcal{P}_{\mcal{G}_{12}}(\msf{v}_2^*)$ transmit a scaled copy of $a$ by choosing their scaling coefficients uniformly and independently, while using nodes in $\mcal{P}(\msf{v}_2^*)\setminus\mcal{P}_{\mcal{G}_{12}}(\msf{v}_2^*)$ to neutralize the symbol $a$ in the reception of $\msf{v}_2^*$ and obtain a clean copy of $\msf{b}_2$. Hence $\msf{d}_1$ can decode $a$, and $\msf{d}_2$ can decode $b_2$ in the second time slot.

{\flushleft 2)} $\mcal{P}(\msf{v}_2^*)\cap \cloud_1 = \emptyset$:
We aim to prove that $(1/2,1)$ is achievable in this case. User 1 has one symbol $a$ and user 2 has two symbols $b_1,b_2$ to be delivered over two time slots.

In the first time slot, all nodes up to layer $k_2^*-1$ perform RLC with $\msf{s}_1$ transmitting a scaled copy of $a,$ and $\msf{s}_2$ transmitting a scaled copy of $b_1.$ Note that because $\mincut\lp \msf{s}_1,\msf{s}_2; \mcal{P}(\msf{v}_2^*)\rp = 2,$ we have that $\mcal{P}(\msf{v}_2^*)$ can collectively decode both $a$ and $b_1$ due to Lemma~\ref{lem_two_source}(a). In the second time slot, all nodes up to layer $k_1^*-2$ perform RLC with $\msf{s}_1$ sending $a$ and $\msf{s}_2$ sending $b_2$. 
Due to Lemma~\ref{lem_two_source}(a), we can arrange the transmission of $\mcal{P}(\msf{v}_1^*)$ so that $\msf{v}_1^*$ receives only $a$ in the second time slot, since $\mincut\lp \msf{s}_1,\msf{s}_2; \mcal{P}(\msf{v}_1^*)\rp = 2$.
As $\emptyset\subsetneq \mcal{P}^{\msf{s}_2}(\msf{u}_2)\ne \mcal{P}^{\msf{s}_2}(\msf{v}_1^*),$ $\msf{u}_2$ receives a linear combination with a non-zero coefficient of user 2's symbol due to Lemma~\ref{lem_special}. Further, all nodes perform RLC up to layer $k_2^*-1.$ As $\msf{u}_2$ has a path to $\mcal{P}(\msf{v}_2^*),$ some node in $\mcal{P}(\msf{v}_2^*)$ receives a linear combination of the three symbols with a non-zero coefficient for $b_2.$ Thus, $\mcal{P}(\msf{v}_2^*)$ can collectively decode all three symbols $a, b_1,b_2.$ Since this decoding is a linear operation, these nodes can arrange their transmissions so as to form $b_1$ and $b_2$ at $\msf{v}_2^*$'s reception in first and second time slots respectively. All nodes in layer $k_2^*$ onwards perform RLC with no mixing across time slots. Thus, $\msf{d}_2$ can recover both $b_1$ and $b_2.$ As nodes in $\mcal{L}_{k_2^*-1}\cap \cloud_1$ perform RLC with no mixing across time slots, destination $\msf{d}_1$ can recover both the symbols $a$ and $b_1.$


\section{Proof of Outer Bounds}\label{sec_OutBd}
\subsection{Outer Bound on $R_1+R_2$: the Omniscient Bound}
We show that if a node $\msf{v}$ is omniscient, then it can decode both user's messages and hence, the achievable sum rate is upper bounded by 1. This explains the motivation for the name. Let $\msf{v}$ be omniscient and satisfy condition (A) in the definition of omniscient nodes: $\mcal{K}(\msf{v})$ is a $\lp\msf{s}_1,\msf{s}_2; \msf{d}_1\rp$-vertex-cut and $\mcal{K}^{\msf{s}_2}(\msf{v})$ is a $\lp \msf{s}_2; \msf{d}_2\rp$-vertex-cut. Since $\mcal{K}(\msf{v})$ is a $\lp\msf{s}_1,\msf{s}_2; \msf{d}_1\rp$-vertex-cut, the reception of the destination $\msf{d}_1$ is a function of the reception of $\msf{v}.$ This means $\msf{v}$ can decode the message of $\msf{s}_1.$ The reception of each node in $\mcal{K}^{\msf{s_2}}(\msf{v})$ is some function of the reception of node $\msf{v}$ and the transmission of $\msf{s}_1.$ Since $\msf{v}$ can now recover the transmission of $\msf{s}_1,$ and since $\mcal{K}^{\msf{s_2}}(\msf{v})$ forms a $\lp \msf{s}_2; \msf{d}_2\rp$-vertex-cut, $\msf{v}$ can recover the reception of $\msf{d}_2,$ and thus, also the message of $\msf{s}_2.$ 
We leave the formal proof of this outer bound in the appendix.

\subsection{Outer Bounds on $2R_1+R_2$ and $R_1+2R_2$}
We want to show that if the condition $\textrm{T}^{(12)}$ is satisfied, then $2R_1+R_2 \le 2$ for any achievable $(R_1,R_2)$. We first show the following claim.
\begin{claim}\label{claim_2R1R2}
If there exists random variables $\{Z_1, Z_{21}, Z_{22}\}$ in the network satisfying
{\rm
\begin{itemize}
\item [1)] $H\lp Z_1\rp\le1, H\lp Z_{21}\rp\le 1, H\lp Z_{22}\rp\le1$.
\item [2)] $X_{\msf{s}_1}^N \lra Z_1^N \lra Y_{\msf{d}_1}^N$ and $X_{\msf{s}_2}^N \lra (Z_{21}^N,Z_{22}^N) \lra Y_{\msf{d}_2}^N$
\item [3)] $X_{\msf{s}_1}^N\lra (Z_{21}^N, X_{\msf{s}_2}^N) \lra Y_{\msf{d}_1}^N$
\item [4)] $H\lp Z_1^N| X_{\msf{s}_1}^N\rp \ge H\lp Z_{22}^N\rp$
\item [5)] $Z_{22}^N$ is a function of $X_{\msf{s}_2}^N$
\end{itemize}
}
then $2R_1+R_2 \le 2$ for any achievable $(R_1,R_2)$.
\end{claim}
Proof is detailed in the appendix.

We shall use the above claim to complete the proof of the outer bound $2R_1+R_2\le 2$. We set
$
Z_1 := Y_{\msf{v}^*_1}$, $Z_{21} := Y_{\msf{u}_{21}}$, $Z_{22} := \sum_{\msf{w}\in \mcal{P}^{\msf{s}_2}\lp \msf{v}^*_1\rp} X_{\msf{w}}$.
\begin{itemize}
\item
Hence by the definition of the channels, condition 1) of the claim is satisfied. 
\item
By the definition of $\msf{v}^*_1$, we see that the first Markov chain in condition 2) is satisfied. By condition $\textrm{T}^{(12)}_3$ and the definition of the induced graph $\mcal{G}_{12}$ we see that the second Markov chain is also satisfied. Hence, condition 2) is satisfied. 
\item
By conditions $\textrm{T}^{(12)}_3$ and $\textrm{T}^{(12)}_4$, we see that the Markov chain in condition 3) is satisfied.
\item
Condition 4) is satisfied with equality due to the definition of $Z_{22}$ and condition $\textrm{T}^{(12)}_4$.
\item
Condition 5) is satisfied due to the definition of $Z_{22}$ and conditions $\textrm{T}^{(12)}_2$ and $\textrm{T}^{(12)}_4$.
\end{itemize}


Similarly, if the condition $\textrm{T}^{(21)}$ is satisfied, then $R_1+2R_2 \le 2$ by symmetry.

\subsection{Outer Bound on $2R_1+2R_2$}
We want to show that if the condition $\textrm{P}^{(12)}$ is satisfied, then $2R_1+2R_2 \le 3$ for any achievable $(R_1,R_2)$. We first show the following claim.
\begin{claim}\label{claim_2R12R2}
If there exists random variables $\{Z_{11}, Z_{12},Z_{21},Z_{22}\}$ in the network satisfying
\begin{itemize}
\item[1)] $H(Z_{11})\le 1$, $H(Z_{12})\le 1$, $H(Z_{21})\le 1$, $H(Z_{22})\le 1$
\item[2)] $X_{\msf{s}_1}^N \lra Z_{11}^N \lra Y_{\msf{d}_1}^N$ and $X_{\msf{s}_2}^N \lra (Z_{21}^N,Z_{22}^N) \lra Y_{\msf{d}_2}^N$
\item[3)]
$X_{\msf{s}_1}^N \lra (Z_{21}^N, Z_{22}^N, X_{\msf{s}_2}^N) \lra Y_{\msf{d}_1}^N$ and \\$X_{\msf{s}_2}^N \lra (Z_{12}^N,X_{\msf{s}_1}^N) \lra Y_{\msf{d}_2}^N$
\item[4)] $H\lp Z_{11}^N|X_{\msf{s}_1}^N\rp \ge H\lp Z_{22}^N|X_{\msf{s}_1}^N\rp$
\item[5)] $Z_{22}^N$ is a function of $Z_{12}^N$
\end{itemize}
then $2R_1+2R_2 \le 3$ for any achievable $(R_1,R_2)$.
\end{claim}
Proof is detailed in the appendix.

We shall use the above claim to complete the proof of the outer bound $2R_1+2R_2\le 3$. We set
$Z_{11} := Y_{\msf{v}_1}$, $Z_{12} := Y_{\msf{w}_{12}}$, $Z_{21} := Y_{\msf{u}_{21}}$, $Z_{22} := \sum_{\msf{w}\in \mcal{P}^{\msf{s}_2}\lp \msf{v}^*_1\rp} X_{\msf{w}}$.

\begin{itemize}
\item
By the definition of the channels, condition 1) of the claim is satisfied.
\item
By the definition of $\msf{v}^*_1$, we see that the first Markov chain in condition 2) is satisfied. By condition $\textrm{P}^{(12)}_3$ and the definition of the induced graph $\mcal{G}_{12}$ we see that the second Markov chain is also satisfied. Hence, condition 2) is satisfied. 
\item
The first Markov chain in condition 3) is due to condition $\textrm{P}^{(12)}_3$ and the definition of the induced graph $\mcal{G}_{12}$. The second Markov chain is due to condition $\textrm{P}^{(12)}_4$. Hence condition 3) is satisfied.
\item
Condition 4) is satisfied with equality due to the definition of $Z_{22}$. 
\item
Condition 5) is satisfied due to the definition of $Z_{12}, Z_{22}$ and conditions $\textrm{P}^{(12)}_2$ and $\textrm{P}^{(12)}_4$.
\end{itemize}

Similarly, if the condition $\textrm{P}^{(21)}$ is satisfied, then $2R_1+2R_2 \le 3$ by symmetry.

\section{Concluding Remarks}\label{sec_Conclude}

In this paper, we completely characterize the capacity region of two unicast information flows over a layered linear deterministic network with base field $\mbb{F}_2$ under the unit-channel strength assumption. It turns out that when each source can reach its own destination, the capacity region is one of the five: the triangle $\mfrak{T}$, the trapezoids $\mfrak{T}_{12},\mfrak{T}_{21}$, the pentagon $\mfrak{P}$, and the square $\mfrak{S}$. The necessary and sufficient condition for the capacity region to be one of them elucidates when and how the connectivity of the network limits the amount of information deliverable to the destination under the presence of the other interfering information flow.

Our result extends to a more general linear deterministic channel setting where a general matrix in $\mbb{F}_2$ (not necessarily a shift matrix) is associated to each edge in the network. Such generalization is made possible by looking at entries of the receive/transmit vectors, called ``bubbles", and redefining omniscience, clone sets, parents, cuts, etc., for bubbles. This result will be detailed in a later version of this paper.


\bibliographystyle{ieeetr}

\appendices

\section{Proof of Lemmas and Claims}\label{app_PfLemmas}
By the phrase ``with high probability'', we mean a probability that goes to 1 as the size of the field $\mbb{F}_{2^r}$ goes to infinity.

\subsection{Proof of Lemma~\ref {k_1^*=0}}
If $k^*_1=0$ and $k^*_2=1$, then both $\msf{s}_1$ and $\msf{s}_2$ are $\msf{v}^*_2$'s parents, and obviously $\msf{v}^*_2$ is omniscient, violating the assumption. Hence, $k^*_2=0$ or $k^*_2\ge 2$. If $k^*_2=0$, then there is no interference at destination $\msf{d}_1$ from source $\msf{s}_2$ and vice versa. In this case, clearly $(1,1)$ is achievable.

If $k^*_2\ge 2$, we shall show that $(1,1)$ can be achieved provided that there is no omniscient node. Nodes do RLC with $\msf{s}_1$ transmitting a scaled copy of symbol $a\in\mbb{F}_{2^r}$ and $\msf{s}_2$ transmitting a scaled copy of symbol $b\in\mbb{F}_{2^r}$, until layer $\mcal{L}_{k^*_2-1}$. By definitions of $\msf{v}^*_2$ and $\cloud_1$, layer $\mcal{L}_{k^*_2}$ is partitioned by $\mcal{K}(\msf{v}^*_2)$ and $\cloud_1\cap\mcal{L}_{k^*_2}$. Since $\msf{v}^*_2$ is not omniscient, 
\begin{align*}
\exists\ \msf{u}_1\in \cloud_1\cap\mcal{L}_{k^*_2} \textrm{ such that } \mcal{P}^{\msf{s}_1}(\msf{u}_1)( = \mcal{P}(\msf{u}_1)) \ne \mcal{P}^{\msf{s}_1}(\msf{v}^*_2).
\end{align*}
Note that since $\msf{u}_1\in\cloud_1$, $\mcal{P}^{\msf{s}_1}(\msf{u}_1) = \mcal{P}(\msf{u}_1)$. Also note that all nodes in $\cloud_1$ are $\msf{s}_1$-only-reachable. Consider the following two cases:
\begin{enumerate}
\item $\mcal{P}(\msf{u}_1) \setminus \mcal{P}^{\msf{s}_1}(\msf{v}_2^*) \ne\emptyset$. In this case we arrange nodes in  $\mcal{P}(\msf{v}_2^*)$ so that user 2's symbol $b$ can be decoded at $\msf{v}^*_2$. Then use nodes in $\mcal{P}(\msf{u}_1) \setminus \mcal{P}^{\msf{s}_1}(\msf{v}_2^*)$ to provide user 1's symbol $a$ at $\msf{u}_1$ if necessary. $\msf{u}_1$ and $\msf{v}^*_2$ and their successors do RLC.

\item $\mcal{P}(\msf{u}_1) \subsetneq \mcal{P}^{\msf{s}_1}(\msf{v}_2^*)$. In this case we first let nodes in $\mcal{P}(\msf{u}_1)$ do RLC and place user 1's symbol $a$ at $\msf{u}_1$. Then, use nodes in $\mcal{P}^{\msf{s}_1}(\msf{v}^*_2) \setminus \mcal{P}(\msf{u}_1)$ and nodes in $\mcal{P}(\msf{v}^*_2)$ to neutralize user 1's symbol $a$ and place user 2's symbol $b$ at $\msf{v}^*_2$. $\msf{u}_1$ and $\msf{v}^*_2$ and their successors do RLC.
\end{enumerate}
Hence, $(1,1)$ is achievable when $k_1^*=0$.

\subsection{Proof of Lemma~\ref{lem_critical}}
Suppose $i=1.$ Suppose $\mincut\lp \msf{s}_1,\msf{s}_2;\mcal{P}(\msf{v}_1^*)\rp \neq 2.$ It cannot be larger than 2 by definition and it cannot be 0 because $\{\msf{s}_1,\msf{s}_2\}$ has paths to $\mcal{P}(\msf{v}_1^*).$ So, suppose $\mincut\lp \msf{s}_1,\msf{s}_2;\mcal{P}(\msf{v}_1^*)\rp =1.$

Let $\mcal{A}\subseteq\mcal{V}$ be the set of nodes in the graph that can be reached by $\{\msf{s}_1,\msf{s}_2\}$ and can reach $\mcal{P}(\msf{v}_1^*).$ Let $\mcal{G}^\prime$ be the graph induced by nodes in $\mcal{A}$ and for $\mcal{U}\subseteq\mcal{A},$ let $\mincut_{\mcal{G}^\prime}\lp \msf{s}_1,\msf{s}_2;\mcal{U}\rp$ denote the mincut from $\{\msf{s}_1,\msf{s}_2\}$ to $\mcal{U}$ in the graph $\mcal{G}^\prime.$ Then, obviously $\mincut_{\mcal{G}^\prime}\lp \msf{s}_1,\msf{s}_2;\mcal{P}(\msf{v}_1^*)\rp = 1.$

Note that for any partition of the vertices of $\mcal{A}$ into $(B,\mcal{A}\setminus B)$ with $\{\msf{s}_1,\msf{s}_2\}\subseteq B, \mcal{P}(\msf{v}_1^*)\subseteq \mcal{A}\setminus B,$ if there exist nodes in the same layer $\msf{u}_1,\msf{u}_2\in\mcal{A}$ such that $\msf{u}_1\in B$ and $\msf{u}_2\in\mcal{A}\setminus B,$ then the rank of the transfer matrix across the cut $(B,\mcal{A}\setminus B)$ is at least 2. Thus, if there exists a cut $(B,\mcal{A}\setminus B)$ of value $1,$ then the cut must be of the form $B=\lp \cup_{l=0}^t\mcal{L}_l \rp \cap \mcal{A},$ for some $t\geq 0.$ This tells us that if $\msf{u}\in\mcal{A}\cap\mcal{L}_{t+1},$ then $\mcal{K}(\msf{u})\supseteq \mcal{A}\cap\mcal{L}_{t+1},$ so that $\mcal{K}(\msf{u})$ is a $\lp \msf{s}_1,\msf{s}_2;\msf{d}_1 \rp$-vertex-cut violating the definition of critical node $\msf{v}_1^*.$ Hence we complete the proof by contradiction.

\subsection{Proof of Lemma~\ref{lem_critical_omniscient}}
Suppose node $\msf{v}$ is omniscient, say $\mcal{K}(\msf{v})$ is a $\lp\msf{s}_1,\msf{s}_2; \msf{d}_1\rp$-vertex-cut and $\mcal{K}^{\msf{s}_2}(\msf{v})$ is a $\lp \msf{s}_2; \msf{d}_2\rp$-vertex-cut. Suppose $\msf{v}_1^*$ is not omniscient. As $\mcal{K}(\msf{v})$ is a $\lp\msf{s}_1,\msf{s}_2; \msf{d}_1\rp$-vertex-cut, we have that either $\msf{v}\in\mcal{K}(\msf{v}_1^*)$ or that $\msf{v}$ lies in a layer $\mcal{L}_k$ with $k>k_1^*.$ This follows from the definition of the critical node $\msf{v}_1^*.$ In the first case, we automatically have that $\msf{v}_1^*$ is omniscient. So, suppose $\msf{v}$ lies in layer $\mcal{L}_k$ with $k>k_1^*.$ Then, since $\msf{v}_1^*$ is not omniscient, there exists a path from $\msf{s}_2$ to $\msf{d}_2$ with a node $\msf{u}_{k_1^*}$ in layer $\mcal{L}_{k_1^*}$ and a node $\msf{u}_k$ in layer $\mcal{L}_k$ such that $\mcal{P}^{\msf{s}_2}(\msf{v}_1^*)\neq \mcal{P}^{\msf{s}_2}(u_{k_1^*}).$ Since node $\msf{v}$ is omniscient, we must have that $\mcal{P}^{\msf{s}_2}(\msf{v})= \mcal{P}^{\msf{s}_2}(u_k).$ But this is impossible since $\msf{u}_k$ has an $\msf{s}_2$-reachable parent from the path that does not lie in the cloud $\mcal{C}_1$ which contains all the parents of $\msf{v}.$ This contradiction establishes that $\msf{v}_1^*$ must have been omniscient.

\subsection{Proof of Lemma~\ref{lem_PrimaryMinCut}}
First note that if we restrict attention to the induced subgraph $\mcal{G}^\prime$ obtained by deleting all nodes which can either not reach the set of nodes $\mcal{U}$ or cannot be reached by at least one of $\msf{s}_1$ and $\msf{s}_2,$ then the mincut value $\mincut\lp \msf{s}_1,\msf{s}_2; \mcal{U}\rp$ is preserved. Since each node can be reached by at least one of $\msf{s}_1$ or $\msf{s}_2,$ we only have to delete nodes that cannot reach some node in $\mcal{U}.$

Now, we are looking at a graph where the set of nodes in layer $l$ is $\mcal{U}_l$ for $0\leq l<k$ and $\mcal{U}$ for layer $k.$

Consider, for this graph, the set of all vertex bipartitions between $\{\msf{s}_1,\msf{s}_2\}$ and $\mcal{U}$ which yield a transfer matrix of rank 1. All such bipartition cuts must be `vertical', i.e. they are partitions of the form $(\mcal{A},\mcal{A}^c)$ where $\mcal{A}=\cup_{l=0}^r\mcal{U}_l$ for some $l, 0\leq l<k.$ This is because any non-`vertical' cut yields a transfer matrix of rank at least 2. This establishes that the parents sets of all nodes in $\mcal{U}_{l^*}$ are identical in this graph and so, also in the original graph because every node in the new graph $\mcal{G}^\prime$ has the same parent set as in the original graph. This concludes the proof of the lemma.

\subsection{Proof of Lemma~\ref {all_G12s_are_essentially_the_same}}

We are in the scenario where $\mincut \lp \msf{s}_1,\msf{s}_2;\mcal{P}^{\msf{s}_2}(\msf{v}_1^*)\rp=1.$ Consider $\mcal{G}_{12}(\msf{w})$ for some node $\msf{w}\in\mcal{P}^{\msf{s}_2}(\msf{v}_1^*)$ and suppose that there is no omniscient node in $\mcal{G}_{12}(\msf{w}).$ As $\mcal{G}_{12}(\msf{w})$ has no paths from $\msf{s}_2$ to $\msf{d}_1,$ by Lemma~\ref{k_1^*=0}, we can achieve $(1,1)$ in $\mcal{G}_{12}(\msf{w}),$ with high probability, by all nodes except nodes in $\mcal{P}(\msf{v}_2^*)$ performing RLC. Then, with high probability all nodes in $\mcal{P}^{\msf{s}_2}(\msf{v}_1^*)$ receive a non-trivial linear combination and each has a non-zero coefficient of symbol $b$ sent by $\msf{s}_2.$ Take any such $(1,1)$ achieving scheme and any other node $\msf{w}^\prime \in\mcal{P}^{\msf{s}_2}(\msf{v}_1^*).$ Consider $\mcal{G}_{12}(\msf{w}^\prime).$ Note $\msf{w}$ has no outgoing edges in $\mcal{G}_{12}(\msf{w})$ and $\msf{w}^\prime$ has no outgoing edges in $\mcal{G}_{12}(\msf{w}^\prime).$ Let the reception of node $\msf{w}$ and $\msf{w}^\prime$ in $\mcal{G}_{12}(\msf{w})$ be $\beta_{\msf{w},\msf{s}_1}\cdot a+\beta_{\msf{w},\msf{s}_2}\cdot b$ and $\beta_{\msf{w}^\prime,\msf{s}_1}\cdot a+\beta_{\msf{w}^\prime,\msf{s}_2}\cdot b$ respectively, where $\beta_{\msf{w},\msf{s}_2}, \beta_{\msf{w}^\prime,\msf{s}_2}\neq 0.$ Just make all nodes choose the same coefficients in $\mcal{G}_{12}(\msf{w}^\prime)$ as in $\mcal{G}_{12}(\msf{w})$ except for node $\msf{w}^\prime$ which chooses  $\alpha_{\msf{w}^\prime}|_{\mcal{G}_{12}(\msf{w}^\prime)}=\alpha_{\msf{w}}\cdot\frac{\beta_{\msf{w},\msf{s}_2}}{\beta_{\msf{w}^\prime,\msf{s}_2}}.$ Then, the receptions of all nodes will be identical to those in $\mcal{G}_{12}(\msf{w}).$ This achieves $(1,1)$ in $\mcal{G}_{12}(\msf{w}^\prime)$ and hence, there cannot be any omniscient node in this network either by the Omniscient node outer bound.

\subsection{Proof of Lemma~\ref{lem_reachability}}
Without loss of generality let $i=1$. We shall prove this by induction on the layer index where $\msf{u}$ lies. Say $\msf{u}\in\mcal{L}_k$. The node $\msf{u}$ receives $\beta_{\msf{u},\msf{s}_1}\cdot a + \beta_{\msf{u},\msf{s}_2}\cdot b$.

For $k=1$, $\beta_{\msf{u},\msf{s}_1} = \alpha_{\msf{s}_1}\beta_{\msf{s}_1,\msf{s}_1} = \alpha_{\msf{s}_1}$. Since all predecessors of $\msf{u}$ are doing RLC, so does $\msf{s}_1$ and hence $\alpha_{\msf{s}_1}$ is chosen uniformly and randomly over $\mbb{F}_{2^r}$. Therefore, $\Pr\{\beta_{\msf{u},\msf{s}_1} = 0\} = \Pr\{ \alpha_{\msf{s}_1} = 0\} \ra 0$ as $r\ra \infty$.

Suppose for all nodes in $\mcal{L}_{l}$, $l\ge 1$, that are reachable from $\msf{s}_1$ the coefficient of user 1's symbol $a$ is non-zero with high probability. Consider an $\msf{s}_1$-reachable node in $\mcal{L}_{l+1}.$ We have
\begin{align*}
\beta_{\msf{u},\msf{s}_1} = \sum_{\msf{v}\in\mcal{P}(\msf{u})}\alpha_{\msf{v}}\beta_{\msf{v},\msf{s}_1} = \sum_{\msf{v}\in\mcal{P}^{\msf{s}_1}(\msf{u})}\alpha_{\msf{v}}\beta_{\msf{v},\msf{s}_1},
\end{align*}
since for nodes that cannot be reached by $\msf{s}_1$ the coefficient of $a$ is always $0$. Conditioned on a realization of $\{ \beta_{\msf{v},\msf{s}_1}:\ \msf{v}\in\mcal{P}^{\msf{s}_1}(\msf{u})\}$ where they are not all zero, $\beta_{\msf{u},\msf{s}_1}$ is uniformly distributed over $\mbb{F}_{2^r}$ since $\{\alpha_{\msf{v}}|\ \msf{v}\in\mcal{P}^{\msf{s}_1}(\msf{u})\}$ are chosen independently of one another and $\{ \beta_{\msf{v},\msf{s}_1}:\ \msf{v}\in\mcal{P}^{\msf{s}_1}(\msf{u})\}$, and uniformly over $\mbb{F}_{2^r}$. Consequently,
\begin{align*}
\Pr\lbp \beta_{\msf{u},\msf{s}_1}=0 | \{\beta_{\msf{v},\msf{s}_1}:\ \msf{v}\in\mcal{P}^{\msf{s}_1}(\msf{u})\}\rbp \ra 0
\end{align*}
as $r\ra\infty$, if $\{ \beta_{\msf{v},\msf{s}_1}:\ \msf{v}\in\mcal{P}^{\msf{s}_1}(\msf{u})\}$ are not all zeros. By the induction assumption, the probability that they are all zeros also goes to zero as $r\ra\infty$, and so we have $\Pr\lbp \beta_{\msf{u},\msf{s}_1}=0\rbp \ra 0$ as $r\ra\infty$. This completes the proof by induction.

\subsection{Proof of Lemma~\ref{lem_two_source}}
\subsubsection{Proof of Part (a)}
Consider a super-sink $\msf{d}'$ with full access to the reception of all nodes in $\mcal{U}$. Since $\mincut(\msf{s}_1,\msf{s}_2;\mcal{U})=2$, we have $\mincut(\msf{s}_1,\msf{s}_2;\msf{d}')=2$. Moreover, we can easily argue that both $\msf{s}_1$ and $\msf{s}_2$ can reach $\msf{d}'$ by contradiction, and hence $\mincut(\msf{s}_i;\msf{d}')=1$, for $i=1,2$. Consider a multiple access flow problem with two sources $\msf{s}_1,\msf{s}_2$ and a single destination $\msf{d}'$. The capacity region is the square region $\{(R_1,R_2): R_1,R_2\ge 0, R_1\le \mincut(\msf{s}_1;\msf{d}'), R_2\le \mincut(\msf{s}_2;\msf{d}'), R_1+R_2\le \mincut(\msf{s}_1,\msf{s}_2;\msf{d}')\}$ and can be achieved via scalar random linear coding if the extension field size $2^r$ is sufficiently large \cite{KimMedard_10}. Hence $(1,1)$ can be achieved, and $\msf{d}'$ can decode both user's symbols and so can $\mcal{U}$.

\subsubsection{Proof of Part (b)}
Fix the transmission from $\mcal{P}(\msf{v})\setminus\mcal{U}$. We write the reception of $\msf{v}$ as
\begin{align*}
&\underbrace{\sum_{\msf{u}\in\mcal{U}} \lp \alpha_{\msf{u}}\beta_{\msf{u},\msf{s}_1}\cdot a + \alpha_{\msf{u}}\beta_{\msf{u},\msf{s}_2}\cdot b \rp}_\text{To be determined}\\
+\ &\underbrace{\sum_{\msf{u}\in\mcal{P}(\msf{v})\setminus\mcal{U}} \lp \alpha_{\msf{u}}\beta_{\msf{u},\msf{s}_1}\cdot a + \alpha_{\msf{u}}\beta_{\msf{u},\msf{s}_2}\cdot b \rp}_\text{Given}
\end{align*}

From part (a) we know that $\mcal{U}$ can collectively solve $a$ and $b$ with high probability, and hence it can construct any linear combination of $a$ and $b$. Therefore, they can arrange their transmission by choosing the scaling coefficients $\alpha$'s carefully so that combined with the given part in $\msf{v}$, the aggregate reception at $\msf{v}$ is the desired linear combination.

\subsubsection{Proof of Part (c)}
From part (a) we know the the subspace spanned by the received linear combinations of $\mcal{U}$ has dimension $2$ with high probability. The received linear combination of $\msf{u}$ spans an one-dimensional space with high probability. Note that $\mcal{U}\setminus \{\msf{u}\}\ne\emptyset$.

Consider the subspace spanned by the received linear combination(s) of $\mcal{U}\setminus \{\msf{u}\}$. This subspace is either has dimension $2$ or has dimension $1$ but not aligned with the reception of $\msf{u}$. In the first case, after the nodes in $\mcal{U}\setminus \{\msf{u}\}$ chose the scaling coefficients randomly, uniformly, and independently over $\mbb{F}_{2^r}$, the resulting effective linear combination at $\msf{v}$ contributed by this part is uniformly distributed over the whole two-dimensional space. Hence it is not aligned with the reception of $\msf{u}$ with high probability. $\msf{u}$ can then choose its scaling coefficient properly so that any desired linear combination except those aligned with the reception of $\msf{u}$ can be formed at $\msf{v}$. In the second case, it can be guaranteed that the resulting effective linear combination at $\msf{v}$ contributed by $\mcal{U}\setminus \{\msf{u}\}$ is not aligned with the reception of $\msf{u}$. Hence we arrive at the same conclusion as above.

\subsubsection{Proof of Part (d)}
This is a simple corollary of part (c). Since $\msf{u}$ is $\msf{s}_1\msf{s}_2$-reachable, with all its predecessors doing RLC it will receive a linear combination of $a$ and $b$ with non-zero coefficients for both symbols with high probability.  Hence linear combinations consisting of purely $a$ or $b$ with high probability at $\msf{v}$ can be formed at $\msf{v}$ due to the conclusion in part (c).

\subsection{Proof of Lemma~\ref{lem_mincut_pair}}
$\mincut(\msf{s}_1,\msf{s}_2;\mcal{U})= 2.$ Note that all nodes can be reached by at least one of the source nodes. Fix node $\msf{u}\in\mcal{U}.$

For sufficiently large block length $N,$ if all nodes perform RLC with one symbol from each source, then by Lemma~\ref{lem_reachability} and Lemma~\ref{lem_two_source}(a), we have the following with high probability:
\begin{itemize}
\item the subspace spanned by the received linear combination at $\msf{u}$ has dimension $1$, and
\item the subspace spanned by the received linear combinations at $\mcal{U}$ has dimension $2$.
\end{itemize}

Fix any choice of the coefficients so that the above hold. Pick any other node $\msf{w}\in\mcal{U}$ such that the subspace spanned by the received linear combinations at $\msf{u}$ and $\msf{w}$ has dimension 2. Then, we must have  $\mincut(\msf{s}_1,\msf{s}_2;\msf{u},\msf{w})= 2,$ or else $\msf{u}, \msf{w}$ could not have received linearly independent linear combinations.

(Note: Lemma~\ref{lem_mincut_pair} is a purely graph-theoretic lemma. It is easier to prove it however using the random coding arguments in Lemma~\ref{lem_reachability} and Lemma~\ref{lem_two_source}(a).)

\subsection{Proof of Lemma~\ref{lem_Pattern}}
For all four cases, the direction ``$\Leftarrow$" is quite obvious. Also note that when $k_1^*=k_2^*=k^*$ and there is no omniscient node in, the two clone-sets, $\mcal{K}(\msf{v}_1^*)$ and $\mcal{K}(\msf{v}_2^*)$, partition the whole layer $\mcal{L}_{k^*}$. Therefore, it is sufficient to look at $v_1^*$ and $v_2^*$ only. 

For the other direction ``$\Rightarrow$", we shall prove the first and the third case, in which the superscript of the conditions is ``$(12)$". To satisfy $\textrm{T}^{(12)}_2$ and $\textrm{T}^{(12)}_3$ (equivalently $\textrm{P}^{(12)}_2$ and  $\textrm{P}^{(12)}_3$), we require $\mincut\lp\msf{s}_1,\msf{s}_2; \mcal{P}^{\msf{s}_2}(\msf{v}_1^*)\rp = 1$ as well as $\mcal{K}_{\mcal{G}_{12}}^{\msf{s}_1}\lp\msf{u}_{21}\rp$ forms an $\lp \msf{s}_1; \msf{d}_1\rp$-vertex-cut in $\mcal{G}_{12}$. In generating $\mcal{G}_{12}$, since there is only one node $\msf{v}_2^*$ (up to clones) in the same layer as $\msf{v}_1^*$, the reorganization step will not involve any change in edges, as $M=1$. 
There are two possible cases where $\mcal{K}_{\mcal{G}_{12}}^{\msf{s}_1}\lp\msf{u}_{21}\rp$ forms an $\lp \msf{s}_1; \msf{d}_1\rp$-vertex-cut in $\mcal{G}_{12}$: $\msf{u}_{21} \ne \msf{v}_2^*$, or $\msf{u}_{21} = \msf{v}_2^*$. 

In the first case where $\msf{u}_{21} \ne \msf{v}_2^*$, we have $\mincut_{\mcal{G}_{12}}\lp\msf{s}_1,\msf{s}_2; \mcal{P}_{\mcal{G}_{12}}(\msf{v}_2^*)\rp = 1$. Hence all nodes in $\mcal{P}^{\msf{s}_2}(\msf{v}_1^*)$ are parents of $\msf{v}_2^*$, and $\mincut\lp\msf{s}_1,\msf{s}_2; \mcal{P}(\msf{v}_2^*)\setminus\mcal{P}^{\msf{s}_2}(\msf{v}_1^*)\rp=1$. 
Due to the fact that $\msf{v}_1^*$ is not omniscient, $\msf{s}_2$ must be able to reach $\mcal{P}(\msf{v}_2^*)\setminus\mcal{P}^{\msf{s}_2}(\msf{v}_1^*)$. On the other hand, nodes in $\mcal{P}(\msf{v}_1^*)\setminus\mcal{P}^{\msf{s}_2}(\msf{v}_1^*)$ are all $\msf{s}_1$-only-reachable and hence cannot belong to $\mcal{P}(\msf{v}_2^*)\setminus\mcal{P}^{\msf{s}_2}(\msf{v}_1^*)$. Similarly nodes in $\mcal{P}(\msf{v}_2^*)\setminus\mcal{P}^{\msf{s}_2}(\msf{v}_1^*)$ cannot be in $\mcal{P}(\msf{v}_1^*)\setminus\mcal{P}^{\msf{s}_2}(\msf{v}_1^*)$. Therefore, we conclude that $\mcal{U}_1$ is $\msf{s}_1$-only-reachable, $\mincut\lp \msf{s}_1,\msf{s}_2;\mcal{W}\rp = \mincut\lp \msf{s}_1,\msf{s}_2;\mcal{U}_2\rp = 1$. Then condition 3 and 4 in $\textrm{T}^{(12)}$ ($\textrm{P}^{(12)}$) imply the rest of the conditions in the right-hand-side of the case $\textrm{T}^{(12)}$ ($\textrm{P}^{(12)}$). It is quite easy to see that in this case $\textrm{P}^{(12)}\cap \textrm{T}^{(21)} = \emptyset$.

In the second case where $\msf{u}_{21} = \msf{v}_2^*$, $\mcal{P}(\msf{v}_1^*)\setminus \mcal{P}^{\msf{s}_2}(\msf{v}_1^*)$ should be equal to the set of $\msf{s}_1$-reachable parents of $\msf{v}_2^*$ in $\mcal{G}_{12}$. If $\textrm{T}^{(12)}_3$ is satisfied, then $\msf{v}_1^*$ and $\msf{v}_2^*$ will share the same $\msf{s}_1$-reachable parents, contradicting the assumption that there is no omniscient node. If $\textrm{P}^{(12)}_3$ is satisfied, then it must be the case that $\msf{v}_2^*$ has no parents in $\mcal{P}^{\msf{s}_2}(\msf{v}_1^*)$. Therefore, $\mcal{P}^{\msf{s}_2}(\msf{v}_1^*) = \mcal{U}_1$, all nodes in $\mcal{W}$ are $\msf{s}_1$-only-reachable, and all nodes in $\mcal{U}_1$ are $\msf{s}_2$-only-reachable. Then condition $\textrm{P}^{(12)}_4$ implies that $\mcal{K}^{\msf{s}_2}\lp\msf{w}_{12}\rp$ forms an $\lp \msf{s}_2; \msf{d}_2\rp$-vertex-cut. Then it is easy to verify that $\textrm{T}^{(21)}$ is satisfied. So considering $\textrm{P}^{(12)}\setminus\textrm{T}^{(21)}$, this pattern will not be included. Proof complete.

\subsection{Proof of Claim~\ref{claim_SameLay_1}}
\begin{figure}[htbp]
{\center
\subfigure[Illustration of Case 1)(i). $\msf{u}_1$ and $\msf{w}_2$ are $\msf{s}_1$-reachable and $\msf{u}_2$ is $\msf{s}_2$-reachable.]{\includegraphics[width=2.5in]{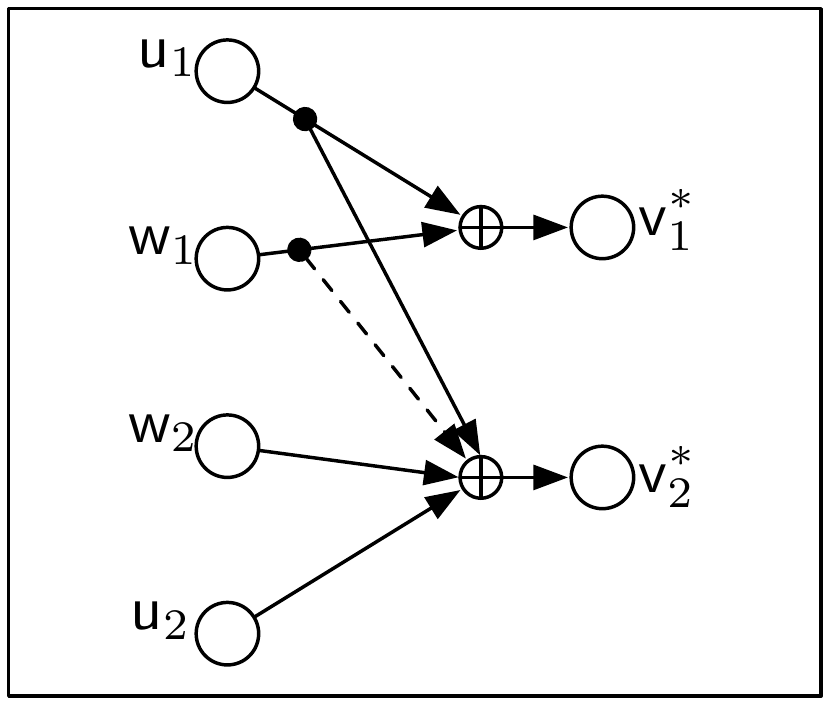}}
\subfigure[Illustration of Case 1)(ii)(1). $\msf{u}_1$ is $\msf{s}_1$-reachable, $\msf{u}_2$ is $\msf{s}_1\msf{s}_2$-reachable, and $\msf{w}_1$ is $\msf{s}_1$-only-reachable.]{\includegraphics[width=2.5in]{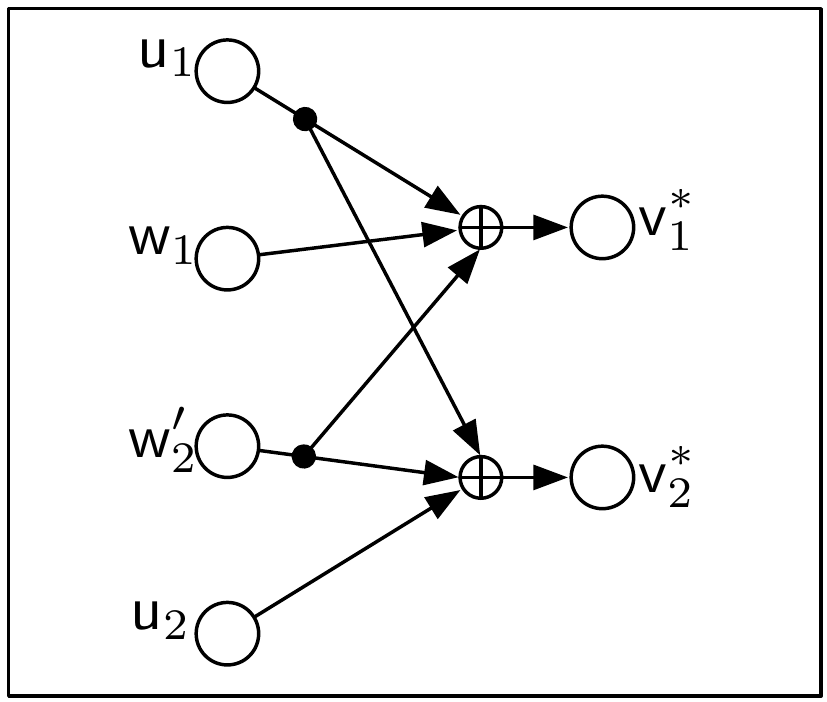}}
\subfigure[Illustration of Case 1)(ii)(2). $\msf{u}_1$ is $\msf{s}_1$-reachable and $\msf{u}_2$ is $\msf{s}_1\msf{s}_2$-reachable.]{\includegraphics[width=2.5in]{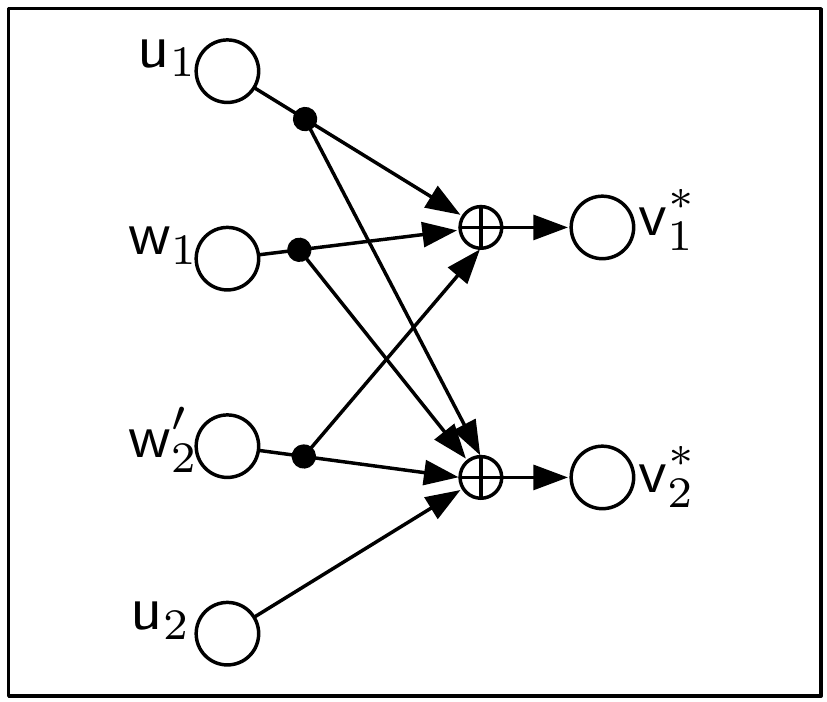}}
\subfigure[Illustration of Case $\ol{A1}\cap \ol{B2}$. $\msf{u}_1$ is $\msf{s}_1$-reachable and $\msf{u}_2$ is $\msf{s}_2$-reachable.]{\includegraphics[width=2.5in]{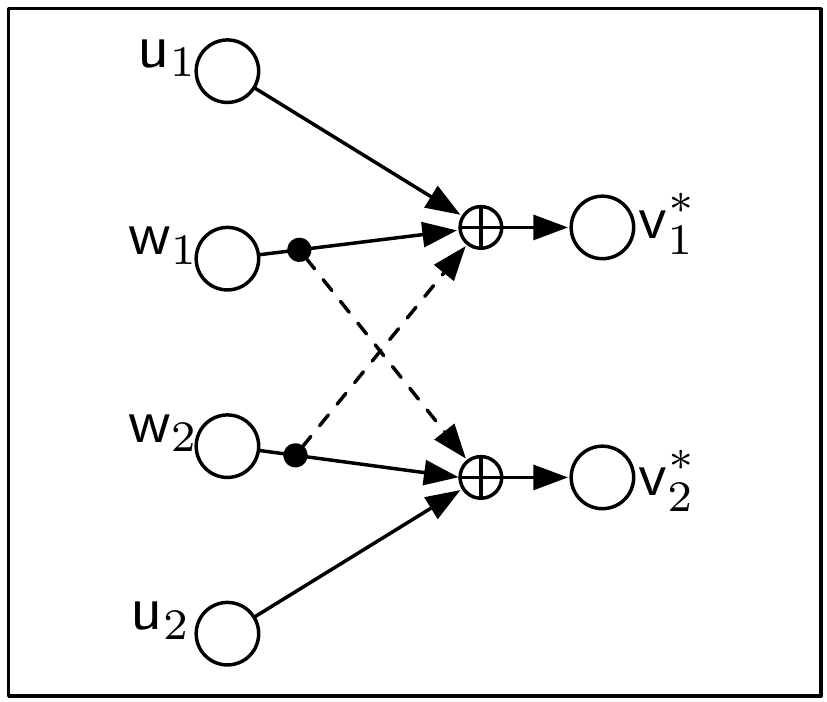}}
\caption{Critical Nodes in the Same Layer}
\label{fig_SameLayer}
}
\end{figure}

It is quite obvious that $(1,1)$ is achievable when $\mcal{W} = \emptyset$. The assumption that there is no omniscient node combined with $\mcal{U}_1^{\msf{s}_1} = \emptyset$ or $\mcal{U}_2^{\msf{s}_2} = \emptyset$, implies the following three cases:


{\flushleft 1)} $\mcal{P}^{\msf{s}_1}(\msf{v}_1^*) \subsetneq \mcal{P}^{\msf{s}_1}(\msf{v}_2^*)$ and $\mcal{P}^{\msf{s}_2}(\msf{v}_2^*) \setminus \mcal{P}^{\msf{s}_2}(\msf{v}_1^*) \ne \emptyset$:
Pick $\msf{u}_1\in\mcal{P}^{\msf{s}_1}(\msf{v}_1^*)$ and then find a node $\msf{w}_1\in\mcal{P}(\msf{v}_1^*)$ such that $\mincut(\msf{s}_1,\msf{s}_2; \msf{u}_1, \msf{w}_{1})=2$. Such a node exists by Lemma~\ref{lem_mincut_pair}. Pick nodes $\msf{u}_2\in \mcal{P}^{\msf{s}_2}(\msf{v}_2^*) \setminus \mcal{P}^{\msf{s}_2}(\msf{v}_1^*)$ and $\msf{w}_2\in \mcal{P}^{\msf{s}_1}(\msf{v}_2^*) \setminus \mcal{P}^{\msf{s}_1}(\msf{v}_1^*)$. Note that $\msf{u}_2,\msf{w}_2$ may be the same node.
\begin{itemize}  
\item [(i)] {Suppose there exist $\msf{u}_2$ and $\msf{w}_2$ described as above such that $\mincut(\msf{s}_1,\msf{s}_2;\msf{u}_2,\msf{w}_2)=2$:} See Fig.~\ref{fig_SameLayer}(a) for an illustration. We first arrange the transmission of $\msf{u}_1$ and $\msf{w}_1$ so that only user 1's symbol appears at $\msf{v}_1^*$. This can be done due to Lemma~\ref{lem_two_source}(a). Next we arrange the transmission of $\msf{u}_2$ and $\msf{w}_2$ so that the effect of user 1's symbol in the transmission of $\msf{u}_1$ (and possibly $\msf{w}_2$) at $\msf{v}_2^*$ can be neutralized, and user 2's symbol can appear cleanly. This can be done due to Lemma~\ref{lem_two_source}(b).



\item [(ii)] {Suppose $\mincut(\msf{s}_1,\msf{s}_2;\msf{u}_2,\msf{w}_2)=1$ for all $\msf{u}_2$ and $\msf{w}_2$ described as above:}
Then, we must have  $\mcal{P}^{\msf{s}_2}(\msf{v}_2^*) \setminus \mcal{P}^{\msf{s}_2}(\msf{v}_1^*)=\mcal{P}^{\msf{s}_1}(\msf{v}_2^*) \setminus \mcal{P}^{\msf{s}_1}(\msf{v}_1^*),$ for if not, we can always find nodes $\msf{u}_2\in \mcal{P}^{\msf{s}_2}(\msf{v}_2^*) \setminus \mcal{P}^{\msf{s}_2}(\msf{v}_1^*)$ and $\msf{w}_2\in \mcal{P}^{\msf{s}_1}(\msf{v}_2^*) \setminus \mcal{P}^{\msf{s}_1}(\msf{v}_1^*)$ such that $\mincut(\msf{s}_1,\msf{s}_2;\msf{u}_2,\msf{w}_2)=2.$ Thus, there must be a node $\msf{w}'_2\in\mcal{P}(\msf{v}_1^*)\cap\mcal{P}(\msf{v}_2^*)$ such that $\mincut(\msf{s}_1,\msf{s}_2;\msf{u}_2,\msf{w}'_2)=2$, by the definition of $\msf{v}_2^*$. Note that $\msf{w}'_2$ may be the same node as $\msf{u}_1$, $\msf{w}_1$, or a clone of either one. Also note that now $\msf{u}_2$ must be $\msf{s}_1\msf{s}_2$-reachable. See Fig.~\ref{fig_SameLayer}(b)(c) for an illustration. We further distinguish into two cases based on whether $\msf{w}_1$ is a parent of $\msf{v}_2^*$ or not:
{\flushleft (1)}
If $\msf{w}_1$ is not a parent of $\msf{v}_2^*$, then it is $\msf{s}_2$-only-reachable since $\mcal{P}^{\msf{s}_1}(\msf{v}_1^*) \subsetneq \mcal{P}^{\msf{s}_1}(\msf{v}_2^*)$. We let $\msf{u}_1$ and $\msf{w}'_2$ do RLC. Since $\msf{u}_2$ is $\msf{s}_1\msf{s}_2$-reachable, by Lemma~\ref{lem_two_source}(d), it can arrange its transmission so that $\msf{v}_2^*$ can decode $\msf{s}_2$'s symbol. We can then use $\msf{w}_1$ to neutralize user 2's symbol in $\msf{v}_1^*$'s reception if necessary. Since $\msf{u}_1$ is $\msf{s}_1$-reachable, $\msf{v}_1^*$ can obtain user 1's symbol cleanly after neutralization.

{\flushleft (2)}
If $\msf{w}_1$ is a parent of $\msf{v}_2^*$, then we first arrange the transmission of $\{\msf{u}_1, \msf{w}_1, \msf{w}'_2\}$ so that $\msf{v}_1^*$ can decode user 1's symbol. This can be done due to Lemma~\ref{lem_two_source}(a). Next, since their aggregate at $\msf{v}_2^*$ has only user 1's symbol and $\msf{u}_2$ is $\msf{s}_1\msf{s}_2$-reachable, we can arrange the transmission of $\msf{u}_2$ so that user 1's symbol is neutralized and only user 2's symbol is left.



\end{itemize}

{\flushleft 2)} $\mcal{P}^{\msf{s}_2}(\msf{v}_2^*) \subsetneq \mcal{P}^{\msf{s}_2}(\msf{v}_1^*)$ and $\mcal{P}^{\msf{s}_1}(\msf{v}_1^*) \setminus \mcal{P}^{\msf{s}_1}(\msf{v}_2^*) \ne \emptyset$: Similar to the previous case.
{\flushleft 3)} $\mcal{P}^{\msf{s}_1}(\msf{v}_1^*) \subsetneq \mcal{P}^{\msf{s}_1}(\msf{v}_2^*)$ and $\mcal{P}^{\msf{s}_2}(\msf{v}_2^*) \subsetneq \mcal{P}^{\msf{s}_2}(\msf{v}_1^*)$:
Pick a node $\msf{u}_1 \in \mcal{P}^{\msf{s}_2}(\msf{v}_1^*) \setminus \mcal{P}^{\msf{s}_2}(\msf{v}_2^*)$ and a node $\msf{u}_2 \in \mcal{P}^{\msf{s}_1}(\msf{v}_2^*) \setminus \mcal{P}^{\msf{s}_1}(\msf{v}_1^*)$. By the definition of $\msf{v}_1^*$ and $\msf{v}_2^*$, we shall be able to find $\msf{w}_{1}\in\mcal{P}(\msf{v}_1^*)$ and $\msf{w}_{2}\in\mcal{P}(\msf{v}_2^*)$ such that $\mincut(\msf{s}_1,\msf{s}_2; \msf{u}_1, \msf{w}_{1})=\mincut(\msf{s}_1,\msf{s}_2; \msf{u}_2, \msf{w}_{2}) = 2$. Note that $\msf{w}_{1}, \msf{w}_{2}$ may be the same node but $\msf{u}_1, \msf{u}_2$ are different nodes, although they may be clones.
We shall show that $(1,1)$ is achievable. First let $\{\msf{w}_1,\msf{w}_2$\} do RLC. Then we can arrange the transmission of $\msf{u}_1$ and $\msf{u}_2$ such that $\msf{v}_1^*$ and $\msf{v}_2^*$ can obtain their desired symbols due to Lemma~\ref{lem_two_source}(c).

\subsection{Proof of Claim~\ref{claim_SameLay_2}}
Note that 
\begin{align*}
&\neg\lp A\vee B\rp = \neg\lp A1\wedge A2\rp  \wedge \neg \lp B1\wedge B2\rp = \lp \neg A1 \vee \neg A2 \rp \wedge \lp \neg B1 \vee \neg B2 \rp\\
&= \lp \neg A1 \wedge \neg B1 \rp \vee \lp \neg A1 \wedge \neg B2\rp \vee \lp \neg A2 \wedge \neg B1\rp \vee \lp \neg A2 \wedge \neg B2\rp.
\end{align*}

We distinguish into 4 cases.
\begin{enumerate}
\item $\neg A1 \wedge \neg B1$:\par
In this case, there is a node in $\mcal{P}_1^{\msf{s}_1}$ that is $\msf{s}_1\msf{s}_2$-reachable and there is another node in $\mcal{P}_1^{\msf{s}_1}$ that is $\msf{s}_1$-only-reachable. Hence $\mincut\lp \msf{s}_1,\msf{s}_2; \mcal{P}_1^{\msf{s}_1}\rp = 2$. We can first arrange the transmission of $\mcal{P}(\msf{v}_2^*)$ so that $\msf{v}_2^*$ can decode $b$. Since $\mincut\lp \msf{s}_1,\msf{s}_2; \mcal{P}_1^{\msf{s}_1}\rp = 2$, we can arrange their transmission to form any linear combination of $a$ and $b$; in particular, the one that combined with the transmission from $\mcal{W}$ forms $a$ at $\msf{v}_1^*$. Hence $(1,1)$ is achievable.

\item $\neg A1 \wedge \neg B2$:\par
In this case there is a node $\msf{u}_1\in\mcal{P}_1^{\msf{s}_1}$ that is $\msf{s}_1\msf{s}_2$-reachable and there is a node $\msf{u}_2\in\mcal{P}_2^{\msf{s}_2}$ that is $\msf{s}_1\msf{s}_2$-reachable.  Locate nodes $\msf{w}_1\in\mcal{P}(\msf{v}_1^*)$ and $\msf{w}_2\in\mcal{P}(\msf{v}_2^*)$ such that $\mincut\lp\msf{s}_1,\msf{s}_2;\msf{u}_1,\msf{w}_1\rp=\mincut\lp\msf{s}_1,\msf{s}_2;\msf{u}_2,\msf{w}_2\rp=2.$ Note that $\msf{w}_1, \msf{w}_2$ may be the same node but $\msf{u}_1, \msf{u}_2$ will be different nodes although they may be clones. Then, let $\msf{w}_1,\msf{w}_2$ perform RLC while $\msf{u}_1,\msf{u}_2$ arrange their transmissions so that $\msf{v}_1^*,\msf{v}_2^*$ can decode their desired symbols. 
This can be done with high probability due to Lemma~\ref{lem_two_source}(d).
See Fig.~\ref{fig_SameLayer}(d) for an illustration.

\item $\neg A2 \wedge \neg B1$:\par
In this case there is a node in $\mcal{P}_1^{\msf{s}_1}$ that is $\msf{s}_1$-only-reachable and there is a node in $\mcal{P}_2^{\msf{s}_2}$ that is $\msf{s}_2$-only-reachable. Obviously $(1,1)$ is achievable.

\item $\neg A2 \wedge \neg B2$:\par
In this case, there is a node in $\mcal{P}_2^{\msf{s}_2}$ that is $\msf{s}_2$-only-reachable and there is another node in $\mcal{P}_2^{\msf{s}_2}$ that is $\msf{s}_1\msf{s}_2$-reachable. Similar to the first case, $(1,1)$ is achievable.
\end{enumerate}
Proof complete.

\subsection{Proof of Lemma~\ref{lem_special}}
We shall distinguish the condition $\mcal{P}^{\msf{s}_2}(\msf{u}_2)\ne \mcal{P}^{\msf{s}_2}(\msf{v}_1^*)$ into two cases, (1) $\mcal{P}^{\msf{s}_2}(\msf{u}_2)\setminus \mcal{P}^{\msf{s}_2}(\msf{v}_1^*)\ne \emptyset$, and (2) $\mcal{P}^{\msf{s}_2}(\msf{u}_2)\subsetneq \mcal{P}^{\msf{s}_2}(\msf{v}_1^*)$.

\subsubsection{$\mcal{P}^{\msf{s}_2}(\msf{u}_2)\setminus \mcal{P}^{\msf{s}_2}(\msf{v}_1^*)\ne \emptyset$}

In this case, if $\mcal{P}^{\msf{s}_2}(\msf{u}_2)\cap \mcal{P}^{\msf{s}_2}(\msf{v}_1^*)= \emptyset$, then the special linear coding operation in $\mcal{P}^{\msf{s}_2}(\msf{u}_2)$ will not affect the coefficient of user 2's symbol $b$ in the reception of $\msf{u}$. Therefore the goal in the claim of this lemma can be met from $\mincut(\msf{s}_1,\msf{s}_2;\mcal{P}(\msf{v}_1^*)) = 2$ and Lemma~\ref{lem_two_source}(b). Below we consider the case where $\mcal{P}^{\msf{s}_2}(\msf{u}_2)\cap \mcal{P}^{\msf{s}_2}(\msf{v}_1^*)\ne \emptyset$.

If $\mcal{P}^{\msf{s}_2}(\msf{v}_1^*)\setminus \mcal{P}^{\msf{s}_2}(\msf{u}_2)\ne \emptyset$, then we shall let the nodes in $\mcal{P}^{\msf{s}_2}(\msf{u}_2)\cap \mcal{P}^{\msf{s}_2}(\msf{v}_1^*)$ do RLC. Hence the parents of $\msf{u}_2$ are all doing RLC. Since $\msf{s}_2$ can reach $\msf{u}_2$, the coefficient of $b$ in the reception of $\msf{u}_2$ is non-zero with high probability. Now we turn to $\msf{v}_1^*$. As $\mincut(\msf{s}_1,\msf{s}_2;\mcal{P}(\msf{v}_1^*)) = 2$ and all nodes other than $\mcal{P}^{\msf{s}_2}(\msf{v}_1^*)\setminus \mcal{P}^{\msf{s}_2}(\msf{u}_2)$ are doing RLC, by Lemma~\ref{lem_two_source}(c) they can arrange their transmission so that $\msf{v}_1^*$ receives a linear combination consisting of $a$ solely.

\subsubsection{$\mcal{P}^{\msf{s}_2}(\msf{u}_2)\subsetneq \mcal{P}^{\msf{s}_2}(\msf{v}_1^*)$}
In this case we let the nodes in $\mcal{P}^{\msf{s}_2}(\msf{u}_2)$ do RLC. Hence the coefficient of $b$ in the reception of $\msf{u}_2$ is non-zero with high probability since all its predecessor are doing RLC. Then as $\mincut(\msf{s}_1,\msf{s}_2;\mcal{P}(\msf{v}_1^*)) = 2$ and all nodes other than $\mcal{P}^{\msf{s}_2}(\msf{v}_1^*)\setminus \mcal{P}^{\msf{s}_2}(\msf{u}_2)$ are doing RLC, by Lemma~\ref{lem_two_source}(c) they can arrange their transmission so that $\msf{v}_1^*$ receives a linear combination consisting of $a$ solely.

\subsection{Proof of Lemma~\ref{lem_Reception}}



The special coding operation performed by nodes in $\mcal{P}^{\msf{s}_2}(\msf{v}_1^*)$ is as follows: Nodes choose their coefficients independently and uniformly over the set of coefficients satisfying $\sum_{\msf{u}\in\mcal{U}} \alpha_{\msf{u}}\beta_{\msf{u},\msf{s}_2}=0.$ Under this special coding, it is easy to show the first part of the assertion, namely, that each node receives a non-trivial linear combination. Because the reception of $\mcal{P}^{\msf{s}_2}(\msf{v}_1^*)$ has full rank, the linear constraint leaves the sum $\sum_{\msf{u}\in\mcal{U}} \alpha_{\msf{u}}\beta_{\msf{u},\msf{s}_1}$ non-zero with high probability. This allows us to argue that any $\msf{s}_1$-reachable node receives a non-zero coefficient for the symbol $a$ transmitted by source $\msf{s}_1$ inspite of the special coding.

Find node $\msf{u}\in\mcal{U}$ such that $\msf{u}\notin \mcal{K}^{\msf{s}_2}(\msf{v}_1^*).$ and $\msf{u}$ is $\msf{s}_2$-reachable. Such a node exists because $\msf{v}_1^*$ is not omniscient. Then, find node $\msf{w}\in\mcal{U}$ such that $\mincut(\msf{s}_1,\msf{s}_2;\msf{u},\msf{w})=2.$ If all nodes in $\mcal{P}^{\msf{s}_2}(\msf{v}_1^*)$ performed random linear coding, then $\msf{u},\msf{w}$ jointly can decode both symbols $a$ and $b$ with high probability.

Let $\mcal{P}_1:=\mcal{P}(\msf{u})\setminus\mcal{P}(\msf{w}), \mcal{P}_{12}:=\mcal{P}(\msf{u})\cap\mcal{P}(\msf{w}), \mcal{P}_2:=\mcal{P}(\msf{w})\setminus\mcal{P}(\msf{u}).$
\begin{itemize}
\item As $\mcal{P}(\msf{u})\neq \emptyset,$ we have $\mcal{P}_1\cup \mcal{P}_{12}\neq \emptyset.$
\item As $\mcal{P}(\msf{w})\neq \emptyset,$ we have $\mcal{P}_{12}\cup \mcal{P}_2\neq \emptyset.$
\item As $\mcal{P}(\msf{u})\neq \mcal{P}(\msf{w})$ (since $\mincut(\msf{s}_1,\msf{s}_2;\msf{u},\msf{w})=2$), we have $\mcal{P}_1\cup \mcal{P}_2\neq \emptyset.$
\end{itemize}

Note that the conditions on the sets $\mcal{P}_1,\mcal{P}_{12},\mcal{P}_2$ are symmetric.

Reception of node $\msf{u}$: \par
$\quad\lp \sum_{\msf{x}\in \mcal{P}_1\cup \mcal{P}_{12}} \alpha_\msf{x}\beta_{\msf{x},\msf{s}_1}\rp\cdot a + \lp \sum_{\msf{x}\in \mcal{P}_1\cup \mcal{P}_{12}} \alpha_\msf{x}\beta_{\msf{x},\msf{s}_2}\rp\cdot b$

Reception of node $\msf{w}$: \par
$\quad\lp \sum_{\msf{x}\in \mcal{P}_{12}\cup \mcal{P}_2} \alpha_\msf{x}\beta_{\msf{x},\msf{s}_1}\rp\cdot a + \lp \sum_{\msf{x}\in \mcal{P}_{12}\cup \mcal{P}_2} \alpha_\msf{x}\beta_{\msf{x},\msf{s}_2}\rp\cdot b$

Let $D:=\begin{vmatrix} \sum_{\msf{x}\in \mcal{P}_1\cup \mcal{P}_{12}} \alpha_\msf{x}\beta_{\msf{x},\msf{s}_1} &  \sum_{\msf{x}\in \mcal{P}_1\cup \mcal{P}_{12}} \alpha_\msf{x}\beta_{\msf{x},\msf{s}_2} \\
  \sum_{\msf{x}\in \mcal{P}_{12}\cup \mcal{P}_2} \alpha_\msf{x}\beta_{\msf{x},\msf{s}_1} & \sum_{\msf{x}\in \mcal{P}_{12}\cup \mcal{P}_2} \alpha_\msf{x}\beta_{\msf{x},\msf{s}_2}\end{vmatrix}.$ $D$ is non-zero if and only if $\msf{u},\msf{w}$ can jointly decode both symbols $a$ and $b.$

For two nodes $\msf{x},\msf{y},$ denote the determinant $\begin{vmatrix} \beta_{\msf{x},\msf{s}_1} & \beta_{\msf{x},\msf{s}_2} \\ \beta_{\msf{y},\msf{s}_1} & \beta_{\msf{y},\msf{s}_2} \end{vmatrix}$ by $\beta(\msf{x},\msf{y}).$ Some algebra allows the determinant $D$ to be expressed as:

\begin{align}\notag
D&=\sum_{\msf{x}\in \mcal{P}_1} \sum_{\msf{y}\in \mcal{P}_{12}} \alpha_{\msf{x}}\alpha_{\msf{y}} \beta(\msf{x},\msf{y})+\sum_{\msf{x}\in \mcal{P}_{12}} \sum_{\msf{y}\in \mcal{P}_2} \alpha_{\msf{x}}\alpha_{\msf{y}} \beta(\msf{x},\msf{y})\\&\quad+\sum_{\msf{x}\in \mcal{P}_2} \sum_{\msf{y}\in \mcal{P}_1} \alpha_{\msf{x}}\alpha_{\msf{y}} \beta(\msf{x},\msf{y}).\label{eq:delta}
\end{align}

Note that $D$ is also symmetric in the sets $\mcal{P}_1,\mcal{P}_{12},\mcal{P}_2.$ Let the special coding set $\mcal{P}^{\msf{s}_2}(\msf{v}_1^*)$ be denoted by $\mcal{P}.$ We are given that $\mincut\lp \msf{s}_1,\msf{s}_2;\mcal{P}\rp=2.$ The constraint placed on the coding coefficients of nodes in $\mcal{P}$ is $\sum_{\msf{x}\in \mcal{P}} \alpha_{\msf{x}}\beta_{\msf{x},\msf{s}_2}=0.$

\begin{itemize}
\item Suppose $\mcal{P}\setminus(\mcal{P}_1\cup \mcal{P}_{12}\cup \mcal{P}_2)\neq \emptyset.$ 

Because we have with high probability, $\beta_{\msf{x},\msf{s}_2}\neq 0 \forall x\in \mcal{P},$ we can view the special coding as all nodes in $\mcal{P}_1\cup \mcal{P}_{12}\cup \mcal{P}_2$ performing random linear coding while nodes in $\mcal{P}\setminus(\mcal{P}_1\cup \mcal{P}_{12}\cup \mcal{P}_2)$ performing restricted coding. In this case, parent nodes of $\msf{u}$ and $\msf{w}$ perform RLC and so, the claim is obviously true.

\item Suppose $\mcal{P}\subseteq \mcal{P}_1\cup \mcal{P}_{12}\cup \mcal{P}_2$ and suppose there are two non-empty sets among $\mcal{P}\cap \mcal{P}_1, \mcal{P}\cap \mcal{P}_{12}, \mcal{P}\cap \mcal{P}_2.$ 

Without loss of generality, assume $\mcal{P}\cap \mcal{P}_1\neq\emptyset.$ Fix $\msf{x}_0\in \mcal{P}\cap \mcal{P}_1.$ Find $\msf{x}_1\in \mcal{P}$ such that $\mincut\lp \msf{s}_1,\msf{s}_2;\msf{x}_0,\msf{x}_1\rp=2.$ If $\msf{x}_1\in \mcal{P}_{12}$ or $\msf{x}_1\in \mcal{P}_2,$ then we have $x_0\in \mcal{P}\cap \mcal{P}_1,$ and $x_1\in \mcal{P}\cap \mcal{P}_{12}$ or $x_1\in \mcal{P}\cap \mcal{P}_2$ such that $\mincut\lp \msf{s}_1,\msf{s}_2;\msf{x}_0,\msf{x}_1\rp=2.$

If $\msf{x}_1\in \mcal{P}_1,$ then pick any node $\msf{x}_2$ in the non-empty set $\mcal{P}\cap (\mcal{P}_{12}\cup \mcal{P}_2).$ By submodularity, we have 
$\mincut \lp \msf{s}_1,\msf{s_2};\msf{x}_0,\msf{x}_1,\msf{x}_2\rp+\mincut \lp \msf{s}_1,\msf{s_2};\msf{x}_2\rp\leq \mincut \lp \msf{s}_1,\msf{s_2};\msf{x}_0,\msf{x}_2\rp+\mincut \lp \msf{s}_1,\msf{s_2};\msf{x}_1,\msf{x}_2\rp.$
Since the two terms on the left are 2 and 1 respectively, at least one term on the right must be greater than 1 and thus, 2.

Thus, we can always find nodes $\msf{x}_0\in \mcal{P}\cap E, \msf{x}_1\in \mcal{P}\cap F,$ where $(E,F)=(\mcal{P}_1,\mcal{P}_{12}), (\mcal{P}_{12},\mcal{P}_2)$ or $(\mcal{P}_2,\mcal{P}_1).$ such that $\mincut \lp \msf{s}_1,\msf{s_2};\msf{x}_0,\msf{x}_1\rp=2.$

Suppose, without loss of generality, we have $\msf{x}_1\in \mcal{P}\cap \mcal{P}_1, \msf{x}_2\in \mcal{P}\cap \mcal{P}_{12}$ so that $\mincut \lp \msf{s}_1,\msf{s}_2;\msf{x}_1,\msf{x}_2\rp=2.$ 
We set $\alpha_{\msf{x}_1}=\beta_{\msf{x}_1,\msf{s}_2}^{-1}\lp \sum_{\msf{x}\in \mcal{P}\setminus\{\msf{x}_1\}} \alpha_{\msf{x}} \beta_{\msf{x},\msf{s}_2} \rp.$

Then, evaluating Equation~\eqref{eq:delta} with this substitution for $\alpha_{\msf{x}_1}$ gives us a polynomial in $(\alpha_{\msf{x}}:\msf{x}\in \mcal{P}_1\cup \mcal{P}_{12}\cup \mcal{P}_2\setminus \{\msf{x}_1\})$ with coefficients being rational functions in $(\beta_{\msf{x},\msf{s}_1}, \beta_{\msf{x},\msf{s}_2}:\msf{x}\in \mcal{P}_1\cup \mcal{P}_{12}\cup \mcal{P}_2\setminus \{\msf{x}_1\})$ which are themselves polynomials in the coding coefficients from the past stages. This polynomial has a coefficient for $\alpha_{\msf{x}_2}^2$ only in the sum

\begin{align*}
& \sum_{\msf{x}\in \mcal{P}} \sum_{\msf{y}\in \mcal{P}_{12}\cup \mcal{P}_2} \alpha_{\msf{x}}\alpha_{\msf{y}}\beta(\msf{x},\msf{y}) \\
& = \sum_{\msf{x}\in \mcal{P}\setminus\{\msf{x}_1\}} \sum_{\msf{y}\in \mcal{P}_{12}\cup \mcal{P}_2} \alpha_{\msf{y}}\left[\alpha_{\msf{x}}\beta(\msf{x},\msf{y}) +\alpha_{\msf{x}}\beta_{\msf{x}_1,\msf{s}_2}^{-1}\beta_{\msf{x},\msf{s}_2}\beta(\msf{x}_1,\msf{y})\right]\\
& = \sum_{\msf{x}\in \mcal{P}\setminus\{\msf{x}_1\}} \sum_{\msf{y}\in \mcal{P}_{12}\cup \mcal{P}_2} \alpha_{\msf{x}}\alpha_{\msf{y}} \frac{\beta_{\msf{y},\msf{s}_2}}{\beta_{\msf{x}_1,\msf{s}_2}}\beta(\msf{x}_1,\msf{x})
\end{align*}

where the last equality follows from the identity $\beta_{\msf{x},\msf{s}_2}\beta(\msf{x}_1,\msf{y}) + \beta_{\msf{x}_1,\msf{s}_2}\beta(\msf{x},\msf{y}) + \beta_{\msf{y},\msf{s}_2}\beta(\msf{x}_1,\msf{x})=0.$

Putting $\msf{x}=\msf{y}=\msf{x}_2$ gives the coefficient of $\alpha_{\msf{x}_2}^2$ to be $\frac{\beta_{\msf{x}_2,\msf{s}_2}}{\beta_{\msf{x}_1,\msf{s}_2}}\beta(\msf{x}_1,\msf{x}_2)$ which is not identically zero since each of $\beta_{\msf{x}_2,\msf{s}_2}, \beta_{\msf{x}_1,\msf{s}_2}, \beta(\msf{x}_1,\msf{x}_2)$ are not identically zero, the first two because $\msf{x}_1,\msf{x}_2$ lie in $\mcal{P}$ and so are $\msf{s}_2$-reachable and the third because $\mincut \lp \msf{s}_1,\msf{s_2};\msf{x}_1,\msf{x}_2\rp=2.$

Thus, $D$ is not identically zero and hence, evaluates to a non-zero value with high probability.

\item Finally, suppose $\mcal{P}\subseteq \mcal{P}_1$ or $\mcal{P}\subseteq \mcal{P}_{12}$ or $\mcal{P}\subseteq \mcal{P}_2.$

\begin{itemize}
\item First, suppose $\mcal{P}\subseteq \mcal{P}_1.$ Fix $\msf{x}_1\in \mcal{P}.$ There exists $\msf{x}_2\in \mcal{P}$ such that $\mincut \lp \msf{s}_1,\msf{s}_2;\msf{x}_1,\msf{x}_2\rp=2.$ Force $\alpha_{\msf{x}_1}=\beta_{\msf{x}_1,\msf{s}_2}^{-1}\lp \sum_{\msf{x}\in \mcal{P}\setminus\{\msf{x}_1\}} \alpha_{\msf{x}} \beta_{\msf{x},\msf{s}_2} \rp.$
{\flushleft
\begin{align*}
  & D\\
  & = \sum_{\msf{x}\in \mcal{P}_1}\sum_{\msf{y}\in \mcal{P}_{12}\cup \mcal{P}_2} \alpha_{\msf{x}}\alpha_{\msf{y}}\beta(\msf{x},\msf{y}) +\sum_{\msf{x}\in \mcal{P}_{12}}\sum_{\msf{y}\in \mcal{P}_2} \alpha_{\msf{x}}\alpha_{\msf{y}}\beta(\msf{x},\msf{y})\\
  & = \sum_{\msf{x}\in \mcal{P}}\sum_{\msf{y}\in \mcal{P}_{12}\cup \mcal{P}_2} \alpha_{\msf{x}}\alpha_{\msf{y}}\beta(\msf{x},\msf{y})\\&\quad + \sum_{\msf{x}\in \mcal{P}_1\setminus \mcal{P}}\sum_{\msf{y}\in \mcal{P}_{12}\cup \mcal{P}_2} \alpha_{\msf{x}}\alpha_{\msf{y}}\beta(\msf{x},\msf{y}) + \sum_{\msf{x}\in \mcal{P}_{12}}\sum_{\msf{y}\in \mcal{P}_2} \alpha_{\msf{x}}\alpha_{\msf{y}}\beta(\msf{x},\msf{y}) \\
  & = \sum_{\msf{x}\in \mcal{P}\setminus\{\msf{x}_1\}}\sum_{\msf{y}\in \mcal{P}_{12}\cup \mcal{P}_2} \alpha_{\msf{x}}\alpha_{\msf{y}}\frac{\beta_{\msf{y},\msf{s}_2}}{\beta_{\msf{x}_1,\msf{s}_2}}\beta(\msf{x}_1,\msf{x})\\&\quad + \sum_{\msf{x}\in \mcal{P}_1\setminus \mcal{P}}\sum_{\msf{y}\in \mcal{P}_{12}\cup \mcal{P}_2} \alpha_{\msf{x}}\alpha_{\msf{y}}\beta(\msf{x},\msf{y}) + \sum_{\msf{x}\in \mcal{P}_{12}}\sum_{\msf{y}\in \mcal{P}_2} \alpha_{\msf{x}}\alpha_{\msf{y}}\beta(\msf{x},\msf{y})
\end{align*}
}
As $\msf{u}\notin\mcal{K}^{\msf{s}_2}(\msf{v}_1^*),$ we have that some node $\msf{y}_0\in \mcal{P}_{12}$ is $\msf{s}_2$-reachable. Then, $\beta_{\msf{y}_0,\msf{s}_2}$ is not identically zero and the coefficient of $\alpha_{\msf{x}_2}\alpha_{\msf{y}_0}$ is not identically zero.

\item Suppose $\mcal{P}\subseteq \mcal{P}_{12}.$ Then, fix $\msf{x}_1\in \mcal{P}.$ There exists $\msf{x}_2\in \mcal{P}$ such that $\mincut \lp \msf{s}_1,\msf{s}_2;\msf{x}_1,\msf{x}_2\rp=2.$ Force $\alpha_{\msf{x}_1}=\beta_{\msf{x}_1,\msf{s}_2}^{-1}\lp \sum_{\msf{x}\in \mcal{P}\setminus\{\msf{x}_1\}} \alpha_{\msf{x}} \beta_{\msf{x},\msf{s}_2} \rp.$
\begin{align*}
D &= \sum_{\msf{x}\in \mcal{P}\setminus\{\msf{x}_1\}}\sum_{\msf{y}\in \mcal{P}_1\cup \mcal{P}_2} \alpha_{\msf{x}}\alpha_{\msf{y}}\frac{\beta_{\msf{y},\msf{s}_2}}{\beta_{\msf{x}_1,\msf{s}_2}}\beta(\msf{x}_1,\msf{x}) + \sum_{\msf{x}\in \mcal{P}_{12}\setminus \mcal{P}}\sum_{\msf{y}\in \mcal{P}_1\cup \mcal{P}_2} \alpha_{\msf{x}}\alpha_{\msf{y}}\beta(\msf{x},\msf{y}) \\&\quad+ \sum_{\msf{x}\in \mcal{P}_1}\sum_{\msf{y}\in \mcal{P}_2} \alpha_{\msf{x}}\alpha_{\msf{y}}\beta(\msf{x},\msf{y}).
\end{align*}
Again, since $\msf{u}\notin\mcal{K}^{\msf{s}_2}(\msf{v}_1^*),$ we have that some node $\msf{y}_0\in \mcal{P}_1$ is $\msf{s}_2$-reachable. Then, $\beta_{\msf{y}_0,\msf{s}_2}$ is not identically zero and the coefficient of $\alpha_{\msf{x}_2}\alpha_{\msf{y}_0}$ is not identically zero.

\item Now, suppose $\mcal{P}\subseteq \mcal{P}_2.$ Then, fix $\msf{x}_1\in \mcal{P}.$ There exists $\msf{x}_2\in \mcal{P}$ such that $\mincut \lp \msf{s}_1,\msf{s}_2;\msf{x}_1,\msf{x}_2\rp=2.$ Force $\alpha_{\msf{x}_1}=\beta_{\msf{x}_1,\msf{s}_2}^{-1}\lp \sum_{\msf{x}\in \mcal{P}\setminus\{\msf{x}_1\}} \alpha_{\msf{x}} \beta_{\msf{x},\msf{s}_2} \rp.$
\begin{align*}
D &= \sum_{\msf{x}\in \mcal{P}\setminus\{\msf{x}_1\}}\sum_{\msf{y}\in \mcal{P}_1\cup \mcal{P}_{12}} \alpha_{\msf{x}}\alpha_{\msf{y}}\frac{\beta_{\msf{y},\msf{s}_2}}{\beta_{\msf{x}_1,\msf{s}_2}}\beta(\msf{x}_1,\msf{x}) + \sum_{\msf{x}\in \mcal{P}_2\setminus \mcal{P}}\sum_{\msf{y}\in \mcal{P}_1\cup \mcal{P}_{12}} \alpha_{\msf{x}}\alpha_{\msf{y}}\beta(\msf{x},\msf{y})\\&\quad + \sum_{\msf{x}\in \mcal{P}_1}\sum_{\msf{y}\in \mcal{P}_{12}} \alpha_{\msf{x}}\alpha_{\msf{y}}\beta(\msf{x},\msf{y}).
\end{align*}
Again, as $\msf{u}$ is $\msf{s}_2$-reachable, we have that some node $\msf{y}_0\in \mcal{P}_1\cup \mcal{P}_{12}$ is $\msf{s}_2$-reachable. Then, $\beta_{\msf{y}_0,\msf{s}_2}$ is not identically zero and the coefficient of $\alpha_{\msf{x}_2}\alpha_{\msf{y}_0}$ is not identically zero.

\end{itemize}
\end{itemize}

\subsection{Proof of Lemma~\ref{lem_Rank}}
For $A\subseteq \mcal{U},$ define $f(A)$ as the rank of the $|A|\times
2$ matrix with rows given by $\begin{bmatrix}\lambda_{\msf{u}} &
  \mu_{\msf{u}}\end{bmatrix}$ for $\msf{u}\in A,$ and define
$g(A)=\mincut\lp A ; \mcal{V}\rp.$ Then, $f(\cdot), g(\cdot)$ are rank
functions of two matroids on the same ground set $\mcal{U}.$ The given
conditions tell us that both these matroids have rank at least two and
every singleton subset has rank 1 in both matroids. We will first show
that there exist a two-element subset of $\mcal{U}$ that has rank 2 in
both matroids.

Find two elements $\msf{x},\msf{y}\in\mcal{U},$ such that $f(\{\msf{x},\msf{y}\})=2.$ If
$g(\{\msf{x},\msf{y}\})=2,$ we have found the desired two-element subset. Else, we
must have $g(\{\msf{x},\msf{y}\})=1.$ Then, there exists an element $\msf{z}\in\mcal{U}$
such that $g(\{\msf{x},\msf{z}\})=2.$ If $f(\{\msf{x},\msf{z}\})=2,$ we have the required
2-element subset. Else if we have $f(\{\msf{x},\msf{z}\})=1,$ then by
submodularity, we must have 

$f(\{\msf{z}\})+f(\{\msf{x},\msf{y},\msf{z}\})\leq f(\{\msf{x},\msf{z}\})+f(\{\msf{y},\msf{z}\})$

$g(\{\msf{y}\})+g(\{\msf{x},\msf{y},\msf{z}\})\leq g(\{\msf{x},\msf{y}\})+g(\{\msf{y},\msf{z}\})$

These give $f(\{\msf{y},\msf{z}\}),g(\{\msf{y},\msf{z}\})\geq 2,$ and thus, $\{\msf{y},\msf{z}\}$ is the
required subset of $\mcal{U}$ that has rank 2 in both matroids.

Thus, we have two nodes $\msf{x},\msf{y}\in\mcal{U}$ such that
$\begin{vmatrix} \lambda_{\msf{x}} & \mu_{\msf{x}} \\
  \lambda_{\msf{y}} & \mu_{\msf{y}} \end{vmatrix}\neq 0$ and
$\mincut\lp \msf{x},\msf{y};\mcal{V}\rp =2.$

Again, for $A\subseteq\mcal{V},$ the function defined by
$h(A)=\mincut\lp \msf{x},\msf{y};A\rp$ is the rank function of a matroid over
ground set $\mcal{V}$ that has rank two. Thus, there exist
$\msf{u},\msf{w}\in\mcal{V}$ such that $\mincut \lp
\msf{x},\msf{y};\msf{u},\msf{w}\rp=2.$

Thus, when all nodes perform RLC except nodes in
$\mcal{U}\setminus\{\msf{x},\msf{y}\}$ remain silent, we have that
$\msf{u},\msf{v}$ can jointly recover both symbols $a$ and $b.$

Now, if all nodes perform RLC, the $a$ and $b$ coefficients of the receptions of nodes
$\msf{u},\msf{v}$ would be polynomials in the random coding
coefficients with a determinant that is a polynomial that is not
identically zero. QED.

\subsection{Proof of Lemma~\ref{claim_Special}}
If $\msf{v}_2^*$ has a parent from user 1's cloud $\mcal{C}_1$, then in either $\mcal{G}_{12}$ or $\mcal{G}'_{12}$ since this node in the cloud becomes $\msf{s}_1$-only-reachable while $\msf{v}_2^*$ can be reached by $\msf{s}_2$, $\msf{v}_2^*$ remains to be the critical node for user 2, ie., $\pmc_{\mcal{G}_{12}}(\msf{d}_2)=\pmc_{\mcal{G}'_{12}}(\msf{d}_2)=\msf{v}_2^*$. Since the $\msf{s}_1$-clones of $\msf{v}_2^*$ form an $\lp \msf{s}_1; \msf{d}_1\rp$-vertex-cut in $\mcal{G}_{12}$ but not in $\mcal{G}$, the only possibility is that some parents of $\msf{v}_2^*$ are not in the cloud $\mcal{C}_1$ and are dropped in $\mcal{G}_{12}$. These nodes are descendants of $\mcal{P}^{\msf{s}_2}(\msf{v}_1^*)$, which becomes $\msf{s}_1$-only-reachable in $\mcal{G}'_{12}$. Therefore in $\mcal{G}'_{12}$, $\msf{v}_2^*$ has some $\msf{s}_1$-reachable parents that is not in the cloud $\mcal{C}_1$, and hence the $\msf{s}_1$-clones of $\msf{v}_2^*$ do not form an $\lp \msf{s}_1; \msf{d}_1\rp$-vertex-cut in $\mcal{G}'_{12}$.

In the rest of the proof we deal with the case where $\msf{v}_2^*$ has no parents from user 1's cloud $\mcal{C}_1$. Hence ``$\msf{s}_1$-clones of $\msf{v}_2^*$ form an $\lp \msf{s}_1; \msf{d}_1\rp$-vertex-cut in $\mcal{G}_{12}$" implies that $\mincut_{\mcal{G}_{12}}(\msf{s}_1,\msf{s}_2; \mcal{P}_{\mcal{G}_{12}}(\msf{v}_2^*)) = 1$. We shall show that, for all possible $\mcal{G}'_{12}$, either $\mincut_{\mcal{G}'_{12}}(\msf{s}_1,\msf{s}_2; \mcal{P}_{\mcal{G}'_{12}}(\msf{v}_2^*)) = 2$, which implies that $\pmc_{\mcal{G}'_{12}}(\msf{d}_2)=\msf{v}_2^*$ and $\msf{s}_1$-clones of $\msf{v}_2^*$ do not form an $\lp \msf{s}_1; \msf{d}_1\rp$-vertex-cut in $\mcal{G}'_{12}$, or directly prove the statement.

Below a few notations are given before we proceed. $\mcal{U} := \lbp\msf{u}\in\mcal{L}_{k_1^*}: \msf{u} \text{ can reach } \msf{d}_2\rbp$. $\mcal{U}|_{\mcal{G}'_{12}}$ and $\mcal{U}|_{\mcal{G}_{12}}$ denote the nodes in the same layer as $\msf{v}_1^*$ that can reach $\msf{d}_2$ in $\mcal{G}'_{12}$ and $\mcal{G}_{12}$ respectively. Recall that $\mcal{R}$ is the set of nodes in $\mcal{P}^{\msf{s}_2}(\msf{v}_1^*)$ that can reach one of the two destinations in $\mcal{G}_{12}$. Define the following subsets of $\mcal{U}$: (use short-hand notations $\mcal{P}$ for $\mcal{P}^{\msf{s}_2}(\msf{v}_2^*)$ and $\mcal{S} := \mcal{P}\setminus \mcal{R}$)
\begin{align*}
\mcal{U}_{\mcal{P}} &:= \lbp \msf{u}: \mcal{P}(\msf{u})\supseteq \mcal{P} \rbp, \
\mcal{U}_{\mcal{Q}} := \lbp \msf{u}: \mcal{P}(\msf{u}) \cap \mcal{P} = \emptyset \rbp\\
\mcal{U}_{\mcal{R}} &:= \lbp \msf{u}: \mcal{P}(\msf{u}) \cap \mcal{P} \ne \emptyset, \mcal{P}(\msf{u}) \cap \mcal{P} \subseteq \mcal{R}\rbp \\
\mcal{U}_{\mcal{S}} &:= \lbp \msf{u}: \mcal{S}\subseteq \mcal{P}(\msf{u}) \cap \mcal{P} \subsetneq \mcal{P}\rbp
\end{align*}
Note that these four sets form a partition of $\mcal{U}$, and $\mcal{K}^{\msf{s}_2}(\msf{v}_1^*)\cap\mcal{U} \subseteq \mcal{U}_{\mcal{P}}$. 


Let us consider the following two cases: 1) $\mcal{R}\ne\emptyset$, and 2) $\mcal{R}=\emptyset$.
Note that when generating induced graphs $\mcal{G}_{12}$ and $\mcal{G}'_{12}$, some nodes may be dropped as they are no longer reachable from the sources. Consequently $\mcal{U}|_{\mcal{G}_{12}}$ and $\mcal{U}|_{\mcal{G}'_{12}}$ may be strictly contained in $\mcal{U}$. In the following discussion, we shall further distinguish into these cases.

{\flushleft 1) $\mcal{R}\ne\emptyset$}: \par
We shall show that in this case, $\mincut_{\mcal{G}'_{12}}(\msf{s}_1,\msf{s}_2; \mcal{P}_{\mcal{G}'_{12}}(\msf{v}_2^*)) = 2$.
\begin{itemize}
\item [(A)] $\mcal{U}|_{\mcal{G}'_{12}} = \mcal{U}$: 
Since no nodes are dropped in $\mcal{U}$ when generating $\mcal{G}'_{12}$, no nodes will be dropped in the later layers and $\mincut\lp\mcal{U}; \mcal{P}(\msf{v}_2^*)\rp$ remains the same in $\mcal{G}$ and $\mcal{G}'_{12}$. As $\mincut\lp\mcal{U}; \mcal{P}(\msf{v}_2^*)\rp \ge 2$ and all non-vertical cuts have cut-values at least $2$, we only need to show that $\mincut_{\mcal{G}'_{12}}\lp\msf{s}_1,\msf{s}_2; \mcal{U}\rp = 2$.
{\flushleft (i) $\mcal{U}|_{\mcal{G}_{12}} = \mcal{U}$}:\par
In this case, since $\mcal{U}|_{\mcal{G}_{12}} = \mcal{U}$, we have $\mincut_{\mcal{G}_{12}}\lp\mcal{U}|_{\mcal{G}_{12}};\mcal{P}_{\mcal{G}_{12}}(\msf{v}_2^*)\rp=\mincut\lp\mcal{U};\mcal{P}(\msf{v}_2^*)\rp \ge 2$. Therefore, $\mincut_{\mcal{G}_{12}}(\msf{s}_1,\msf{s}_2; \mcal{P}_{\mcal{G}_{12}}(\msf{v}_2^*)) = 1$ implies that $\mincut_{\mcal{G}_{12}}(\msf{s}_1,\msf{s}_2; \mcal{U}|_{\mcal{G}_{12}}) = 1$.

Suppose that $\mincut\lp\msf{s}_1,\msf{s}_2; \mcal{U}_{\mcal{R}}\rp=2$. Since nodes in $\mcal{U}_{\mcal{R}}$ will not be affected in generating $\mcal{G}_{12}$, $\mincut_{\mcal{G}_{12}}\lp\msf{s}_1,\msf{s}_2; \mcal{U}_{\mcal{R}}\rp=2$. Hence $\mincut_{\mcal{G}_{12}}\lp\msf{s}_1,\msf{s}_2; \mcal{U}|_{\mcal{G}_{12}}\rp=2$, contradicting the above fact. Besides, $\mcal{U}_{\mcal{R}}\ne \emptyset$. Therefore, $\mincut\lp\msf{s}_1,\msf{s}_2; \mcal{U}_{\mcal{R}}\rp=1$. 

Find a node $\msf{u}\in\mcal{U}$ such that $\mincut\lp\msf{s}_1,\msf{s}_2; \mcal{U}_{\mcal{R}}\cup\{\msf{u}\}\rp=2$. Below we show that this node $\msf{u}\in\mcal{U}_{\mcal{S}}$ by contradiction. Suppose $\msf{u}\in\mcal{U}_{\mcal{Q}}$. As the nodes in $\mcal{U}_{\mcal{Q}}$ will not be affected in generating $\mcal{G}_{12}$, we have $\mincut_{\mcal{G}_{12}}\lp\msf{s}_1,\msf{s}_2; \mcal{U}_{\mcal{R}}\cup\{\msf{u}\}\rp = \mincut\lp\msf{s}_1,\msf{s}_2; \mcal{U}_{\mcal{R}}\cup\{\msf{u}\}\rp=2$, contradicting the above fact that $\mincut_{\mcal{G}_{12}}(\msf{s}_1,\msf{s}_2; \mcal{U}|_{\mcal{G}_{12}}) = 1$. Next, suppose $\msf{u}\in\mcal{U}_{\mcal{P}}$. Let us first consider the min-cut value from $\{\msf{s}_1,\msf{s}_2\}$ to the collection of parents of $\mcal{U}_{\mcal{R}}\cup\{\msf{u}\}$, denoted by $\mcal{P}\lp\mcal{U}_{\mcal{R}}\cup\{\msf{u}\}\rp$. It is $2$ in $\mcal{G}$. In $\mcal{G}_{12}$, nodes in $\mcal{S}$ are dropped, but nodes in $\mcal{R}$ and nodes in $\mcal{P}\lp\mcal{U}_{\mcal{R}}\cup\{\msf{u}\}\rp \setminus \mcal{P}$ are not. Therefore the min-cut value is again $2$ since nodes in $\mcal{R}$ receive the same linear combination as those in $\mcal{P}$ under any RLC scheme. Second, it is clear that in $\mcal{G}_{12}$, $\mcal{U}_{\mcal{R}}\cup\{\msf{u}\}$ are not clones as $\msf{u}$ has no parents in $\mcal{R}$. Hence, $\mincut_{\mcal{G}_{12}}\lp\msf{s}_1,\msf{s}_2; \mcal{U}_{\mcal{R}}\cup\{\msf{u}\}\rp=2$, again contradicting the above fact that $\mincut_{\mcal{G}_{12}}(\msf{s}_1,\msf{s}_2; \mcal{U}|_{\mcal{G}_{12}}) = 1$.

Hence, $\msf{u}\in\mcal{U}_{\mcal{S}}$ for all such $\msf{u}$. We use the same argument as above to show that the min-cut value from $\{\msf{s}_1,\msf{s}_2\}$ to $\mcal{P}\lp\mcal{U}_{\mcal{R}}\cup\{\msf{u}\}\rp$ is again $2$ in $\mcal{G}_{12}$. Then $\mincut_{\mcal{G}_{12}}(\msf{s}_1,\msf{s}_2; \mcal{U}|_{\mcal{G}_{12}}) = 1$ implies that $\mcal{U}_{\mcal{R}}\cup\{\msf{u}\}$ become clones in $\mcal{G}_{12}$. Next, we turn to look at $\mcal{G}'_{12}$. First, obviously $\mcal{U}_{\mcal{R}}\cup\{\msf{u}\}$ are not clones in $\mcal{G}'_{12}$, as $\msf{u}$ has some parents in $\mcal{S}$ which are not dropped in $\mcal{G}'_{12}$. Second, $\mincut_{\mcal{G}'_{12}}\lp\msf{s}_1,\msf{s}_2; \mcal{P}_{\mcal{G}'_{12}}(\mcal{U})\rp = 2$ as $\mcal{R}$ becomes $\msf{s}_1$-only-reachable in $\mcal{G}'_{12}$ while $\msf{s}_2$ can reach some other node in $\mcal{P}_{\mcal{G}'_{12}}(\mcal{U})$. Combining the above two, we have shown that $\mincut_{\mcal{G}'_{12}}\lp\msf{s}_1,\msf{s}_2; \mcal{U}\rp = 2$.

{\flushleft (ii) $\mcal{U}|_{\mcal{G}_{12}} \ne \mcal{U}$}:\par
Some nodes in $\mcal{U}$ are dropped in generating $\mcal{G}_{12}$ and hence $\mcal{U}\cap\mcal{K}^{\msf{s}_2}(\msf{v}_1^*) \ne \emptyset$. The nodes in this intersection will be come $\msf{s}_1$-only-reachable in $\mcal{G}'_{12}$. Since some nodes in $\mcal{U}|_{\mcal{G}'_{12}} = \mcal{U}$ can be reached by $\msf{s}_2$ in $\mcal{G}'_{12}$, we conclude that $\mincut_{\mcal{G}'_{12}}\lp\msf{s}_1,\msf{s}_2; \mcal{U}\rp = 2$.

\item [(B)] $\mcal{U}|_{\mcal{G}'_{12}} \ne \mcal{U}$: 
Some nodes in $\mcal{U}$ are dropped in generating $\mcal{G}'_{12}$, and the collection of these nodes is $\mcal{U} \setminus \mcal{U}|_{\mcal{G}'_{12}}$. In the same layer as $\mcal{P}^{\msf{s}_2}(\msf{w}_{12})$, consider the collection of predecessors of nodes in $\mcal{U} \setminus \mcal{U}|_{\mcal{G}'_{12}}$. It must be equal to $\mcal{P}^{\msf{s}_2}(\msf{w}_{12})$, otherwise nodes in $\mcal{U} \setminus \mcal{U}|_{\mcal{G}'_{12}}$ would not be dropped in generating $\mcal{G}'_{12}$. Hence, $\mcal{U} \setminus \mcal{U}|_{\mcal{G}'_{12}}$ cannot be reached by $\mcal{K}(\msf{w}_{12})$, and has no parents in $\mcal{P}$. Therefore $\mcal{U} \setminus \mcal{U}|_{\mcal{G}'_{12}} \subseteq \mcal{U}_{\mcal{Q}}$. Nodes in $\mcal{U} \setminus \mcal{U}|_{\mcal{G}'_{12}}$ hence will not be dropped in $\mcal{G}_{12}$ and $\mincut_{\mcal{G}_{12}}\lp\msf{s}_1,\msf{s}_2; \mcal{U}_{\mcal{R}}\cup\mcal{U}_{\mcal{Q}}\rp=2$, as nodes in $\mcal{U} \setminus \mcal{U}|_{\mcal{G}'_{12}}$ can only be reached by $\mcal{P}^{\msf{s}_2}(\msf{w}_{12})$ while nodes in $\mcal{U}_{\mcal{R}}$ can be reached by $\msf{w}_{12}$ and its $\msf{s}_1$-only-reachable parents. Hence, the only possibility such that $\mincut_{\mcal{G}_{12}}(\msf{s}_1,\msf{s}_2; \mcal{P}_{\mcal{G}_{12}}(\msf{v}_2^*)) = 1$ is that $\mcal{U}|_{\mcal{G}_{12}} \ne \mcal{U}$ and $\mincut_{\mcal{G}_{12}}(\mcal{U}|_{\mcal{G}_{12}}; \mcal{P}_{\mcal{G}_{12}}(\msf{v}_2^*)) = 1$. 

Note that $\mincut_{\mcal{G}'_{12}}\lp\mcal{U}|_{\mcal{G}'_{12}}; \mcal{P}_{\mcal{G}'_{12}}(\msf{v}_2^*)\rp = \mincut\lp\mcal{U}|_{\mcal{G}'_{12}}; \mcal{P}(\msf{v}_2^*)\rp$ and $\mincut_{\mcal{G}_{12}}\lp\mcal{U}|_{\mcal{G}_{12}}; \mcal{P}_{\mcal{G}_{12}}(\msf{v}_2^*)\rp = \mincut\lp\mcal{U}|_{\mcal{G}_{12}}; \mcal{P}(\msf{v}_2^*)\rp$. Also note that $\mcal{U}\setminus \mcal{U}|_{\mcal{G}_{12}}$ will not be dropped in $\mcal{G}'_{12}$, and $\mcal{U}\setminus \mcal{U}|_{\mcal{G}'_{12}}$ will not be dropped in $\mcal{G}_{12}$. Hence $\mcal{U}$ is partitioned by $\mcal{U}\setminus \mcal{U}|_{\mcal{G}'_{12}}$, $\mcal{U}\setminus \mcal{U}|_{\mcal{G}_{12}}$, and $\mcal{U}|_{\mcal{G}_{12}}\cap \mcal{U}|_{\mcal{G}'_{12}}$. Furthermore, $\mcal{U}_{\mcal{R}}\subseteq \mcal{U}|_{\mcal{G}_{12}}\cap \mcal{U}|_{\mcal{G}'_{12}}$.

We first show that $\mincut\lp\mcal{U}|_{\mcal{G}'_{12}}; \mcal{P}(\msf{v}_2^*)\rp \ge 2$. Define a function of the subsets of $\mcal{U}$ by 
\begin{align*}
f(\mcal{A}) := \mincut\lp \mcal{A}; \mcal{P}(\msf{v}_2^*)\rp,\ \mcal{A}\subseteq\mcal{U}.
\end{align*}
Since $f$ is submodular, we have
\begin{align*}
3 &\overset{\aaaa}{\le} f\lp \mcal{U}|_{\mcal{G}_{12}}\cap \mcal{U}|_{\mcal{G}'_{12}}\rp + f\lp\mcal{U}\rp \le f\lp\mcal{U}|_{\mcal{G}_{12}}\rp + f\lp\mcal{U}|_{\mcal{G}'_{12}}\rp\\& \overset{\bbbb}{=} 1+ f\lp\mcal{U}|_{\mcal{G}'_{12}}\rp 
\implies  f\lp\mcal{U}|_{\mcal{G}'_{12}}\rp \ge 2.
\end{align*}
(a) is due to $f\lp\mcal{U}\rp\ge 2$ and $f\lp \mcal{U}|_{\mcal{G}_{12}}\cap \mcal{U}|_{\mcal{G}'_{12}}\rp\ge 1$ since $\mcal{U}_{\mcal{R}}\subseteq \mcal{U}|_{\mcal{G}_{12}}\cap \mcal{U}|_{\mcal{G}'_{12}}$. (b) is due to $f\lp\mcal{U}|_{\mcal{G}_{12}}\rp = 1$. 

Next we show that $\mincut_{\mcal{G}'_{12}}\lp\msf{s}_1,\msf{s}_2;\mcal{U}|_{\mcal{G}'_{12}}\rp = 2$. This is easy to see, since $\mcal{U}\setminus \mcal{U}|_{\mcal{G}_{12}}$ will become $\msf{s}_1$-only-reachable in $\mcal{G}'_{12}$ and some other nodes in $\mcal{U}|_{\mcal{G}'_{12}}$ can be reached by $\msf{s}_2$.

Combining the above arguments, we conclude that $\mincut_{\mcal{G}'_{12}}\lp\msf{s}_1,\msf{s}_2; \mcal{P}_{\mcal{G}'_{12}}(\msf{v}_2^*)\rp = 2$.

\end{itemize}

{\flushleft 2) $\mcal{R}=\emptyset$}:\par
In this case, $\mcal{U} = \mcal{U}_{\mcal{P}}\cap\mcal{U}_{\mcal{Q}}$. For notational convenience, we denote $\mcal{P}(\mcal{U}) \setminus \mcal{P}$ by $\mcal{Q}$. Since in $\mcal{G}_{12}$ the nodes in $\mcal{P}$ no longer connects to $\mcal{U}_{\mcal{P}}$, $\mcal{P}(\mcal{U}|_{\mcal{G}_{12}}) = \mcal{Q}$. Note that if $\mcal{U}_{\mcal{P}}\ne\emptyset$, then $\mincut_{\mcal{G}'_{12}}\lp\msf{s}_1,\msf{s}_2; \mcal{P}(\mcal{U}|_{\mcal{G}'_{12}})\rp = 2$ since $\mcal{P}\subseteq\mcal{P}(\mcal{U})$, and nodes in $\mcal{P}$ become $\msf{s}_1$-only-reachable in $\mcal{G}'_{12}$ while some other nodes in $\mcal{P}(\mcal{U}|_{\mcal{G}'_{12}})$ can be reached by $\msf{s}_2$.

$\mincut_{\mcal{G}_{12}}\lp\msf{s}_1,\msf{s}_2; \mcal{P}_{\mcal{G}_{12}}(\msf{v}_2^*)\rp = 1$ implies that: (A) $\mincut_{\mcal{G}_{12}}\lp\msf{s}_1,\msf{s}_2; \mcal{Q}\rp = 1$, (B) $\mincut_{\mcal{G}_{12}}\lp\mcal{Q}; \mcal{U}|_{\mcal{G}_{12}}\rp = 1$, or (C) $\mincut_{\mcal{G}_{12}}\lp\mcal{U}|_{\mcal{G}_{12}}; \mcal{P}_{\mcal{G}_{12}}(\msf{v}_2^*)\rp = 1$. Below we discuss the three cases respectively.

\begin{itemize}
\item [(A)] $\mincut_{\mcal{G}_{12}}\lp\msf{s}_1,\msf{s}_2; \mcal{Q}\rp = 1$:
Suppose $\mcal{U}_{\mcal{P}} = \emptyset$, then $\mcal{P}(\mcal{U}) = \mcal{Q}$, and $\mincut_{\mcal{G}_{12}}\lp\msf{s}_1,\msf{s}_2; \mcal{Q}\rp = 1$ implies $\mincut\lp\msf{s}_1,\msf{s}_2; \mcal{Q}\rp = 1$ contradicting the definition of $\msf{v}_2^*$. Hence $\mcal{U}_{\mcal{P}} \ne \emptyset$, implying that $\mincut_{\mcal{G}'_{12}}\lp\msf{s}_1,\msf{s}_2; \mcal{P}(\mcal{U}|_{\mcal{G}'_{12}})\rp = 2$. 

If $\mcal{U}|_{\mcal{G}'_{12}}$ are not clones in $\mcal{G}'_{12}$ and $\mincut_{\mcal{G}'_{12}}\lp\mcal{U}|_{\mcal{G}'_{12}};\mcal{P}_{\mcal{G}'_{12}}(\msf{v}_2^*)\rp \ge 2$, then $\mincut_{\mcal{G}'_{12}}\lp\msf{s}_1,\msf{s}_2; \mcal{P}_{\mcal{G}'_{12}}(\msf{v}_2^*)\rp = 2$. 

If $\mcal{U}|_{\mcal{G}'_{12}}$ become clones in $\mcal{G}'_{12}$, then $\msf{u}_{21}':=\pmc_{\mcal{G}'_{12}}(\msf{d}_1)$ must belong to this new clone set. Its parent set is $\mcal{P} \cup \mcal{Q}|_{\mcal{G}'_{12}}$, as some nodes in $\mcal{Q}$ may be dropped in $\mcal{G}'_{12}$. $\mcal{P}$ becomes $\msf{s}_1$-only-reachable in $\mcal{G}'_{12}$, while 
$\msf{v}_1^*$ has some $\msf{s}_1$-only-reachable parents not in $\mcal{P}$. Hence, $\mcal{K}_{\mcal{G}'_{12}}\msf{v}_1^*\cap\mcal{K}^{\msf{s}_1}_{\mcal{G}'_{12}}(\msf{u}_{21}') = \emptyset$, and $\mcal{K}^{\msf{s}_1}_{\mcal{G}'_{12}}(\msf{u}_{21}')$ does not form a $(\msf{s}_1;\msf{d}_1)$-vertex-cut in $\mcal{G}'_{12}$.

If $\mincut_{\mcal{G}'_{12}}\lp\mcal{U}|_{\mcal{G}'_{12}};\mcal{P}_{\mcal{G}'_{12}}(\msf{v}_2^*)\rp = 1$, then $\mcal{U}|_{\mcal{G}'_{12}}\ne\mcal{U}$, that is, some nodes in $\mcal{U}_{\mcal{Q}}$ are dropped in $\mcal{G}'_{12}$. But no nodes in $\mcal{U}_{\mcal{P}}$ will be dropped. A node in $\mcal{U}_{\mcal{P}}$ is $\msf{s}_1$-reachable in $\mcal{G}'_{12}$, and is an predecessor of $\msf{u}_{21}'$. This node cannot lie in $\mcal{K}(\msf{v}_1^*)$, otherwise $\mcal{Q}$ contains some $\msf{s}_1$-only-reachable nodes implying that all nodes in $\mcal{Q}$ are $\msf{s}_1$-only-reachable, contradicting the fact that in $\mcal{G}'_{12}$ some nodes in $\mcal{Q}|_{\mcal{G}'_{12}}$ can be reached by $\msf{s}_2$. Hence this node is not an predecessor of any node in the cloud $\mcal{C}_1$. In $\mcal{G}'_{12}$, $\msf{u}_{21}'$ has a $\msf{s}_1$-reachable parent whose predecessors include this node in $\mcal{U}_{\mcal{P}}$, and this parent is not in the cloud $\mcal{C}_1$. Therefore, $\mcal{K}^{\msf{s}_1}_{\mcal{G}'_{12}}(\msf{u}_{21}')$ does not form a $(\msf{s}_1;\msf{d}_1)$-vertex-cut in $\mcal{G}'_{12}$.

\item [(B)] $\mincut_{\mcal{G}_{12}}\lp\mcal{Q}; \mcal{U}|_{\mcal{G}_{12}}\rp = 1$: 
In this case, $\mcal{U}|_{\mcal{G}_{12}}$ become clones in $\mcal{G}_{12}$. Suppose $\mcal{U}_{\mcal{P}} = \emptyset$. Then $\mcal{U}|_{\mcal{G}_{12}} = \mcal{U}$, and $\mcal{U}$ are clones in $\mcal{G}$, contradicting the definition of $\msf{v}_2^*$. Hence $\mcal{U}_{\mcal{P}} \ne \emptyset$, implying that $\mincut_{\mcal{G}'_{12}}\lp\msf{s}_1,\msf{s}_2; \mcal{P}(\mcal{U}|_{\mcal{G}'_{12}})\rp = 2$. Moreover, we see that $\mcal{U}_{\mcal{Q}}$ are clones in $\mcal{G}$.

Suppose $\mcal{U}|_{\mcal{G}'_{12}}\ne\mcal{U}$. We know that $\mcal{U}\setminus\mcal{U}|_{\mcal{G}'_{12}} \subseteq\mcal{U}_{\mcal{Q}}$. Since $\mcal{U}_{\mcal{Q}}$ are clones in $\mcal{G}$, we conclude that $\mcal{U}\setminus\mcal{U}|_{\mcal{G}'_{12}} = \mcal{U}_{\mcal{Q}}$, implying that all nodes in $\mcal{U}_{\mcal{Q}}$ and $\mcal{Q}$ will be dropped in $\mcal{G}'_{12}$. This contradicts the fact that some nodes in $\mcal{U}|_{\mcal{G}'_{12}}$ can be reached by $\msf{s}_2$. Therefore, $\mcal{U}|_{\mcal{G}'_{12}}=\mcal{U}$.

In $\mcal{G}'_{12}$, nodes in $\mcal{U}_{\mcal{P}}$ have parents in $\mcal{P}$. Therefore obviously $\mcal{U}|_{\mcal{G}'_{12}}=\mcal{U}$ are not clones in $\mcal{G}'_{12}$. Combining the above discussions, we conclude that $\mincut_{\mcal{G}'_{12}}\lp\msf{s}_1,\msf{s}_2; \mcal{P}_{\mcal{G}'_{12}}(\msf{v}_2^*)\rp = 2$.

\item [(C)] $\mincut_{\mcal{G}_{12}}\lp\mcal{U}|_{\mcal{G}_{12}}; \mcal{P}_{\mcal{G}_{12}}(\msf{v}_2^*)\rp = 1$:
In this case, we must have $\mcal{U}_{\mcal{G}_{12}} \ne \mcal{U}$. If $\mcal{U}|_{\mcal{G}'_{12}} = \mcal{U}$, we use the same argument in Case 1)(A)(ii) to show that $\mincut_{\mcal{G}'_{12}}\lp\msf{s}_1,\msf{s}_2; \mcal{P}_{\mcal{G}'_{12}}(\msf{v}_2^*)\rp = 2$. If If $\mcal{U}|_{\mcal{G}'_{12}} \ne \mcal{U}$, we use the same argument in Case 1)(B)(ii) to show that $\mincut_{\mcal{G}'_{12}}\lp\msf{s}_1,\msf{s}_2; \mcal{P}_{\mcal{G}'_{12}}(\msf{v}_2^*)\rp = 2$.

\end{itemize}

Proof of the claim is now complete.

\section{$(1/2,1)$-Achievability in Case $A$ when $k_1^*=k_2^*=k^*$}\label{app_Pf_NonzeroDet}

We first state a useful lemma.
\begin{lemma}\label{spreads_out}
Let $p(\alpha_1,\alpha_2,\ldots,\alpha_n), q(\alpha_1,\alpha_2,\ldots,\alpha_n)\in\mathbb{F}_2[\alpha_1,\alpha_2,\ldots,\alpha_n]$ such that $p,q$ are not identically equal to zero or to each other. If $\alpha_1,\alpha_2,\ldots, \alpha_n$ are chosen independently and uniformly over $\mathbb{F}_{2^k},$ then 
\begin{itemize}
\item $q(\alpha_1,\alpha_2,\ldots,\alpha_n)\neq 0$ with probability at least $1-O(\frac{1}{2^k}),$ so the rational function $\frac{p}{q}$ is well-defined with high probability,
\item and $P(\frac{p}{q}=\gamma)=O(\frac{1}{2^k})$ for all $\gamma\in\mathbb{F}_{2^k}.$
\end{itemize}
\end{lemma}
\begin{proof}
We use a standard result from finite fields which states that if a multivariate polynomial $g$ in $n$ variables over finite field $\mathbb{F}$ with degree in each variable at most $d,$ is evaluated at an argument chosen uniformly over the set of possible arguments, then it yields zero with probability at most $\frac{nd}{|\mathbb{F}|},$ provided of course that the polynomial is not identically zero.

This proves the first item in the lemma with $g=q$ and the second item in the lemma for the case $\gamma=0$ using $g=p.$

For $\gamma=1,$ we use the fact that $p-q$ is not identically zero to get $P(\frac{p}{q}=1)=O(\frac{1}{2^k}).$

For any other $\gamma\in\mathbb{F}_{2^k},$ we notice that $p-\gamma q$ cannot possibly be identically zero unless both $p$ and $q$ are identically zero. This is because $p,q$ have coefficients from $\mathbb{F}_2$ while $\gamma\neq 0,1.$ This establishes that $p-\gamma q$ evaluates to zero with probability atmost $O(\frac{1}{2^k}).$
\end{proof}

We start the proof of $(1/2,1)$-achievability below.

Here, we have
\begin{itemize}
\item $\mcal{P}^{\msf{s}_1}(\msf{v}_1^*)\setminus \mcal{P}^{\msf{s}_1}(\msf{v}_2^*)\ne \emptyset, \mcal{P}^{\msf{s}_2}(v_2^*)\setminus \mcal{P}^{\msf{s}_2}(\msf{v}_1^*)\ne \emptyset,$ 
\item $\msf{u}_1\in \mcal{P}^{\msf{s}_1}(\msf{v}_1^*)\setminus \mcal{P}^{\msf{s}_1}(\msf{v}_2^*)\ne \emptyset$ and $\msf{u}_2\in \mcal{P}^{\msf{s}_2}(\msf{v}_2^*)\setminus \mcal{P}^{\msf{s}_2}(\msf{v}_1^*)\ne \emptyset,$
\item $\msf{u}_1$ is $\msf{s}_1$-only-reachable and $\msf{u}_2$ is $\msf{s}_1\msf{s}_2$-reachable,
\item $\msf{w}_2\in\mcal{P}(\msf{v}_2^*)$ such that $\mincut \lp \msf{s}_1,\msf{s}_2;\msf{u}_2,\msf{w}_2\rp=2,$ and $\msf{w}_2$ is a parent of $\msf{v}_1^*,$ and $\msf{w}_2$ is $\msf{s}_2$-reachable. 
\end{itemize}

We will use RLC for the transmission of all nodes in layers $0$ through $k^*-2.$ The RLC is performed without mixing across the time steps. In the first time step, $\msf{s}_1$ transmits the symbol $a$ while $\msf{s}_2$ transmits the symbol $b_1.$ In the second time step, $\msf{s}_1$ transmits symbol $a$ while $\msf{s}_2$ transmits the symbol $b_2.$ 

Suppose now that $\msf{w}_2$ is $\msf{s}_1\msf{s}_2$-reachable. 

Consider the scheme where $\msf{w}_2$ and $\msf{u}_2$ both zero-force user 1's symbol $a.$ $\msf{u}_1$ and $\msf{u}_2$ transmit in the first time slot, thus causing no interference at $\msf{v}_1^*$ and $\msf{v}_2^*.$ $\msf{w}_2$ transmits in the second time slot.

We have $\beta^{(1)}_{\msf{u}_1,s_2}=\beta^{(1)}_{\msf{u}_1,s_2}=0,$ while $\beta^{(1)}_{\msf{u}_1,\msf{s}_1}\ne 0$ with high probability from Lemma~\ref{lem_reachability}. Thus $\msf{u}_1$ can decode $\msf{s}_1$'s symbol $a$ with high probability.

Now, the receptions of $\msf{u}_2$ in the two time slots are $\beta_{\msf{u}_2,\msf{s}_1}^{(1)}\cdot a+ \beta_{\msf{u}_2,s_2}^{(1)}\cdot b_1$ and $\beta_{\msf{u}_2,\msf{s}_1}^{(2)}\cdot a+ \beta_{\msf{u}_2,s_2}^{(2)}\cdot b_2$ respectively. Similarly, the receptions of $\msf{w}_2$ are $\beta_{\msf{w}_2,\msf{s}_1}^{(1)}\cdot a+ \beta_{\msf{w}_2,s_2}^{(1)}\cdot b_1$ and $\beta_{\msf{w}_2,\msf{s}_1}^{(2)}\cdot a+ \beta_{\msf{w}_2,s_2}^{(2)}\cdot b_2.$ Note that the coefficients of these symbols are all non-zero with high probability from Lemma~\ref{lem_reachability}.

The zero-forcing yields:
\begin{itemize}
\item Transmission of $\msf{u}_2:$ $\beta_{\msf{u}_2,s_2}^{(1)}\beta_{\msf{u}_2,\msf{s}_1}^{(2)}\cdot b_1 - \beta_{\msf{u}_2,s_2}^{(2)}\beta_{\msf{u}_2,\msf{s}_1}^{(1)}\cdot b_2$
\item Transmission of $\msf{w}_2:$ $\beta_{\msf{w}_2,s_2}^{(1)}\beta_{\msf{w}_2,\msf{s}_1}^{(2)}\cdot b_1 - \beta_{\msf{w}_2,s_2}^{(2)}\beta_{\msf{w}_2,\msf{s}_1}^{(1)}\cdot b_2$
\end{itemize}

To show that $\msf{v}_2^*$ can decode, we only need to show that the determinant $\begin{vmatrix} \beta_{\msf{u}_2,s_2}^{(1)}\beta_{\msf{u}_2,\msf{s}_1}^{(2)} &  - \beta_{\msf{u}_2,s_2}^{(2)} \beta_{\msf{u}_2,\msf{s}_1}^{(1)} \\ \beta_{\msf{w}_2,s_2}^{(1)}\beta_{\msf{w}_2,\msf{s}_1}^{(2)} &  - \beta_{\msf{w}_2,s_2}^{(2)} \beta_{\msf{w}_2,\msf{s}_1}^{(1)} \end{vmatrix} $ is non-zero, ie that $\beta_{\msf{u}_2,s_2}^{(1)}\beta_{\msf{u}_2,\msf{s}_1}^{(2)} \beta_{\msf{w}_2,s_2}^{(2)} \beta_{\msf{w}_2,\msf{s}_1}^{(1)}\ne \beta_{\msf{u}_2,s_2}^{(2)} \beta_{\msf{u}_2,\msf{s}_1}^{(1)} \beta_{\msf{w}_2,s_2}^{(1)}\beta_{\msf{w}_2,\msf{s}_1}^{(2)}$ or 
\[
\frac{\beta_{\msf{u}_2,s_2}^{(1)}\beta_{\msf{w}_2,\msf{s}_1}^{(1)}}{\beta_{\msf{u}_2,\msf{s}_1}^{(1)}\beta_{\msf{w}_2,s_2}^{(1)}}\ne \frac{\beta_{\msf{u}_2,s_2}^{(2)}\beta_{\msf{w}_2,\msf{s}_1}^{(2)}}{\beta_{\msf{u}_2,\msf{s}_1}^{(2)} \beta_{\msf{w}_2,s_2}^{(2)}}
\]

Note that the coefficients with 1 superscript are independent of the coefficients with 2 superscript. So, LHS and RHS are two independent and identically distributed random variables taking values in $\mbb{F}_{2^r}.$

By Lemma~\ref{lem_two_source}, we have that the determinant $\begin{vmatrix} \beta_{\msf{u}_2,\msf{s}_1}^{(1)} & \beta_{\msf{u}_2,s_2}^{(1)} \\ \beta_{\msf{w}_2,\msf{s}_1}^{(1)} & \beta_{\msf{w}_2,s_2}^{(1)} \end{vmatrix}\ne 0$ with high probability. So, the above random variable is not equal to 1 with high probability.

Now, we note that the random variable is a ratio of two polynomials with coefficients from $\mbb{F}_2,$ a ratio that is not identically 1. The equality stating that the ratio equals $\gamma\in\mathbb{F}_{2^r}, \gamma\ne 0,1$ is an equality stating that a polynomial not identically zero evaluates to 0. If all coefficients are chosen indpendently and uniformly at random, this polynomial evaluates to 0 with probability $O\lp \frac{1}{|\mbb{F}_{2^r}|} \rp.$ Thus, the random variable does not concentrate on any given value $\gamma\in\mathbb{F}_{2^r}$ and so, two independent and identically distributed copies of the random variable are unequal with high probability. 

Suppose that $\msf{w}_2$ is $\msf{s}_2$-only-reachable. Then, $\msf{u}_1, \msf{u}_2$ transmit in the first time slot with $\msf{u}_2$ zero-forcing user 1's symbol $a.$ In the second time slot, $\msf{w}_2$ which can recover both $\msf{b}_1$ and $\msf{b}_2$ with high probability, provides a linearly independent signal to $\msf{u}_2$'s transmission.

\section{Formal Proofs of Outer Bounds}
\subsection{Proof of the Omniscient Bound}

Since $\mcal{K}(\msf{v})$ is a $\lp\msf{s}_1,\msf{s}_2; \msf{d}_1\rp$-vertex-cut, the received signal at $\msf{d}_1$, $Y_{\msf{d}_1}$ is a function of $Y_{\msf{v}}$. On the other hand, since $\mcal{K}^{\msf{s}_2}(\msf{v})$ is a $\lp \msf{s}_2; \msf{d}_2\rp$-vertex-cut, we have that $Y_{\msf{d}_2}^N$ is a function of $X_{\msf{s}_1}^N$ and $Y_{\msf{v}}^N.$ Hence we have the Markov chains
\begin{align}
&X_{\msf{s}_1}^N \lra Y_{\msf{v}}^N \lra Y_{\msf{d}_1}^N \label{markov_1} \\
&X_{\msf{s}_2}^N \lra \lp Y_{\msf{v}}^N, X_{\msf{s}_1}^N\rp \lra Y_{\msf{d}_1}^N \label{markov_2}
\end{align}
By Fano's inequality and the data processing inequality, we have for any scheme of block length $N,$
\begin{align*}
&N\lp R_1+R_2-\epsilon_N\rp\\ 
&\le I\lp X_{\msf{s}_1}^N; Y_{\msf{d}_1}^N\rp + I\lp X_{\msf{s}_2}^N; Y_{\msf{d}_2}^N\rp\\
&\le I\lp X_{\msf{s}_1}^N; Y_{\msf{v}}^N\rp + I\lp X_{\msf{s}_2}^N; Y_{\msf{v}}^N, X_{\msf{s}_1}^N\rp {\hspace{10pt} \mbox{(from \eqref{markov_1} and \eqref{markov_2})}}\\
&\le I\lp X_{\msf{s}_1}^N; Y_{\msf{v}}^N\rp + I\lp X_{\msf{s}_2}^N; Y_{\msf{v}}^N| X_{\msf{s}_1}^N\rp\\
&= H\lp Y_{\msf{v}}^N\rp - H\lp Y_{\msf{v}}^N | X_{\msf{s}_1}^N\rp + H\lp Y_{\msf{v}}^N | X_{\msf{s}_1}^N\rp\\
&= H\lp Y_{\msf{v}}^N\rp \le N,
\end{align*}
where $\epsilon_N\ra 0$ as $N\ra\infty$.
Hence $R_1+R_2\le 1$.

\subsection{Proof of Claim~\ref{claim_2R1R2}}
\begin{proof}
If $(R_1,R_2)$ is achievable, from data processing inequality and Fano's inequality, we have
\begin{align*}
&N\lp 2R_1+R_2-\epsilon_N\rp\\ 
&\le I\lp X_{\msf{s}_1}^N; Y_{\msf{d}_1}^N\rp + I\lp X_{\msf{s}_1}^N; Y_{\msf{d}_1}^N\rp + I\lp X_{\msf{s}_2}^N; Y_{\msf{d}_2}^N\rp\\
&\overset{\aaaa}{\le} I\lp X_{\msf{s}_1}^N; Z_{21}^N, X_{\msf{s}_2}^N\rp + I\lp X_{\msf{s}_1}^N; Z_1^N\rp + I\lp X_{\msf{s}_2}^N; Z_{21}^N, Z_{22}^N\rp \\
&\overset{\bbbb}{=}  I\lp X_{\msf{s}_1}^N; Z_{21}^N| X_{\msf{s}_2}^N\rp + I\lp X_{\msf{s}_1}^N; Z_1^N\rp + I\lp X_{\msf{s}_2}^N; Z_{21}^N, Z_{22}^N\rp\\
&= H\lp Z_{21}^N| X_{\msf{s}_2}^N\rp + H\lp Z_1^N\rp - H\lp Z_1^N| X_{\msf{s}_1}^N\rp + H\lp Z_{21}^N,Z_{22}^N\rp\\&\quad - H\lp Z_{21}^N,Z_{22}^N|X_{\msf{s}_2}^N\rp\\
&\overset{\cccc}{=} H\lp Z_1^N\rp + H\lp Z_{21}^N| X_{\msf{s}_2}^N\rp - H\lp Z_{21}^N|X_{\msf{s}_2}^N\rp + H\lp Z_{21}^N,Z_{22}^N\rp\\&\quad - H\lp Z_1^N| X_{\msf{s}_1}^N\rp\\
&\overset{\dddd}{\le} H\lp Z_1^N\rp + H\lp Z_{21}^N,Z_{22}^N\rp - H\lp Z_{22}^N\rp\\
&= H\lp Z_1^N\rp + H\lp Z_{21}^N|Z_{22}^N\rp \overset{\eeee}{\le} 2N
\end{align*}
where $\epsilon_N\ra0$ as $N\ra\infty$. (a) is due to condition 2) and 3). (b) is due to the fact that $X_{\msf{s}_1}^N$ and $X_{\msf{s}_2}^N$ are independent. (c) is due to condition 5) and rearranging terms. (d) is due to condition 4). (e) is due to condition 1).
\end{proof} 
\subsection{Proof of Claim~\ref{claim_2R12R2}}
\begin{proof}
If $(R_1,R_2)$ is achievable, from data processing inequality and Fano's inequality, we have
\begin{align*}
&N\lp 2R_1+R_2-\epsilon_{1,N}\rp\\ 
&\le I\lp X_{\msf{s}_1}^N; Y_{\msf{d}_1}^N\rp + I\lp X_{\msf{s}_1}^N; Y_{\msf{d}_1}^N\rp + I\lp X_{\msf{s}_2}^N; Y_{\msf{d}_2}^N\rp\\
&\overset{\aaaa}{\le} I\lp X_{\msf{s}_1}^N; Z_{21}^N, Z_{22}^N, X_{\msf{s}_2}^N\rp + I\lp X_{\msf{s}_1}^N; Z_{11}^N\rp\\&\quad + I\lp X_{\msf{s}_2}^N; Z_{21}^N, Z_{22}^N\rp \\
&\overset{\bbbb}{=}  I\lp X_{\msf{s}_1}^N; Z_{21}^N, Z_{22}^N| X_{\msf{s}_2}^N\rp + I\lp X_{\msf{s}_1}^N; Z_{11}^N\rp\\&\quad + I\lp X_{\msf{s}_2}^N; Z_{21}^N, Z_{22}^N\rp\\
&= H\lp Z_{21}^N, Z_{22}^N| X_{\msf{s}_2}^N\rp + H\lp Z_{11}^N\rp - H\lp Z_{11}^N| X_{\msf{s}_1}^N\rp\\&\quad + H\lp Z_{21}^N,Z_{22}^N\rp - H\lp Z_{21}^N,Z_{22}^N|X_{\msf{s}_2}^N\rp\\
&\overset{\cccc}{\le} H\lp Z_{11}^N\rp - H\lp Z_{22}^N| X_{\msf{s}_1}^N\rp + H\lp Z_{22}^N\rp + H\lp Z_{21}^N| Z_{22}^N\rp \\
&\overset{\dddd}{=} H\lp Z_{11}^N\rp + H\lp Z_{21}^N|Z_{22}^N\rp + I\lp X_{\msf{s}_1}^N; Z_{22}^N\rp\\
&\overset{\eeee}{\le} 2N + I\lp X_{\msf{s}_1}^N; Z_{12}^N\rp,
\end{align*}
where $\epsilon_{1,N}\ra0$ as $N\ra\infty$. (a) is due to condition 2) and 3). (b) is due to the fact that $X_{\msf{s}_1}^N$ and $X_{\msf{s}_2}^N$ are independent. (c) is due to cancellation of terms and condition 4). (d) is due to $I\lp X_{\msf{s}_1}^N; Z_{22}^N\rp = H\lp Z_{22}^N\rp - H\lp Z_{22}^N| X_{\msf{s}_1}^N\rp$. (e) is due to condition 1) and 5).


We see that we cannot upper bound $2R_1+R_2$ by $2$ in this case. On the other hand, 
\begin{align*}
&N\lp R_2-\epsilon_{2,N}\rp \le I\lp X_{\msf{s}_2}^N; Y_{\msf{d}_2}^N\rp\\ 
&\overset{\aaaa}{\le} I\lp X_{\msf{s}_2}^N; Z_{12}^N,X_{\msf{s}_1}^N\rp \overset{\bbbb}{=} I\lp X_{\msf{s}_2}^N; Z_{12}^N| X_{\msf{s}_1}^N\rp = H\lp Z_{12}^N| X_{\msf{s}_1}^N\rp.
\end{align*}
where $\epsilon_{2,N}\ra0$ as $N\ra\infty$. (a) is due to condition (3). (b) is due to the fact that $X_{\msf{s}_1}^N$ and $X_{\msf{s}_2}^N$ are independent. 
%

Combining the above two, we have
\begin{align*}
&N\lp 2R_1+2R_2-\epsilon_{N}\rp\\
&\le 2N + I\lp X_{\msf{s}_1}^N; Z_{12}^N\rp + H\lp Z_{12}^N| X_{\msf{s}_1}^N\rp = 2N+H\lp Z_{12}^N\rp\\
&\overset{\aaaa}{\le} 3N,
\end{align*}
where $\epsilon_N = \epsilon_{1,N}+\epsilon_{2,N}\ra0$ as $N\ra\infty$. (a) is due to condition 1). Proof complete.
\end{proof}

\end{document}